\documentclass[prx,aps,groupedaddress,twocolumn,superscriptaddress,floatfix,nofootinbib]{revtex4-2}

\usepackage{mathptmx}
\usepackage[utf8]{inputenc}
\usepackage[T1]{fontenc}

\usepackage{amssymb,amsmath,amsfonts}
\usepackage{mathtools}
\usepackage{mathrsfs}
\usepackage{graphicx}
\usepackage{color}
\usepackage{bm}
\usepackage{xspace,soul}
\usepackage{siunitx}
\usepackage{dsfont}
\usepackage{braket}
\usepackage{booktabs,multirow}
\graphicspath{{figures/}}

\usepackage{comment}
\usepackage{float}
\usepackage[vmargin=1in, hmargin=0.9in]{geometry}
\usepackage{times}
\usepackage[normalem]{ulem}
\usepackage[breaklinks=true,colorlinks=true,urlcolor=blue,linkcolor=blue,citecolor=blue]{hyperref}
\usepackage{xr}
\usepackage{soul}
\usepackage[sort&compress]{natbib}

\usepackage{amsmath,amsthm}

\newtheorem{theorem}{Theorem}
\newtheorem{lemma}{Lemma}
\newtheorem{proposition}{Proposition}

\newcommand{\bLozenge}{\mathbin{\blacklozenge}}%

\newcommand{\Tr} {\operatorname{Tr}}

\makeatletter
\@ifundefined{textcolor}{}
{\definecolor{BLACK}{gray}{0}
 \definecolor{WHITE}{gray}{1}
 \definecolor{RED}{rgb}{1,0,0}
 \definecolor{GREEN}{rgb}{0,.4,0}
 \definecolor{BLUE}{rgb}{0,0,1}
 \definecolor{CYAN}{cmyk}{1,0,0,0}
 \definecolor{MAGENTA}{cmyk}{0,1,0,0}
 \definecolor{YELLOW}{cmyk}{0,.3,1,0}}
 \definecolor{light-gray}{gray}{0.90}
\makeatother

\usepackage[dvipsnames, table]{xcolor}

\begin{document}

\title{Experimental single-setting quantum state tomography} 

\author{Roman Stricker}
\affiliation{Institut f\"ur Experimentalphysik, Universit\"at Innsbruck, 6020 Innsbruck, Austria}
\author{Michael Meth}
\affiliation{Institut f\"ur Experimentalphysik, Universit\"at Innsbruck, 6020 Innsbruck, Austria}
\author{Lukas Postler}
\affiliation{Institut f\"ur Experimentalphysik, Universit\"at Innsbruck, 6020 Innsbruck, Austria}
\author{Claire Edmunds}
\affiliation{Institut f\"ur Experimentalphysik, Universit\"at Innsbruck, 6020 Innsbruck, Austria}
\author{Chris Ferrie}
\affiliation{Centre for Quantum Software and Information, University of Technology Sydney, Ultimo, NSW 2007, Australia}
\author{Rainer Blatt}
\affiliation{Institut f\"ur Experimentalphysik, Universit\"at Innsbruck, 6020 Innsbruck, Austria}
\affiliation{Institut f\"{u}r Quantenoptik und Quanteninformation, \"{O}sterreichische Akademie der Wissenschaften, 6020 Innsbruck, Austria}
\affiliation{Alpine Quantum Technologies GmbH, 6020 Innsbruck, Austria}
\author{Philipp Schindler}
\affiliation{Institut f\"ur Experimentalphysik, Universit\"at Innsbruck, 6020 Innsbruck, Austria}
\author{Thomas Monz}
\affiliation{Institut f\"ur Experimentalphysik, Universit\"at Innsbruck, 6020 Innsbruck, Austria}
\affiliation{Alpine Quantum Technologies GmbH, 6020 Innsbruck, Austria}
\author{Richard Kueng}
\affiliation{Institute for Integrated Circuits, Johannes Kepler University Linz, 4040 Linz, Austria}
\author{Martin Ringbauer}
\affiliation{Institut f\"ur Experimentalphysik, Universit\"at Innsbruck, 6020 Innsbruck, Austria}

\begin{abstract}
Quantum computers solve ever more complex tasks using steadily growing system sizes. Characterizing these quantum systems is vital, yet becoming increasingly challenging. The gold-standard is \emph{quantum state tomography} (QST), capable of fully reconstructing a quantum state without prior knowledge. Measurement and classical computing costs, however, increase exponentially in the system size --- a bottleneck given the scale of existing and near-term quantum devices. 
Here, we demonstrate a scalable and practical QST approach that uses a single measurement setting, namely symmetric informationally complete (SIC) positive operator-valued measures (POVM). We implement these nonorthogonal measurements on an ion trap device by utilizing more energy levels in each ion - without ancilla qubits. More precisely, we locally map the SIC POVM to orthogonal states embedded in a higher-dimensional system, which we read out using repeated in-sequence detections, providing full tomographic information in every shot.
Combining this \emph{SIC tomography} with the recently developed \emph{randomized measurement toolbox} (`classical shadows') proves to be a powerful combination. SIC tomography alleviates the need for choosing measurement settings at random (`derandomization'), while classical shadows enable the estimation of arbitrary polynomial functions of the density matrix orders of magnitudes faster than standard methods. The latter enables in-depth entanglement studies, which we experimentally showcase on a 5-qubit absolutely maximally entangled (AME) state. 
Moreover, the fact that the full tomography information is available in every shot enables online QST in real time. We demonstrate this on an 8-qubit entangled state, as well as for fast state identification. All in all, these features single out SIC-based classical shadow estimation as a highly scalable and convenient tool for quantum state characterization. 
\end{abstract}

\maketitle

Quantum systems are prepared in laboratories and in engineered devices such that their state delicately encodes quantum information essential for achieving goals in both science and technology. Any small adjustments, changes in the environment, or active control all change the state. Yet, an accurate mathematical description of the state is a necessary component for most higher-level tasks. A crucial requirement for ensuring the performance of quantum devices is thus having methods for accurately determining the quantum states that have been prepared. The gold-standard approach for this fundamental task is quantum state tomography (QST) or simply tomography \cite{paris2004quantum}. QST enables the full reconstruction of the system's quantum state from an exponential number of measurements. Often, however, we are not even interested in the full quantum state, but rather certain features, like entanglement across a particular bipartition. Yet, to access such non-linear functions, one would like to have the reconstructed quantum state.

Formally, QST methods use an \emph{informationally complete} set of measurements to reconstruct the complete description of the quantum state. The optimal measurement for collecting the necessary tomography data has long been known to be the so-called SIC POVMs~\cite{Renes2004SIC, Scott2006}. SIC POVMs are constructed from the minimal number of $d^2$ measurements for a $d$-dimensional system, which are arranged in a way that maximizes the pairwise distance in Hilbert space. SIC POVMs are known to exist for several low-dimensional systems~\cite{Scott2010, Fuchs2017The}, and for qubits, take the form of 4 non-orthogonal vectors arranged as a tetrahedron in the Bloch sphere, see Fig.~\ref{fig0:measurement-illustration}(d). While SIC POVMs uniquely offer access to the complete tomographic information in every single experimental run (shot), implementing these measurements in practice is very challenging, requiring purpose-built setups~\cite{Tabia2012,Bian2015}, sequential measurement~\cite{Proietti2019,Bent2015}, or ancilla assisted schemes~\cite{Chen2007,Maniscalco2021}. Hence, tomography remains almost exclusively performed using the simpler, but over-complete Pauli basis, requiring $3^N$ orthogonal measurement \emph{settings}, each with $2^N$ outcomes for an $N$-qubit system. The resulting overhead effectively limits full tomography to system sizes of only a few qubits. 

From a conceptual point of view, the qubit SIC POVM is favorable, as a single experimental shot already contains complete tomographic information. This distinct advantage has also been recognized in recent theoretical work on adaptive tomography for linear cost functions~\cite{Maniscalco2021}, as well as neural network quantum state tomography~\cite{Torlai2018,Carrasquilla2019}.
Experimentally, it is also much cheaper to repeat the same measurement setting many times than to switch settings an exponential number of times as the system size grows with Pauli tomography. So, this feature can have a significant impact in practice. Moreover, since full tomographic information is contained in every shot, the experimenter is free to stop the tomography at any point, e.g.\ when certain quantities of interest have converged. In contrast, other QST approaches would require at least one shot for each measurement setting to collect sufficient information in the first place. This discrepancy is particularly relevant when we are not interested in the full density matrix, but only in certain (non-linear) properties, which often require far fewer shots than the $3^N$ minimum in Pauli tomography~\cite{Huang2020,Paini2019}. Finally, for randomized measurement schemes~\cite{Elben2022}, where ideally a different measurement setting is required in each shot, the SIC approach obviates this requirement completely (`derandomization'), making these schemes even more practical. Hence, in most situations, SIC tomography has the potential to substantially outperform standard methods for tomography or for the direct estimation of state properties.

Here, we describe our realization of SIC POVMs in a trapped ion quantum processor and their use for characterizing unknown quantum states. We put an emphasis on demonstrating the speed and robustness obtained from reducing the number of measurement settings in conjunction with new data processing techniques that come with rigorous accuracy guarantees. With our approach, we are able to comfortably reconstruct the full 8-qubit quantum state encoded in the electronic energy levels of calcium ions in \textit{real-time} using a standard laboratory computer. Moreover, we demonstrate the simultaneous real-time estimation of Renyi entropies across all bipartitions using a sampling-free classical shadow method~\cite{Huang2020}. This enables full entanglement characterization of arbitrary (but close to pure) quantum states with orders of magnitude fewer experimental shots than standard QST methods. 

\begin{figure}[ht]
    \includegraphics[width=0.9\columnwidth]{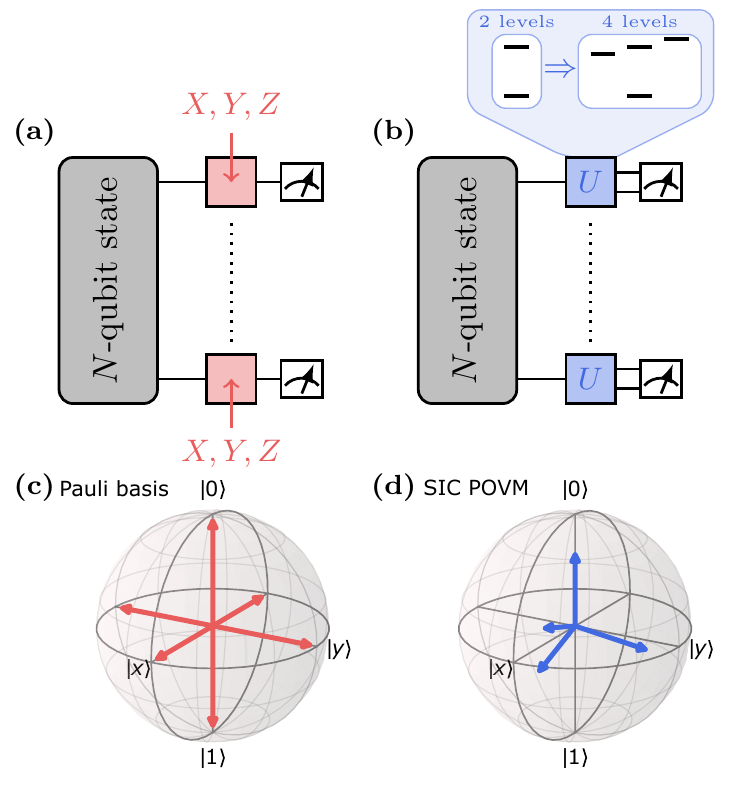}
    \caption{\textbf{Schematic illustration of Pauli and SIC tomography:} \newline
    \emph{Pauli tomography} \textbf{(a)} uses 3 basis measurements per qubit to obtain tomographic information about an unknown $N$-qubit system, see \textbf{(c)}. Each basis measurement is contingent on one of three possible unitary rotations -- red boxes in \textbf{(a)}. This produces a total of $3^N$ different measurement settings that need to be accessed.
    \newline
    \textit{SIC tomography} \textbf{(b)}, on the other hand, uses the same measurement setting for each qubit, see \textbf{(d)}. 
    This non-orthogonal measurement is achieved by isometrically embedding each 2-level system (qubit) into a larger 4-level system (ququart) -- blue boxes in \textbf{(b)} -- and subsequently measuring this larger system. The experimental realization of this embedding within each ion is shown in Fig.~\ref{fig1:TomographySchematic}.
    }
    \label{fig0:measurement-illustration}
\end{figure}

\begin{figure}[ht]
    \centering
    \includegraphics[width=0.9\columnwidth]{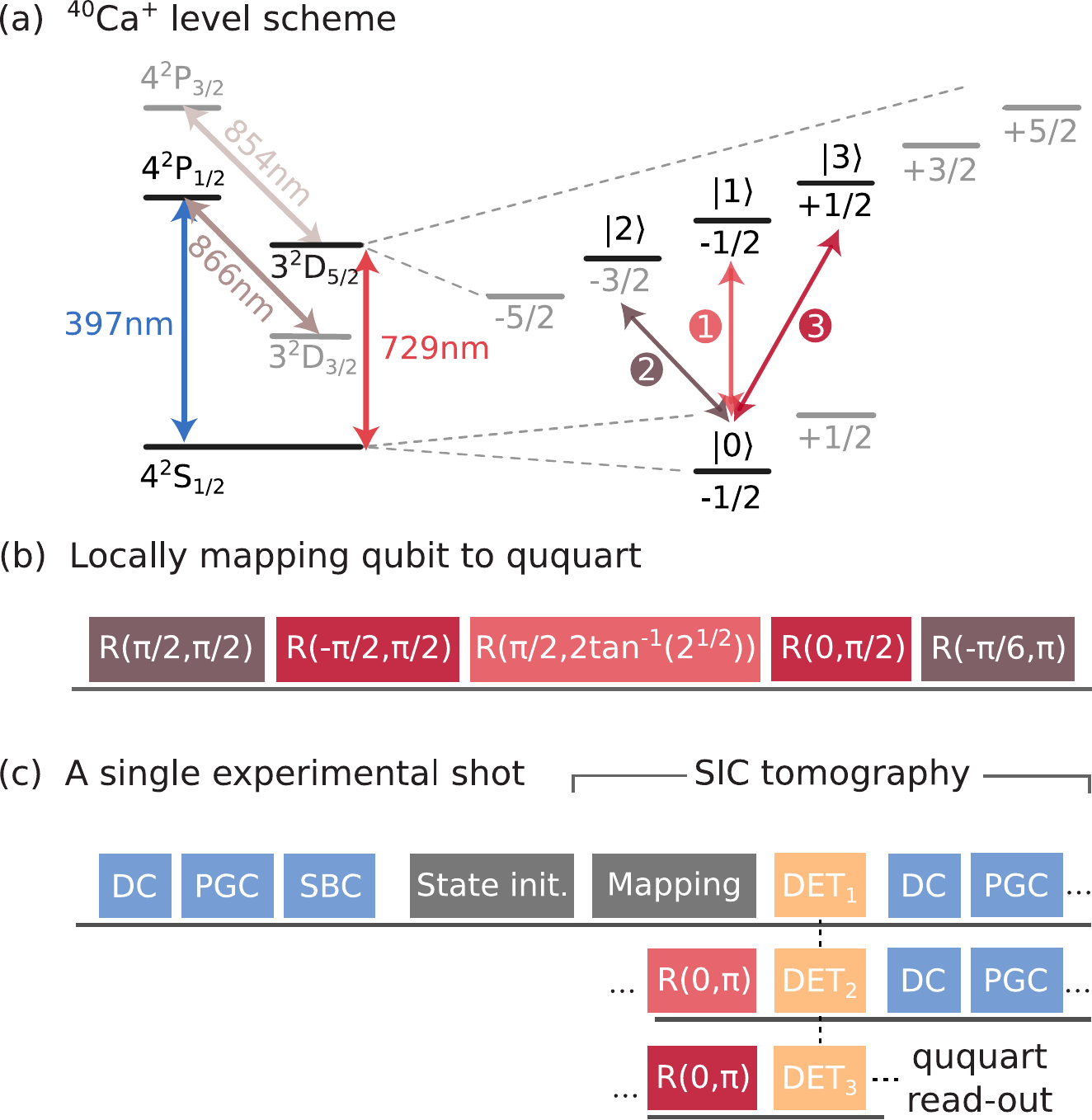}
    \caption{\textbf{Experimental implementation of SIC tomography from Fig.~\ref{fig0:measurement-illustration}(b,d)}. \textbf{(a)} Level scheme of a $^{40}$Ca$^+$ ion representing a qubit or ququart with important transitions marked: (blue) dipole transition for cooling and detection, (red) metastable quadrupole transition for encoding qubits and ququarts within the Zeeman sub-manifold and (brown) additional transitions for repumping. \textbf{(b)} Gate sequence for locally mapping the SIC POVM from Fig.~\ref{fig0:measurement-illustration}(d) on qubit-level to four orthogonal basis states of a ququart denoted in (a). This enables full read-out of the SIC POVM in a single experimental shot by means of a four-outcome projective measurement. \textbf{(c)} Experimental realization of SIC tomography comprised of cooling (DC = doppler cooling, PGC = polarization gradient cooling, SBC = sideband cooling), preparation of the state to be analyzed, mapping from qubit to ququart according to the SIC POVM, and finally the four-outcome projective measurement. For the latter, three sequential fluorescence detections (DET) are required~\cite{Ringbauer2021}, see Appendix~\ref{app:setup} for details.}
    \label{fig1:TomographySchematic}
\end{figure}

\section{SIC Tomography}
QST aims at reconstructing an unknown quantum state $\rho$ from an informationally complete set of measurements, which spans the entire Hilbert-space of the quantum system. The minimal number of measurement outcomes to reconstruct an arbitrary $d$-dimensional quantum state then is $d^2$. An experimenter performs these measurements on many copies of $\rho$ referred to as experimental shots, and attempts to reconstruct $\rho$ from the observed measurement counts. The standard approach to QST of $N$-qubit systems combines tomographic measurements of each individual qubit. The over-complete Pauli basis is a particularly prominent choice, see Fig.~\ref{fig0:measurement-illustration}\textbf{(c)}. Three distinct measurement settings are required to evenly cover the Bloch sphere and obtain tomographic single-qubit information. Extending this to $N$-qubit systems produces $3^N$ distinct measurement settings that need to be explored, see Figure~\ref{fig0:measurement-illustration}(a).

In contrast to (single-qubit) Pauli basis measurements, (single-qubit) SIC POVMs provide access to a complete set of tomographic data from a single experimental shot. No change in measurement settings is required. This desirable feature extends to $N$-qubit measurements: a single measurement setting per qubit suffices to obtain tomographically complete data, see Fig.~\ref{fig0:measurement-illustration}(b). SIC tomography utilizes (tensor products of) the single-qubit SIC POVM depicted in Fig.~\ref{fig0:measurement-illustration}(d) (or a local rotation thereof). Following Naimark's dilation theorem~\cite{Naimark1940}, every POVM can be realized as a projective measurement on a higher-dimensional Hilbert space. Using this result, together with a qudit quantum processor~\cite{Ringbauer2021}, we realize the qubit SIC POVM as a projective measurement on a 4-level system (ququart) employing two more states within each calcium ion. For this purpose, we map each qubit locally to a ququart using the unitary

\begin{equation}
\hat M = 
\begin{pmatrix}
\frac{1}{\sqrt{2}} & 0 & 0 & \frac{1}{\sqrt{2}} \\
\frac{1}{\sqrt{6}} & \frac{1}{\sqrt{3}} & \frac{1}{\sqrt{3}} & -\frac{1}{\sqrt{6}} \\
\frac{1}{\sqrt{6}} & \frac{(-1)^{2/3}}{\sqrt{3}} & \frac{(-1)^{4/3}}{\sqrt{3}} & -\frac{1}{\sqrt{6}} \\
\frac{1}{\sqrt{6}} & \frac{(-1)^{4/3}}{\sqrt{3}} & \frac{(-1)^{2/3}}{\sqrt{3}} & -\frac{1}{\sqrt{6}} \\
\end{pmatrix}
\end{equation}

Here, the four two-dimensional vectors contained in the first two columns represent the measurement vectors of the qubit basis given in Fig.~\ref{fig0:measurement-illustration}(d). The optimized gate sequence for mapping the measurement states from qubit to ququart is shown in Fig.~\ref{fig1:TomographySchematic}(b) and consists of five local rotations for each qubit, with phase gates absorbed into the rotation phases. Therefore, our single-setting SIC tomography implementation remains the very same independent of qubit number.

\section{Reconstruction methods \& classical shadows}
Even with a complete set of measurements, reconstructing $\rho$ is computationally demanding, especially
if one insist on enforcing physicality constraints. Two standard QST methods in the field are linear inversion (lin-inv) and maximum likelihood estimation (MLE) \cite{PhysRevA.55.R1561}, which can be readily applied to both Pauli and SIC data. Lin-inv provides an analytical approach to estimating $\rho$ from a complete set of projectors $\Pi_j$ that span the entire Hilbert space. Access to (approximations of) the associated probabilities $\hat{p}_j\approx\Tr(\Pi_j\rho)$ allows us to reconstruct the underlying state:
\begin{equation}
\label{eq:lin-inv}
\vert\hat{\rho}\rangle\rangle=\sum_j\hat{p}_j\cdot S_p^{-1}\vert\Pi_j\rangle\rangle
\end{equation}
with $\vert\Pi_j\rangle\rangle$ the vectorized projector, obtained by stacking the columns of $\Pi_j$. Further, $S_p^{-1}$ denotes the Moore-Penrose pseudo inverse of the measurement superoperator $S_p=\sum_j\vert\Pi_j\rangle\rangle\langle\langle\Pi_j\vert$~\cite{Wood2015}, and $\hat{p}_j=n_j/N_j$ is the observed frequency of outcome $j$ after averaging over $N_j$ experimental shots. As an unconstrained method, lin-inv version bears the risk of producing unphysical estimators for $\rho$ featuring negative eigenvalues. This is particularly pronounced when few experimental shots are used and is very problematic for estimating non-linear observables. Physical constraints are thus typically introduced through MLE, which, following Ref.~\cite{Wood2015}, can be approximated by a convex optimization problem
\begin{equation}
\label{eq:MLEcvxopt}
    \begin{split}
        \text{minimize}\hspace{1.2cm}&\vert\vert W(S\vert\hat{\rho}\rangle\rangle-\vert f\rangle)\vert\vert_2\\
        \text{subject to}\hspace{1.2cm}&\hat{\rho}\geq0,\hspace{2pt}\Tr(\hat{\rho})=1
    \end{split}
\end{equation}
Here, $S=\sum_j\vert j\rangle\langle\langle\Pi_j\vert$ denotes a change of basis operator, $\ket{f}=\sum_j\frac{n_j}{N_j}\ket{j}$ a column vector of the observed frequencies, and $W$ a diagonal matrix of statistical weights $W$. Optimization is performed under the constraints that the estimator for $\rho$ is positive semidefinite ($\rho \geq 0$) with unit trace ($\Tr (\rho)=1$), i.e.\ it must be a valid quantum state.
The convex optimization in Eq.~\eqref{eq:MLEcvxopt} is computationally more efficient than full MLE and recovers the latter in the limit of large sample sizes. Nonetheless, the computational complexity remains intractable for anything but very small systems. 
Lin-inv is much more efficient by comparison. Nonetheless, inverting the superoperator $S_p$ also becomes more challenging as system size increases. Viewed as a matrix, every tomographically complete $N$-qubit superoperator $S_p$ must have (at least) $4^N$ rows and (at least) $4^N$ columns. 
Performing the inversion row-by-row can offer some relief in terms of memory load, but the exponential number of multiplications remains challenging. Finally, physical constraints ($\rho\geq0$ and $\Tr(\rho)=1$) can be incorporated into lin-inv by truncating negative eigenvalues to obtain the closest quantum state under the Frobenius norm~\cite{Smolin2012,Sugiyama2013,Guta2020} referred to as \emph{projected least squares} (PLS). It should be noted that more principled, yet ever more computationally challenging, approaches exist \cite{Blume_Kohout_2010, Lukens_2020}.

So far, we considered full QST, i.e.\ experimentally extracting a complete description of $\rho$, which is traditionally required for predicting certain properties of complex quantum systems, especially non-linear functions, most prominently purity or entanglement. In large-scale systems, however, predicting such properties becomes very costly independent of the data acquisition (SIC, Pauli) and reconstruction (lin-inv, MLE) method, both in regards to the number of required shots and in regards to the computational power required to analyze the data. 

A promising alternative comes in the form of classical shadows~\cite{Huang2020,Paini2019} as a general-purpose method to construct classical descriptions of quantum states using very few experimental shots.
Consequently, the classical shadows framework allow for the prediction of $L$ different functions of the state with high accuracy, using order $\log (L)$ experimental shots. Importantly, the number of shots is independent of the system size and saturates information-theoretic lower bounds. Moreover, target properties can be selected after the measurements are completed. A big drawback of existing classical shadow methods, however, is that they require a different measurement to be sampled randomly for each shot~\cite{Elben2022}, which is demanding and slows down data acquisition. We show in the following that SIC POVMs naturally alleviate this sampling requirement (`derandomization'). SIC POVMs are thus an ideal choice for unlocking the full potential of the classical shadows framework. This has, in parts, been already pointed out in Ref.~\cite{Maniscalco2021}, which explores adaptive SIC tomography for linear cost functions inspired by VQE. Instead, we are here interested in a general framework for efficiently predicting general linear and nonlinear properties of the quantum state.

Formally, classical shadows provide an alternative approach for a linear-inversion estimator deduced from SIC measurements on an $N$-qubit state $\rho$. Each experimental shot $m$, containing complete tomographic information, can be be assigned to an size-$N$ string $\hat{i}_{m,1},...,\hat{i}_{m,N}\in\lbrace1,2,3,4\rbrace^{\times N}$, where each quartic value keeps track of the SIC POVM outcome observed.
For each shot $m$, an $N$-qubit estimator for the density matrix $\hat{\sigma}_m=\bigotimes_{n=1}^N \left( 3 |\psi_{\hat{i}_{m,n}} \rangle \! \langle \psi_{\hat{i}_{m,n}}| - \mathbb{I} \right)$ is obtained, referred to as a \emph{classical shadow}. A total of $M$ such estimators can be experimentally inferred and accumulated to approximate $\rho$ as
\begin{equation}
   \hat{\rho}=\frac{1}{M}\sum_{m=1}^M\bigotimes_{n=1}^N \left( 3 |\psi_{\hat{i}_{m,n}} \rangle \! \langle \psi_{\hat{i}_{m,n}}| - \mathbb{I} \right) \overset{M \to \infty}{\longrightarrow} \rho
   \label{eq:rho_classical_shadows} ,
\end{equation}
A crucial observation here is that, compared to standard linear inversion in Eq.~\eqref{eq:lin-inv}, the processing of classical shadows is performed in the dimension of the quantum state $2^N\cdot2^N$, which is the minimal possible dimension for full tomography, see Appendix~\ref{app:multi-qubit-systems}. 
Moreover, predicting linear observables using classical shadows is even more efficient as it suffices to reconstruct a subset of $\rho$ solely where operators act on. In Appendix~\ref{app:linearobservables}, we show how we can formalize these considerations to derive a measurement budget for estimating linear observables. Suppose that we are interested in estimating a total of $L \gg 1$ subsystem observables $\mathrm{tr}(O_l \rho)$, where each $O_l$ only acts non-trivially on (at most) $K \leq N$ qubits. Then,
\begin{equation}
M \geq \tfrac{8}{3} 6^K \log (2L/\delta)/\epsilon^2, \label{eq:main-text-observable-measurement-budget}
\end{equation}
measurements suffice to jointly $\epsilon$-approximate all observables with probability (at least) $1-\delta$. 
We emphasize that this is a novel, rigorous a-priori bound based on minimal assumptions. In practice, convergence sets in (much) earlier.
A full derivation and additional context is provided in the Appendix. For now, we merely point out that improvements of order $2^K$ are possible for the exponential scaling in case the observables in question have small Hilbert-Schmidt norm, as is the case for fidelities.
Apart from linear observables, classical shadows also promise to allow for efficient estimation of non-linear functions, see Appendix~\ref{app:nonlinearobservables}. Whereas the full scope of non-linear functions is covered in the Appendix, here we focus on a quadratic estimator in form of the (subsystem) purity
\begin{equation}
    \hat{p}_{(M)} = \binom{M}{2}^{-1} \sum_{m < m'} \mathrm{tr} \left( \hat{\sigma}_m \hat{\sigma}_{m'}\right) ,
\label{eq:main-text-purity-estimator} 
\end{equation}
as purity generally obeys hard convergence properties and also provides a means to measuring entanglement via Rényi-entropies (see below). The latter are given by the negative logarithm of subsystem purities, corresponding to certain bipartitions of the state. Similar to before we can derive a measurement budget, where
\begin{equation}
    M \geq 6 L3^{K}/(\epsilon^2 \delta)
\label{eq:main-text-reduced-purity-measurement-budget}
\end{equation}
measurements allow for $\epsilon$-approximating $L$ subsystem purities of size (at most) $K$ with probability (at least) $1-\delta$. We emphasize that this is again a novel, rigorous a-priori bound based on minimal assumptions, actual convergence sets in much quicker. However, this bound still marks an improvement over the best available results for purity estimation with randomized single-qubit measurements~\cite{elben2020,Neven2021}. The improvement follows from exploiting the geometric structure of SIC POVMs and we refer to Appendix~\ref{app:nonlinearobservables} for additional context and complete proofs.  
These results demonstrate that classical shadows in combination with SIC measurements offer a powerful tool set for measuring entanglement in a scalable fashion~\cite{elben2020,Neven2021}. Whereas our experimental studies will primarily focus on quadratic estimators of Rényi-entropy (Eq.~\eqref{eq:Renyi}), classical shadows can be extended to higher order estimators following the same principles: (i) rewrite a degree-$d$ polynomial as $\mathrm{tr} \left( O \rho^{\otimes d}\right)$, (ii) replace each $\rho$ with an independent classical shadow $\hat{\sigma}_m$ and (iii) average over all different sub-selection of distinct classical shadows. We refer to Refs.~\cite{Huang2020,elben2020,Neven2021} for details. Finally, we remark that classical shadow estimators from Pauli basis measurements can in principle be obtained by randomly sampling over the measurement settings from shot to shot. Although experimentally feasibly, this is highly impractical. Remarkably, the Pauli basis also leads to slower convergence than SIC measurements.

\section{Experimental setup \& SIC implementation}
Experimental results in this work are obtained on a trapped-ion quantum processor based on a linear string of $^{40}$Ca$^{+}$ ions, each encoding a single qubit in the (meta-)stable electronic states $\lbrace S_{1/2}(m=-1/2)=\ket{0}, D_{5/2}(m=-1/2)=\ket{1}\rbrace$~\cite{Schindler2013}. A universal set of quantum gate operations is realized upon coherent laser-ion interaction and comprises of arbitrary local single-qubit rotations together with two-qubit entangling operations, enabling all-to-all connectivity. A binary qubit measurement is implemented by scattering on the dipole transition, where fluorescence is only observed if the ion is in the $\ket{0}$ state, thereby separating the computational basis states $\lbrace S_{1/2}(m=-1/2)=\ket{0}, D_{5/2}(m=-1/2)=\ket{1}\rbrace$; see Fig.~\ref{fig1:TomographySchematic}(a). Equivalent control over the entire S- and D-state Zeeman manifold allows for encoding a higher dimensional quantum decimal digit (qudit) with up to 8 levels in each ion, combined with full fluorescence read-out of the whole qudit space~\cite{Ringbauer2021}, see Appendix~\ref{app:setup}. 

The present work builds on this capability by utilizing up to four levels per ion to implement SIC POVMs. To this extend, two additional levels $D_{5/2}(m=-3/2)=\ket{2}$ and $D_{5/2}(m=+1/2)=\ket{3}$ are taken into account, see Fig.~\ref{fig1:TomographySchematic}(a). Upon applying the mapping sequence depicted in Fig.~\ref{fig1:TomographySchematic}(b), each qubit is locally extended to a ququart where each basis state encodes one SIC-vector. A four-outcome projective measurement is implemented by three sequential fluorescence detections, where before the second detection the population between state $\ket{0}$ and $\ket{1}$ is flipped and likewise before the final detection the states $\ket{0}$ and $\ket{2}$ are flipped. This enables us to evaluate the full ququart state probabilities from three binary outcomes. The entire experimental sequence comprised of cooling, state preparation, mapping the SIC POVM to the ququart and four-outcome read-out is shown in Fig.~\ref{fig1:TomographySchematic}(c) to which we refer as a single experimental shot. We remark that this SIC tomography procedure works independently on each qubit and that such a single experimental shot delivers the complete tomographic information of the $N$-qubit system.

\section{Reconstruction time}
While SIC POVMs can significantly speed up data acquisition, the classical resources needed for reconstructing and storing the quantum state $\rho$ is typically an additional bottleneck in QST (see Eq.~\eqref{eq:lin-inv} and Eq.~\eqref{eq:MLEcvxopt}). In the following, we compare the computational time for reconstructing $\rho$ following various tomography approaches. For the moment, we solely focus on reconstruction time and discuss the convergence properties of the various methods later and in the Appendix~\ref{app:tomo-comparison}. For a system of $N$ qubits, we consider tomography data comprised of $M=100\cdot3^N$ shots. This corresponds to 100 shots for each measurement setting used in Pauli tomography, which, on the trapped ion platform, has proven to be a good trade-off accounting for statistics, systematic drifts, and measurement time. 

\begin{figure}[ht]
    \centering
    \includegraphics[width=\columnwidth]{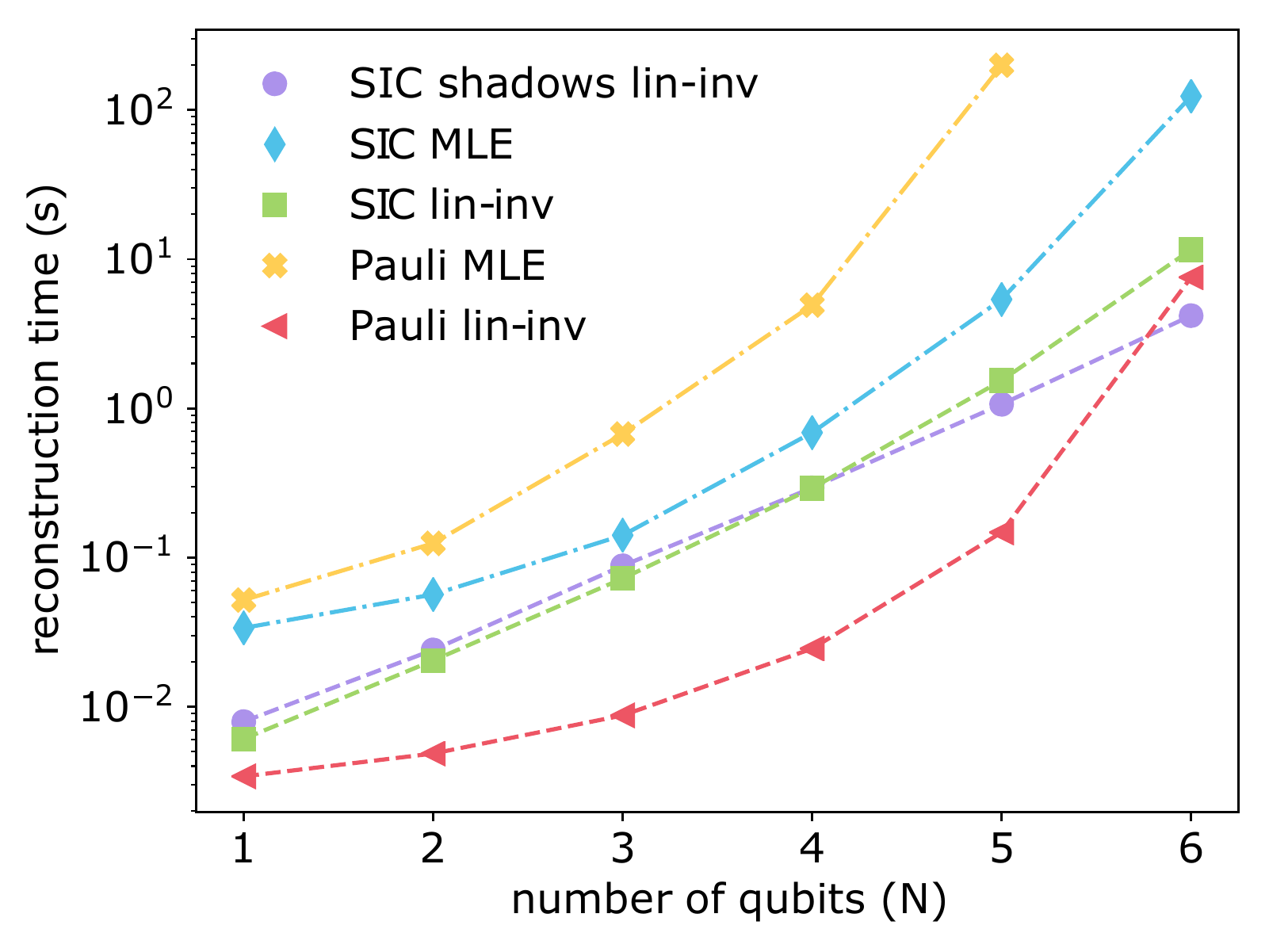}
    \caption{\textbf{Classical run time comparison for tomographic reconstruction methods.} Comparison of the computational run time for state reconstruction using lin-inv (Eq.~\eqref{eq:lin-inv}) or MLE (Eq.~\eqref{eq:MLEcvxopt}) in case of both Pauli (Fig.~\ref{fig0:measurement-illustration}(c)) and SIC tomography (Fig.~\ref{fig0:measurement-illustration}(d)), as well as SIC-based classical shadows (Eq.~\eqref{eq:rho_classical_shadows}) as a function of the system size. For each qubit number $N$, the reconstruction considers $M=100\cdot3^N$ experimental shots (i.e.\ 100 shots per Pauli basis, see text). The analysis is conducted on a standard desktop computer and plotted double logarithmically in the number of experimental shots. We find both MLE methods to require the highest computational resources, with SIC MLE significantly faster as the processed dimension is lower with $4^N\cdot 4^N$ compared to $6^N\cdot 4^N$ for Pauli tomography. Lin-inv with SIC measurements shows an initial time-offset to Pauli tomography arising from handling ququart ($\mathrm{dim}=4^N$) instead of qubit data ($\mathrm{dim}=2^N$), but grows much more slowly with system size. Among the lin-inv methods, we find the classical shadow reconstruction to be the fastest for an increasing number of qubits, as it accumulates each individual shot in dimension $2^N\cdot 2^N$ (see Eq.~\eqref{eq:rho_classical_shadows}), avoiding the costly matrix inversion in dimension $4^N\cdot 4^N$ or $6^N\cdot 4^N$ as required for SIC and Pauli tomography, respectively.} 
    \label{fig2:RunTime}
\end{figure}

\begin{figure*}[ht]
    \centering
    \includegraphics[width=\textwidth]{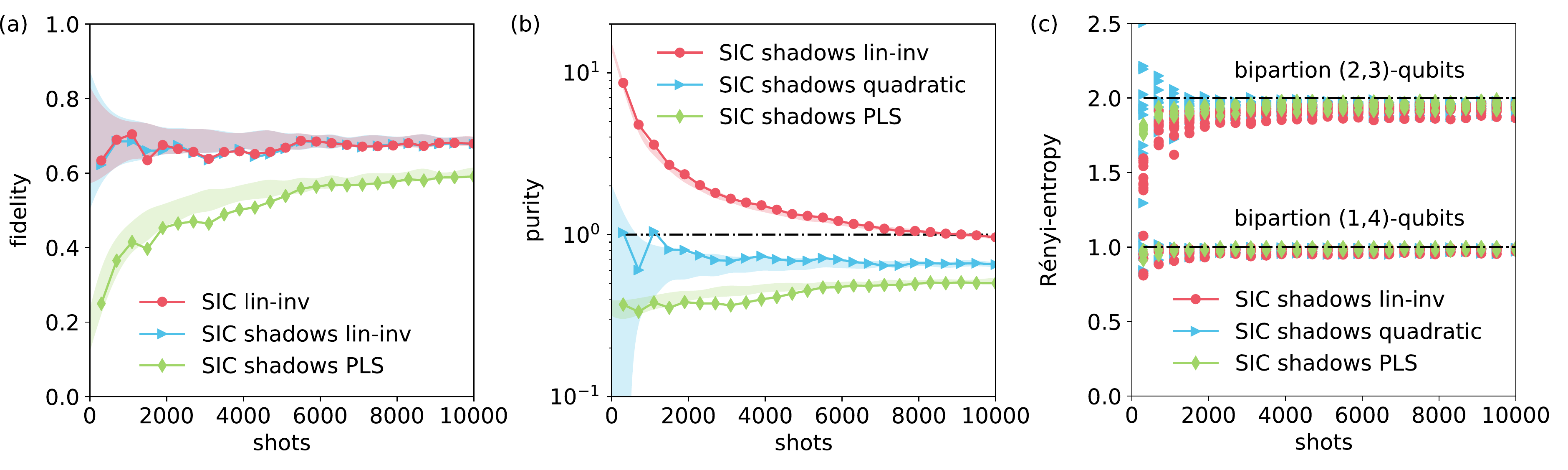}
    \caption{\textbf{SIC-based classical shadow tomography of a 5-qubit AME state.} We compare the convergence of lin-inv (Eq.~\eqref{eq:lin-inv}) and classical shadow tomography using SIC measurements (Eqs.\eqref{eq:rho_classical_shadows}-\eqref{eq:main-text-purity-estimator}). We additionally use PLS following Ref.~\cite{Smolin2012,Sugiyama2013,Guta2020}, explained in text. \textbf{(a)} In terms of convergence, we find both SIC lin-inv from Eq.~\eqref{eq:lin-inv} and SIC shadows lin-inv from Eq.~\eqref{eq:rho_classical_shadows} to overlap which is expected due to their similarity. We remark that in terms of reconstruction SIC shadows lin-inv is computationally more efficient, see Fig.~\ref{fig2:RunTime}. PLS incorporates physical constraints ($\Tr{(\rho)}=1$ and $\rho>0$), but underestimates the fidelity for small number of shots and converges very slowly requiring at least an order of magnitude more shots (indeed slower than MLE, additionally confirmed by experiments in Appendix~\ref{app:tomo-comparison} and numerical simulations in Appendix~\ref{app:simulations}). Shaded regions illustrate 1 standard deviation around the mean of multiple batches of equal shot numbers. \textbf{(b)} When estimating purity, lin-inv shows highly unphysical behaviour under insufficient statistics and not yet converged in the plot. Physical constraints can be corrected by PLS, although at the cost of much slower convergence. On the other hand, classical shadow purity estimators (Eq.~\eqref{eq:main-text-purity-estimator}) display very quick convergence. \textbf{(c)} Estimating Rényi-entropies (Eq.~\eqref{eq:Renyi}) from subset purities across all bipartitions converges very quickly for all methods. The lack of difference between the methods is likely due to the small system sizes, and for increasing bipartition sizes we expect similar behaviour as in (b). Note that bipartitions are denoted by tuples (2,3)-qubits and (1,4)-qubits referring to the number of qubits in each part.}
    \label{fig3:5qAME}
\end{figure*}

Figure~\ref{fig2:RunTime} illustrates the classical reconstruction time versus the number of qubits $N$ for experimental data used throughout this manuscript. Whereas absolute time reflects a laboratory desktop computer, relative scaling between methods remains generally valid. Note that the plot is double-logarithmic in the number of shots $M=100 \cdot 3^N$. While MLE methods always obey physical constraints, solving the convex optimization problems is costly and only feasible for small system sizes. We find MLE with SIC measurements to be more efficient, due to handling matrices of maximum size $4^N\cdot 4^N$, in contrast to $6^N\cdot4^N$ for the over-complete Pauli basis. However, MLE quickly becomes infeasible as the number of qubits increases. SIC lin-inv suffers an initial offset to Pauli lin-inv due to computing ququart instead of qubit state probabilities, which for the MLE approaches was masked in the overhead of convex optimization from Eq.~\eqref{eq:MLEcvxopt}. As the number of qubits increases, SIC makes up for this, as the computations are performed in a smaller dimension. Although computationally much cheaper than MLE, even lin-inv is becoming increasingly costly due to the memory requirements of processing the inverse superoperator $S_p$ from Eq.~\eqref{eq:lin-inv}. Already for 6 qubits this requires \SI{268.4}{MB} and \SI{3100}{MB} for SIC and Pauli measurements, respectively. Scaled up further, this will rapidly exceed the memory of today's computers. Alternatively, inversion of $S_p$ could be done row-by-row to reduce memory load, but this would be more time-consuming than pre-calculating the inverse $S_p^{-1}$ as we have done here. While linear observables under lin-inv are proven to quickly converge (see Appendix~\ref{app:linearobservables}), non-linear functions suffer from the nonphysical properties in the form of negative eigenvalues, see Fig.~\ref{fig3:5qAME}(b). Furthermore, PLS adds negligible computational overhead over lin-inv and is thus neglected in this comparison. PLS does, however affect the convergence and accuracy of the estimators, as shown and discussed further below. 

Finally, we find the best scaling for the SIC-based classical shadows from Eq.~\eqref{eq:rho_classical_shadows}, where data is processed at the dimension of the density matrix, $2^N\cdot2^N$, avoiding matrix inversion or optimization altogether. Instead, individual experimental shots are accumulated, offering convenient updates of $\rho$ for every new set of data as where further below. As a consequence of this individual accounting for every shot, the computational complexity of this method grows linearly with the total number of shots. While this linear overhead leads to slightly worse performance for very small systems, it is more than compensated by the improved exponential scaling with qubit number ($2^N\cdot2^N$ vs.\ at least $4^N\cdot4^N$) for large systems. Hence, the SIC-based classical shadows clearly outperform all other methods for 6 or more qubits. Note, that for large-scale systems the gap between classical shadows and standard lin-inv becomes even bigger as row-by-row lin-inv becomes requisite.

\section{Estimating properties of the state}

Here we shift our attention towards convergence of the different tomography estimators. In particular, we showcase the classical shadows' unique feature of efficiently predicting non-linear properties of even large-scale quantum systems. To this end, we experimentally perform tomography on a 5-qubit AME state~\cite{Enriquez2016}
\begin{equation}
\label{eq:AMEstate}
\begin{split}    
2\sqrt{2}\ket{\Omega_{5,2}} = &\ket{00000}+\ket{00011} \\+&\ket{01100}-\ket{01111}+\ket{11010}\\+&\ket{11001}+\ket{10110}-\ket{10101} .
\end{split}    
\end{equation}
AME states are the most entangled states in the sense that they are maximally entangled in all bipartitions~\cite{Helwig2013}. This makes them interesting for applications in quantum error correction~\cite{Grassl2003}, quantum teleportation, quantum secret-sharing and superdense coding~\cite{Muralidharan2008}. Alas, their general existence remains unknown for all but the smallest systems. 

\begin{figure*}[ht]
    \centering
    \includegraphics[width=\textwidth]{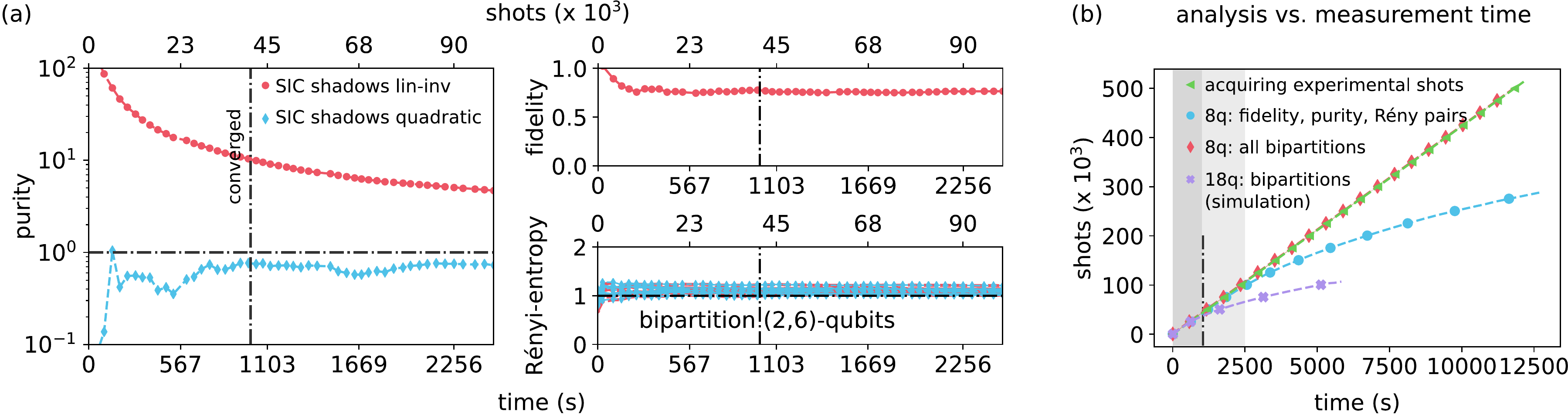}
    \caption{\textbf{Live-update SIC tomography of a maximally-entangled 8-qubit state.} We prepare a locally-rotated GHZ state and probe the state via live-update SIC tomography. The local $\pi/4$ rotation ensures that the state is not aligned with any tomography basis and can serve as a proxy for an arbitrary entangled state. Data analysis is performed in parallel to data acquisition to get the quickest possible feedback. Shot-by-shot reconstruction is realized via classical shadows (Eq.~\eqref{eq:rho_classical_shadows}-\eqref{eq:purity-estimator}). \textbf{(a)} Experimental results on purity (left), fidelity (middle top), and all (2,6)-qubit bipartition Rényi-entropies (middle bottom) obtained live in a time-regime where the data acquisition time dominates. We find the purity from classical shadows to have converged after less than \SI{1000}{s}, with the other measures converging significantly faster. Fidelity converges very fast below \SI{500}{s} and Rényi-entropies saturate almost immediately. Note, that for the latter, the curves for SIC shadows lin-inv and SIC shadows quadratic overlap due to the fast convergence for small subsystem sizes. \textbf{(b)} Comparison of the time for data acquisition ($\blacktriangleleft$) following the expected linear curve along times-curves for data analysis utilizing different methods and metrics covering the entire \SI{12 500}{s} of data taking. When analyzing fidelity, parity and all (2,6)-qubit bipartition Rényi-entropies ($\bullet$) as shown in the experimental results of (a), live update is possible for around 2500s. This is expected, since the shot-by-shot data analysis scales quadratically in total shot number, see Eq.~\eqref{eq:main-text-purity-estimator}. Estimating only Rényi-entropies for all bipartitions ($\bLozenge$), on the other hand, remains feasible for the entire time as it overlaps with data acquisition ($\blacktriangleleft$) following the linear curve. Additionally, we simulate the estimation of all Rényi-entropies of an 18-qubit state ($\boldsymbol{\times}$) by considering the data acquisition overhead ($\blacktriangleleft$), demonstrating that also this in principle remains feasible to perform live.}
    \label{fig4:Live8qGHZrotated}
\end{figure*}

We characterize the 5-qubit AME state from Eq.~\eqref{eq:AMEstate} using SIC tomography data. The results in Fig.~\ref{fig3:5qAME}(a) indicate that both lin-inv methods converge very quickly in fidelity, as is likewise expected for all linear observables, see Appendix~\ref{app:linearobservables}. In terms of convergence SIC lin-inv from Eq.~\ref{eq:lin-inv} and SIC shadows lin-inv from Eq.~\ref{eq:rho_classical_shadows} expectantly overlap due to their similarity, whereas the latter was found to be more efficient in terms of reconstruction, see Fig.~\ref{fig2:RunTime}. In contrast, incorporating physical constraints with PLS drastically slows convergence~\cite{Guta2020}, because truncation of negative eigenvalues produces a bias. We repeatedly confirm this bias by experiments in Appendix~\ref{app:tomo-comparison} as well as numerical simulations in Appendix~\ref{app:simulations}. We further estimate the state's purity, as an example of an archetypal non-linear function of the full state. Here lin-inv shows highly unphysical results under insufficient statistics which not yet converged in the plot, while the classical-shadow purity estimators from Eq.~\eqref{eq:main-text-purity-estimator} converges rapidly after only about 3000 shots, see Fig.~\ref{fig3:5qAME}(b). Finally, most inspired by applications, we probe the state's entanglement by estimating all second-order Rényi-entropies
\begin{equation}
S^{(2)}(\rho_{\mathsf{A}}) = -\log_2\Tr{(\rho_{\mathsf{A}}^2)},
\label{eq:Renyi}
\end{equation}
with the reduced density matrix $\rho_{\mathsf{A}}$ for part $\mathsf{A}$ of a bipartition $(\mathsf{A},\bar{\mathsf{A}})$ together forming $\rho$~\cite{Brydges2019}. In Fig.~\ref{fig3:5qAME}(c), we present results on Rényi-entropies following Eq.~\eqref{eq:main-text-purity-estimator} and particularly cover all bipartitions denoted by tuples (1,4)-qubits and (2,3)-qubits, referring to the number of qubits in each part. Note that, for classical shadow prediction, only the subset qubits in the smaller partition need to be taken into account. This leads to a drastic speed-up of the analysis for larger scale systems. Moreover, all predictions can be analyzed after the data has been acquired. 

For comparison, we also analysed the AME state with Pauli tomography, see Appendix Fig.~\ref{figS2:5qAMEComparison}. Slowest convergence is consistently found for PLS, which is also notably slower than regular MLE. We additionally confirm this by numerical simulations considering only shot noise, i.e.\ statistical noise, see Appendix~\ref{app:simulations}). Given only few experimental shots, SIC tomography outperforms Pauli tomography in case of MLE reconstruction, likely because the SIC POVM provides the optimal information gain per shot. Curiously, however, for large shot numbers, Pauli MLE starts to outperform SIC MLE in terms of convergence. We suspect this to be a result of the over-completeness of the Pauli basis, where very large shot numbers may lead to improved accuracy for each orthogonal direction. Note that all methods converge to the same point, as verified in numerical simulations, see Appendix Fig.~\ref{figS6:TomographyComparison}. 

In conclusion, we find SIC tomography to be preferable over Pauli tomography in regards to both, classical computation time, and convergence speed. At the same time, the underlying classical shadow formalism provides the potential for scaling to large quantum systems. We emphasize that the moderately lower quantitative performance of SIC tomography observed in our data is not inherent to the method, but due to experimental imperfections, i.e.\ the additional overhead in mapping SIC POVMs to ququarts, and the four-outcome read-out, see Appendix~\ref{app:setup}. These technical imperfections can be overcome on future devices.

\section{Live-update tomography}
Since the SIC POVM contains full tomographic information within each shot, it provides a unique way to speed up QST. Combined with classical shadows, which work by accumulating estimates in each shot according to Eq.~\eqref{eq:rho_classical_shadows}, SIC tomography can be performed in real-time (or `online'), i.e.\ a live update is performed for every new set of experimental shots. Apart from reducing time overheads by performing experiment and analysis at the same time, this approach has the big advantage that the experiment can be stopped as soon as all properties of interest are (believed to be) accurately estimated. Based on these ideas, we demonstrate a live reconstruction of a maximally-entangled 8-qubit state. Specifically, we use a GHZ-state with an additional local $\pi/4$ rotation of all qubits as a proxy for a generic maximally-entangled state that is not aligned with any of the tomography bases to provide a fair comparison.

Figure~\ref{fig4:Live8qGHZrotated}(a) presents results on the estimation of fidelity, purity and Rényi-entropies (Eq.~\eqref{eq:Renyi}), the latter for all possible (2,6)-qubit bipartitions. Purity of the full 8-qubit state is found to converge the slowest, after around \SI{1000}{s}. This is still a drastic speedup compared to lin-inv, which, after an order of magnitude longer averaging time still produces unphysical results. Fidelity converges after less than \SI{500}{s} and Rényi-entropies saturate almost immediately on the presented time-scale. We remark that both curves for SIC shadows lin-inv and SIC shadows quadratic overlap as a consequence of the fast converging 2-qubit subsets, see Appendix Fig.~\ref{figS3:Acc8qGHZrotated} on bigger bipartitions. In Fig.~\ref{fig4:Live8qGHZrotated}(a), live updates are tracked up to \SI{2500}{s}, which is the the limit beyond where the analysis starts to take longer than the data acquisition (due to the quadratic scaling in the number of shots for purity estimation, see Eq.~\eqref{eq:main-text-purity-estimator}). We extend the discussion about the time relation between data acquisition and data processing in Fig.~\ref{fig4:Live8qGHZrotated}(b), where we acquire 100 experimental shots in about \SI{2.4}{s} and show the entire \SI{12 500}{s} of data taking. Over this time, we observe that the computational time nicely follows the expected quadratic growth with shot number, relating to the number of approximated $\rho$'s to compare in Eq.~\eqref{eq:main-text-purity-estimator}. Dropping the full-system purity, we find that all bipartition Rényi-entropies can be estimated in real-time throughout the entire time of data taking. On top, simulations suggest that even on an 18-qubit state all bipartition Rényi-entropies can be estimated live for around \SI{1000}{s}.

In practice, we accumulate 100 shots for each classical shadow approximating $\rho$ following Eq.~\eqref{eq:main-text-purity-estimator} which we refer to as \emph{batching}. This post-processing trick enables a trade-off between computational time and convergence speed, which is studied thoroughly through numerical simulations in Appendix Fig.~\ref{figS4:3qClassicalShadows2D}. A suitable batch-size must be decided on a case-by-case basis. From a practical point of view, the experimental noise might also fall into consideration as it affects the targeted accuracy. Thus, batching experimental shots for the analysis comes as a handy tool for reducing analysis time with a limited effect on convergence and thereby extends the window for doing real-time analysis. 

\section{Discussion \& Outlook}
We have demonstrated that real-time SIC tomography enables the prediction of Rényi-entropies in less experimental shots than required for a minimal Pauli tomography implementation (see Fig.~\ref{fig4:Live8qGHZrotated}). Depending on the state and estimator, SIC tomography has the potential to significantly speed up the prediction of many more properties of quantum states. Beyond Rényi-entropies, we challenge SIC tomography in a state identification game partly inspired by Ref.~\cite{Huang2021}. We show that it excels over other state-of-the-art methods (see Appendix Fig.~\ref{figS5a:DecisionGame}). For the challenge in question, SIC tomography required less than 20 shots to correctly identify a randomly chosen 4-qubit linear cluster state among 16 orthogonal state possibilities, clearly outperforming any Pauli QST method. Note that, such a speed-up in favour of SIC tomography will become even more pronounced on up-scaled systems. 

Though SIC POVMs are optimal for tomography \cite{Scott2010}, they are not known to exist in every dimension. Indeed, just resolving this issue would have far-reaching consequences for the foundations of mathematics \cite{Appleby2017}. For a given Hilbert space of dimension $d=2^N$, a SIC constructed in that space is referred to as a \textit{global} SIC, whereas a measurement composed of $N$ two-dimensional SICs is referred to as a \textit{local} SIC. 
If they indeed exist, global SIC tomography would be sample-optimal in the sense that it saturates fundamental lower bounds from information theory~\cite{Guta2020,Haah2017} (joint measurements across many state copies could still yield further improvements~\cite{Wright2016,Haah2017}). Global SIC measurements would, however, be challenging to realize on quantum hardware. Quadratic circuit sizes (in the number of qubits) may be necessary, because the associated SIC states form a 2-design --- a concept closely related to chaos and information scrambling~\cite{Yoshida2017}. And recent works provide lower bounds on the minimum circuit depth required to achieve information scrambling~\cite{Brandao2021}.
We conclude that, although local SICs may be less efficient than global SICs in terms of measurement complexity~\cite{Zhu2012,Guta2020}, they are much cheaper to implement. In addition, they are informationally complete and optimal amongst all possible local measurements~\cite{Zhu2012}. 

To further improve system predictions, one might want to adapt the measurement basis depending on the state to analyse, potentially reducing quantum shot noise and offering a speed-up in convergence. The concept of adaptive tomography follows those means, where the measurement setting is adjusted based on the outcomes of prior measurements~\cite{mahler2013adaptive,huszar2012adaptive,ferrie2014self,hou2016achieving,Maniscalco2021, Kazim2021}. In the most extreme cases, the measurement settings are changed after each shot. Although this has been demonstrated in some settings \cite{PhysRevLett.117.040402, PhysRevA.101.022317, PhysRevLett.126.100402, Qi2017, PhysRevA.93.012103}, it requires both additional classical computation and physical setting adjustments, rendering it potentially very time consuming.

In contrast, the SIC POVM representation from Fig.~\ref{fig0:measurement-illustration}(d) turns out to be very efficient and practical. Moreover, we studied convergence properties of states with different overlap to either SIC POVM or Pauli basis, see Appendix~\ref{app:overlap}. Interestingly, purely local states analyzed by SIC significantly increase convergence when a component along one of the non-orthogonal SIC POVMs vanishes obeying the concept of unambiguous state discrimination~\cite{Dieks1988}. For randomly aligned, or correlated states, however, the effect vanishes making the local rotation of the SIC POVM irrelevant. However, for estimating Pauli observables, the rotation does matter, with the optimal alignment given such that the overlap with the Pauli basis is symmetric, see Appendix Fig.~\ref{figS9:SICrotated}(b). This rotated SIC will be notably useful for VQE applications, as those rely on efficient Pauli observable measurements.

Finally, we emphasize that the combination of SIC POVM measurements with the classical shadow formalism is well-suited for directly estimating higher-order polynomials of an unknown density matrix $\rho$. As discussed in Appendix~\ref{app:entanglement-probing} and following ideas from Refs.~\cite{elben2020,Neven2021}, this opens the door for mixed-state entanglement characterization of large-scale systems in real time, a likely requirement in the development of scalable quantum technology.

\vspace{0.5cm}
\noindent
\textit{\textbf{Note:} In the final stages of this project we became aware of independent and complementary research using SIC POVMs on a superconducting quantum processor~\cite{Fischer2022}}.
\vspace*{5mm}

\noindent
\textbf{Acknowledgments -- } RS, LP, MM, CE, MR, PS, TM and RB gratefully acknowledge funding by the U.S. ARO Grant No. W911NF-21-1-0007. We also acknowledge funding by the Austrian Science Fund (FWF), through the SFB BeyondC (FWF Project No. F7109), by the Austrian Research Promotion Agency (FFG) contract 872766, by the EU H2020-FETFLAG-2018-03 under Grant Agreement no. 820495, and by the Office of the Director of National Intelligence (ODNI), Intelligence Advanced Research Projects Activity (IARPA), via the U.S. ARO Grant No. W911NF-16-1-0070 and the US Air Force Office of Scientific Research (AFOSR) via IOE Grant No. FA9550-19-1-7044 LASCEM. CF was supported by the Australian Department of Industry, Innovation and Science (Grant No. AUSMURI000002). This project has received funding from the European Union’s Horizon 2020 research and innovation programme under the Marie Sk{\l}odowska-Curie grant agreement No 840450. It reflects only the author's view, the EU Agency is not responsible for any use that may be made of the information it contains. TM and RB acknowledge support by the IQI GmbH. \\
\textbf{Appendix -- } is available for this paper.\\
\textbf{Author contributions -- } MR conceived the project. CF and RK derived the theory results. RS, LP, MM, CE, PS, TM and MR performed the experiments. RS analyzed the data. RK (theory), MR and RB (experiment) supervised the project. All authors contributed to writing the manuscript.\\
\textbf{Competing interests -- } The authors declare no competing interests.\\
\textbf{Materials \& Correspondence -- } Requests for materials and correspondence should be addressed to RS (email: roman.stricker@uibk.ac.at).

\bibliographystyle{apsrev4-1new}
\bibliography{bibliography}

\begin{thebibliography}{66}%
\makeatletter
\providecommand \@ifxundefined [1]{%
 \@ifx{#1\undefined}
}%
\providecommand \@ifnum [1]{%
 \ifnum #1\expandafter \@firstoftwo
 \else \expandafter \@secondoftwo
 \fi
}%
\providecommand \@ifx [1]{%
 \ifx #1\expandafter \@firstoftwo
 \else \expandafter \@secondoftwo
 \fi
}%
\providecommand \natexlab [1]{#1}%
\providecommand \enquote  [1]{``#1''}%
\providecommand \bibnamefont  [1]{#1}%
\providecommand \bibfnamefont [1]{#1}%
\providecommand \citenamefont [1]{#1}%
\providecommand \href@noop [0]{\@secondoftwo}%
\providecommand \href [0]{\begingroup \@sanitize@url \@href}%
\providecommand \@href[1]{\@@startlink{#1}\@@href}%
\providecommand \@@href[1]{\endgroup#1\@@endlink}%
\providecommand \@sanitize@url [0]{\catcode `\\12\catcode `\$12\catcode
  `\&12\catcode `\#12\catcode `\^12\catcode `\_12\catcode `\%12\relax}%
\providecommand \@@startlink[1]{}%
\providecommand \@@endlink[0]{}%
\providecommand \url  [0]{\begingroup\@sanitize@url \@url }%
\providecommand \@url [1]{\endgroup\@href {#1}{\urlprefix }}%
\providecommand \urlprefix  [0]{URL }%
\providecommand \Eprint [0]{\href }%
\providecommand \doibase [0]{http://dx.doi.org/}%
\providecommand \selectlanguage [0]{\@gobble}%
\providecommand \bibinfo  [0]{\@secondoftwo}%
\providecommand \bibfield  [0]{\@secondoftwo}%
\providecommand \translation [1]{[#1]}%
\providecommand \BibitemOpen [0]{}%
\providecommand \bibitemStop [0]{}%
\providecommand \bibitemNoStop [0]{.\EOS\space}%
\providecommand \EOS [0]{\spacefactor3000\relax}%
\providecommand \BibitemShut  [1]{\csname bibitem#1\endcsname}%
\let\auto@bib@innerbib\@empty
\bibitem [{\citenamefont {Paris}\ and\ \citenamefont
  {Řeháček}(2004)}]{paris2004quantum}%
  \BibitemOpen
  \bibfield  {author} {\bibinfo {author} {\bibfnamefont {M.}~\bibnamefont
  {Paris}}\ and\ \bibinfo {author} {\bibfnamefont {J.}~\bibnamefont
  {Řeháček}},\ }\href {http://dx.doi.org/https://doi.org/10.1007/b98673}
  {\emph {\bibinfo {title} {Quantum State Estimation}}},\ edited by\ \bibinfo
  {editor} {\bibfnamefont {M.}~\bibnamefont {Paris}}\ and\ \bibinfo {editor}
  {\bibfnamefont {J.}~\bibnamefont {Řeháček}},\ \bibinfo {series} {Lecture
  Notes in Physics}, Vol.\ \bibinfo {volume} {649}\ (\bibinfo  {publisher}
  {Springer Berlin, Heidelberg},\ \bibinfo {year} {2004})\BibitemShut {NoStop}%
\bibitem [{\citenamefont {Renes}\ \emph {et~al.}(2004)\citenamefont {Renes},
  \citenamefont {Blume-Kohout}, \citenamefont {Scott},\ and\ \citenamefont
  {Caves}}]{Renes2004SIC}%
  \BibitemOpen
  \bibfield  {author} {\bibinfo {author} {\bibfnamefont {J.~M.}\ \bibnamefont
  {Renes}}, \bibinfo {author} {\bibfnamefont {R.}~\bibnamefont {Blume-Kohout}},
  \bibinfo {author} {\bibfnamefont {A.~J.}\ \bibnamefont {Scott}}, \ and\
  \bibinfo {author} {\bibfnamefont {C.~M.}\ \bibnamefont {Caves}},\ }\bibfield
  {title} {\bibinfo {title} {\emph {{Symmetric informationally complete quantum
  measurements}}},\ }\href {http://dx.doi.org/10.1063/1.1737053} {\bibfield
  {journal} {\bibinfo  {journal} {J. Math. Phys.}\ }\textbf {\bibinfo {volume}
  {45}},\ \bibinfo {pages} {2171} (\bibinfo {year} {2004})}\BibitemShut
  {NoStop}%
\bibitem [{\citenamefont {Scott}(2006)}]{Scott2006}%
  \BibitemOpen
  \bibfield  {author} {\bibinfo {author} {\bibfnamefont {A.~J.}\ \bibnamefont
  {Scott}},\ }\bibfield  {title} {\bibinfo {title} {\emph {Tight
  informationally complete quantum measurements}},\ }\href
  {http://dx.doi.org/10.1088/0305-4470/39/43/009} {\bibfield  {journal}
  {\bibinfo  {journal} {J. Phys. A}\ }\textbf {\bibinfo {volume} {39}},\
  \bibinfo {pages} {13507} (\bibinfo {year} {2006})}\BibitemShut {NoStop}%
\bibitem [{\citenamefont {Scott}\ and\ \citenamefont
  {Grassl}(2010)}]{Scott2010}%
  \BibitemOpen
  \bibfield  {author} {\bibinfo {author} {\bibfnamefont {A.~J.}\ \bibnamefont
  {Scott}}\ and\ \bibinfo {author} {\bibfnamefont {M.}~\bibnamefont {Grassl}},\
  }\bibfield  {title} {\bibinfo {title} {\emph {{Symmetric informationally
  complete positive-operator-valued measures: A new computer study}}},\ }\href
  {http://dx.doi.org/10.1063/1.3374022} {\bibfield  {journal} {\bibinfo
  {journal} {J. Math. Phys.}\ }\textbf {\bibinfo {volume} {51}},\ \bibinfo
  {pages} {042203} (\bibinfo {year} {2010})}\BibitemShut {NoStop}%
\bibitem [{\citenamefont {Fuchs}\ \emph {et~al.}(2017)\citenamefont {Fuchs},
  \citenamefont {Hoang},\ and\ \citenamefont {Stacey}}]{Fuchs2017The}%
  \BibitemOpen
  \bibfield  {author} {\bibinfo {author} {\bibfnamefont {C.~A.}\ \bibnamefont
  {Fuchs}}, \bibinfo {author} {\bibfnamefont {M.~C.}\ \bibnamefont {Hoang}}, \
  and\ \bibinfo {author} {\bibfnamefont {B.~C.}\ \bibnamefont {Stacey}},\
  }\bibfield  {title} {\bibinfo {title} {\emph {The SIC Question: History and
  State of Play}},\ }\href {http://dx.doi.org/10.3390/axioms6030021} {\bibfield
   {journal} {\bibinfo  {journal} {Axioms}\ }\textbf {\bibinfo {volume} {6}}
  (\bibinfo {year} {2017}),\ 10.3390/axioms6030021}\BibitemShut {NoStop}%
\bibitem [{\citenamefont {Tabia}(2012)}]{Tabia2012}%
  \BibitemOpen
  \bibfield  {author} {\bibinfo {author} {\bibfnamefont {G.~N.~M.}\
  \bibnamefont {Tabia}},\ }\bibfield  {title} {\bibinfo {title} {\emph
  {{Experimental scheme for qubit and qutrit symmetric informationally complete
  positive operator-valued measurements using multiport devices}}},\ }\href
  {http://dx.doi.org/10.1103/PhysRevA.86.062107} {\bibfield  {journal}
  {\bibinfo  {journal} {Phys. Rev. A}\ }\textbf {\bibinfo {volume} {86}},\
  \bibinfo {pages} {062107} (\bibinfo {year} {2012})}\BibitemShut {NoStop}%
\bibitem [{\citenamefont {Bian}\ \emph {et~al.}(2015)\citenamefont {Bian},
  \citenamefont {Li}, \citenamefont {Qin}, \citenamefont {Zhan}, \citenamefont
  {Zhang}, \citenamefont {Sanders},\ and\ \citenamefont {Xue}}]{Bian2015}%
  \BibitemOpen
  \bibfield  {author} {\bibinfo {author} {\bibfnamefont {Z.}~\bibnamefont
  {Bian}}, \bibinfo {author} {\bibfnamefont {J.}~\bibnamefont {Li}}, \bibinfo
  {author} {\bibfnamefont {H.}~\bibnamefont {Qin}}, \bibinfo {author}
  {\bibfnamefont {X.}~\bibnamefont {Zhan}}, \bibinfo {author} {\bibfnamefont
  {R.}~\bibnamefont {Zhang}}, \bibinfo {author} {\bibfnamefont {B.~C.}\
  \bibnamefont {Sanders}}, \ and\ \bibinfo {author} {\bibfnamefont
  {P.}~\bibnamefont {Xue}},\ }\bibfield  {title} {\bibinfo {title} {\emph
  {{Realization of Single-Qubit Positive-Operator-Valued Measurement via a
  One-Dimensional Photonic Quantum Walk}}},\ }\href
  {http://dx.doi.org/10.1103/PhysRevLett.114.203602} {\bibfield  {journal}
  {\bibinfo  {journal} {Phys. Rev. Lett.}\ }\textbf {\bibinfo {volume} {114}},\
  \bibinfo {pages} {203602} (\bibinfo {year} {2015})}\BibitemShut {NoStop}%
\bibitem [{\citenamefont {Proietti}\ \emph {et~al.}(2019)\citenamefont
  {Proietti}, \citenamefont {Ringbauer}, \citenamefont {Graffitti},
  \citenamefont {Barrow}, \citenamefont {Pickston}, \citenamefont {Kundys},
  \citenamefont {Cavalcanti}, \citenamefont {Aolita}, \citenamefont {Chaves},\
  and\ \citenamefont {Fedrizzi}}]{Proietti2019}%
  \BibitemOpen
  \bibfield  {author} {\bibinfo {author} {\bibfnamefont {M.}~\bibnamefont
  {Proietti}}, \bibinfo {author} {\bibfnamefont {M.}~\bibnamefont {Ringbauer}},
  \bibinfo {author} {\bibfnamefont {F.}~\bibnamefont {Graffitti}}, \bibinfo
  {author} {\bibfnamefont {P.}~\bibnamefont {Barrow}}, \bibinfo {author}
  {\bibfnamefont {A.}~\bibnamefont {Pickston}}, \bibinfo {author}
  {\bibfnamefont {D.}~\bibnamefont {Kundys}}, \bibinfo {author} {\bibfnamefont
  {D.}~\bibnamefont {Cavalcanti}}, \bibinfo {author} {\bibfnamefont
  {L.}~\bibnamefont {Aolita}}, \bibinfo {author} {\bibfnamefont
  {R.}~\bibnamefont {Chaves}}, \ and\ \bibinfo {author} {\bibfnamefont
  {A.}~\bibnamefont {Fedrizzi}},\ }\bibfield  {title} {\bibinfo {title} {\emph
  {{Enhanced Multiqubit Phase Estimation in Noisy Environments by Local
  Encoding}}},\ }\href {http://dx.doi.org/10.1103/PhysRevLett.123.180503}
  {\bibfield  {journal} {\bibinfo  {journal} {Phys. Rev. Lett.}\ }\textbf
  {\bibinfo {volume} {123}},\ \bibinfo {pages} {180503} (\bibinfo {year}
  {2019})}\BibitemShut {NoStop}%
\bibitem [{\citenamefont {Bent}\ \emph {et~al.}(2015)\citenamefont {Bent},
  \citenamefont {Qassim}, \citenamefont {Tahir}, \citenamefont {Sych},
  \citenamefont {Leuchs}, \citenamefont {S{\'{a}}nchez-Soto}, \citenamefont
  {Karimi},\ and\ \citenamefont {Boyd}}]{Bent2015}%
  \BibitemOpen
  \bibfield  {author} {\bibinfo {author} {\bibfnamefont {N.}~\bibnamefont
  {Bent}}, \bibinfo {author} {\bibfnamefont {H.}~\bibnamefont {Qassim}},
  \bibinfo {author} {\bibfnamefont {A.~A.}\ \bibnamefont {Tahir}}, \bibinfo
  {author} {\bibfnamefont {D.}~\bibnamefont {Sych}}, \bibinfo {author}
  {\bibfnamefont {G.}~\bibnamefont {Leuchs}}, \bibinfo {author} {\bibfnamefont
  {L.~L.}\ \bibnamefont {S{\'{a}}nchez-Soto}}, \bibinfo {author} {\bibfnamefont
  {E.}~\bibnamefont {Karimi}}, \ and\ \bibinfo {author} {\bibfnamefont {R.~W.}\
  \bibnamefont {Boyd}},\ }\bibfield  {title} {\bibinfo {title} {\emph
  {{Experimental Realization of Quantum Tomography of Photonic Qudits via
  Symmetric Informationally Complete Positive Operator-Valued Measures}}},\
  }\href {http://dx.doi.org/10.1103/PhysRevX.5.041006} {\bibfield  {journal}
  {\bibinfo  {journal} {Phys. Rev. X}\ }\textbf {\bibinfo {volume} {5}},\
  \bibinfo {pages} {041006} (\bibinfo {year} {2015})}\BibitemShut {NoStop}%
\bibitem [{\citenamefont {Chen}\ \emph {et~al.}(2007)\citenamefont {Chen},
  \citenamefont {Bergou}, \citenamefont {Zhu},\ and\ \citenamefont
  {Guo}}]{Chen2007}%
  \BibitemOpen
  \bibfield  {author} {\bibinfo {author} {\bibfnamefont {P.-X.}\ \bibnamefont
  {Chen}}, \bibinfo {author} {\bibfnamefont {J.~A.}\ \bibnamefont {Bergou}},
  \bibinfo {author} {\bibfnamefont {S.-Y.}\ \bibnamefont {Zhu}}, \ and\
  \bibinfo {author} {\bibfnamefont {G.-C.}\ \bibnamefont {Guo}},\ }\bibfield
  {title} {\bibinfo {title} {\emph {{Ancilla dimensions needed to carry out
  positive-operator-valued measurement}}},\ }\href
  {http://dx.doi.org/10.1103/PhysRevA.76.060303} {\bibfield  {journal}
  {\bibinfo  {journal} {Phys. Rev. A}\ }\textbf {\bibinfo {volume} {76}},\
  \bibinfo {pages} {060303} (\bibinfo {year} {2007})}\BibitemShut {NoStop}%
\bibitem [{\citenamefont {Garc\'{\i}a-P\'erez}\ \emph
  {et~al.}(2021)\citenamefont {Garc\'{\i}a-P\'erez}, \citenamefont {Rossi},
  \citenamefont {Sokolov}, \citenamefont {Tacchino}, \citenamefont
  {Barkoutsos}, \citenamefont {Mazzola}, \citenamefont {Tavernelli},\ and\
  \citenamefont {Maniscalco}}]{Maniscalco2021}%
  \BibitemOpen
  \bibfield  {author} {\bibinfo {author} {\bibfnamefont {G.}~\bibnamefont
  {Garc\'{\i}a-P\'erez}}, \bibinfo {author} {\bibfnamefont {M.~A.}\
  \bibnamefont {Rossi}}, \bibinfo {author} {\bibfnamefont {B.}~\bibnamefont
  {Sokolov}}, \bibinfo {author} {\bibfnamefont {F.}~\bibnamefont {Tacchino}},
  \bibinfo {author} {\bibfnamefont {P.~K.}\ \bibnamefont {Barkoutsos}},
  \bibinfo {author} {\bibfnamefont {G.}~\bibnamefont {Mazzola}}, \bibinfo
  {author} {\bibfnamefont {I.}~\bibnamefont {Tavernelli}}, \ and\ \bibinfo
  {author} {\bibfnamefont {S.}~\bibnamefont {Maniscalco}},\ }\bibfield  {title}
  {\bibinfo {title} {\emph {Learning to Measure: Adaptive Informationally
  Complete Generalized Measurements for Quantum Algorithms}},\ }\href
  {http://dx.doi.org/10.1103/PRXQuantum.2.040342} {\bibfield  {journal}
  {\bibinfo  {journal} {PRX Quantum}\ }\textbf {\bibinfo {volume} {2}},\
  \bibinfo {pages} {040342} (\bibinfo {year} {2021})}\BibitemShut {NoStop}%
\bibitem [{\citenamefont {Torlai}\ \emph {et~al.}(2018)\citenamefont {Torlai},
  \citenamefont {Mazzola}, \citenamefont {Carrasquilla}, \citenamefont
  {Troyer}, \citenamefont {Melko},\ and\ \citenamefont {Carleo}}]{Torlai2018}%
  \BibitemOpen
  \bibfield  {author} {\bibinfo {author} {\bibfnamefont {G.}~\bibnamefont
  {Torlai}}, \bibinfo {author} {\bibfnamefont {G.}~\bibnamefont {Mazzola}},
  \bibinfo {author} {\bibfnamefont {J.}~\bibnamefont {Carrasquilla}}, \bibinfo
  {author} {\bibfnamefont {M.}~\bibnamefont {Troyer}}, \bibinfo {author}
  {\bibfnamefont {R.}~\bibnamefont {Melko}}, \ and\ \bibinfo {author}
  {\bibfnamefont {G.}~\bibnamefont {Carleo}},\ }\bibfield  {title} {\bibinfo
  {title} {\emph {Neural-network quantum state tomography}},\ }\href
  {http://dx.doi.org/10.1038/s41567-018-0048-5} {\bibfield  {journal} {\bibinfo
   {journal} {Nature Physics}\ }\textbf {\bibinfo {volume} {14}},\ \bibinfo
  {pages} {447} (\bibinfo {year} {2018})}\BibitemShut {NoStop}%
\bibitem [{\citenamefont {Carrasquilla}\ \emph {et~al.}(2019)\citenamefont
  {Carrasquilla}, \citenamefont {Torlai}, \citenamefont {Melko},\ and\
  \citenamefont {Aolita}}]{Carrasquilla2019}%
  \BibitemOpen
  \bibfield  {author} {\bibinfo {author} {\bibfnamefont {J.}~\bibnamefont
  {Carrasquilla}}, \bibinfo {author} {\bibfnamefont {G.}~\bibnamefont
  {Torlai}}, \bibinfo {author} {\bibfnamefont {R.~G.}\ \bibnamefont {Melko}}, \
  and\ \bibinfo {author} {\bibfnamefont {L.}~\bibnamefont {Aolita}},\
  }\bibfield  {title} {\bibinfo {title} {\emph {Reconstructing quantum states
  with generative models}},\ }\href
  {http://dx.doi.org/10.1038/s42256-019-0028-1} {\bibfield  {journal} {\bibinfo
   {journal} {Nature Machine Intelligence}\ }\textbf {\bibinfo {volume} {1}},\
  \bibinfo {pages} {155} (\bibinfo {year} {2019})}\BibitemShut {NoStop}%
\bibitem [{\citenamefont {Huang}\ \emph {et~al.}(2020)\citenamefont {Huang},
  \citenamefont {Kueng},\ and\ \citenamefont {Preskill}}]{Huang2020}%
  \BibitemOpen
  \bibfield  {author} {\bibinfo {author} {\bibfnamefont {H.-Y.}\ \bibnamefont
  {Huang}}, \bibinfo {author} {\bibfnamefont {R.}~\bibnamefont {Kueng}}, \ and\
  \bibinfo {author} {\bibfnamefont {J.}~\bibnamefont {Preskill}},\ }\bibfield
  {title} {\bibinfo {title} {\emph {Predicting many properties of a quantum
  system from very few measurements}},\ }\href
  {http://dx.doi.org/10.1038/s41567-020-0932-7} {\bibfield  {journal} {\bibinfo
   {journal} {Nature Physics}\ }\textbf {\bibinfo {volume} {16}},\ \bibinfo
  {pages} {1050} (\bibinfo {year} {2020})}\BibitemShut {NoStop}%
\bibitem [{\citenamefont {Paini}\ and\ \citenamefont
  {Kalev}(2019)}]{Paini2019}%
  \BibitemOpen
  \bibfield  {author} {\bibinfo {author} {\bibfnamefont {M.}~\bibnamefont
  {Paini}}\ and\ \bibinfo {author} {\bibfnamefont {A.}~\bibnamefont {Kalev}},\
  }\bibfield  {title} {\bibinfo {title} {\emph {An approximate description of
  quantum states}},\ }\href@noop {} {\bibfield  {journal} {\bibinfo  {journal}
  {Preprint at arXiv:1910.10543}\ } (\bibinfo {year} {2019})}\BibitemShut
  {NoStop}%
\bibitem [{\citenamefont {Elben}\ \emph {et~al.}(2022)\citenamefont {Elben},
  \citenamefont {Flammia}, \citenamefont {Huang}, \citenamefont {Kueng},
  \citenamefont {Preskill}, \citenamefont {Vermersch},\ and\ \citenamefont
  {Zoller}}]{Elben2022}%
  \BibitemOpen
  \bibfield  {author} {\bibinfo {author} {\bibfnamefont {A.}~\bibnamefont
  {Elben}}, \bibinfo {author} {\bibfnamefont {S.~T.}\ \bibnamefont {Flammia}},
  \bibinfo {author} {\bibfnamefont {H.-Y.}\ \bibnamefont {Huang}}, \bibinfo
  {author} {\bibfnamefont {R.}~\bibnamefont {Kueng}}, \bibinfo {author}
  {\bibfnamefont {J.}~\bibnamefont {Preskill}}, \bibinfo {author}
  {\bibfnamefont {B.}~\bibnamefont {Vermersch}}, \ and\ \bibinfo {author}
  {\bibfnamefont {P.}~\bibnamefont {Zoller}},\ }\bibfield  {title} {\bibinfo
  {title} {\emph {The randomized measurement toolbox}},\ }\href@noop {}
  {\bibfield  {journal} {\bibinfo  {journal} {arXiv preprint arXiv:2203.11374}\
  } (\bibinfo {year} {2022})}\BibitemShut {NoStop}%
\bibitem [{\citenamefont {Ringbauer}\ \emph {et~al.}(2021)\citenamefont
  {Ringbauer}, \citenamefont {Meth}, \citenamefont {Postler}, \citenamefont
  {Stricker}, \citenamefont {Blatt}, \citenamefont {Schindler},\ and\
  \citenamefont {Monz}}]{Ringbauer2021}%
  \BibitemOpen
  \bibfield  {author} {\bibinfo {author} {\bibfnamefont {M.}~\bibnamefont
  {Ringbauer}}, \bibinfo {author} {\bibfnamefont {M.}~\bibnamefont {Meth}},
  \bibinfo {author} {\bibfnamefont {L.}~\bibnamefont {Postler}}, \bibinfo
  {author} {\bibfnamefont {R.}~\bibnamefont {Stricker}}, \bibinfo {author}
  {\bibfnamefont {R.}~\bibnamefont {Blatt}}, \bibinfo {author} {\bibfnamefont
  {P.}~\bibnamefont {Schindler}}, \ and\ \bibinfo {author} {\bibfnamefont
  {T.}~\bibnamefont {Monz}},\ }\bibfield  {title} {\bibinfo {title} {\emph {A
  universal qudit quantum processor with trapped ions}},\ }\href@noop {}
  {\bibfield  {journal} {\bibinfo  {journal} {Preprint at arXiv:2109.06903}\ }
  (\bibinfo {year} {2021})}\BibitemShut {NoStop}%
\bibitem [{\citenamefont {Naimark}(1940)}]{Naimark1940}%
  \BibitemOpen
  \bibfield  {author} {\bibinfo {author} {\bibfnamefont {M.}~\bibnamefont
  {Naimark}},\ }\bibfield  {title} {\bibinfo {title} {\emph {Spectral functions
  of a symmetric operator}},\ }\href@noop {} {\bibfield  {journal} {\bibinfo
  {journal} {Bull. Acad. Sci. URSS. S\'{e}r. Math. [Izvestia Akad. Nauk SSSR]}\
  }\textbf {\bibinfo {volume} {4}},\ \bibinfo {pages} {277} (\bibinfo {year}
  {1940})}\BibitemShut {NoStop}%
\bibitem [{\citenamefont {Hradil}(1997)}]{PhysRevA.55.R1561}%
  \BibitemOpen
  \bibfield  {author} {\bibinfo {author} {\bibfnamefont {Z.}~\bibnamefont
  {Hradil}},\ }\bibfield  {title} {\bibinfo {title} {\emph {Quantum-state
  estimation}},\ }\href {http://dx.doi.org/10.1103/PhysRevA.55.R1561}
  {\bibfield  {journal} {\bibinfo  {journal} {Phys. Rev. A}\ }\textbf {\bibinfo
  {volume} {55}},\ \bibinfo {pages} {R1561} (\bibinfo {year}
  {1997})}\BibitemShut {NoStop}%
\bibitem [{\citenamefont {Wood}(2015)}]{Wood2015}%
  \BibitemOpen
  \bibfield  {author} {\bibinfo {author} {\bibfnamefont {C.~J.}\ \bibnamefont
  {Wood}},\ }\emph {\bibinfo {title} {Initialization and characterization of
  open quantum systems}},\ \href@noop {} {Ph.D. thesis},\ \bibinfo  {school}
  {University of Waterloo} (\bibinfo {year} {2015})\BibitemShut {NoStop}%
\bibitem [{\citenamefont {Smolin}\ \emph {et~al.}(2012)\citenamefont {Smolin},
  \citenamefont {Gambetta},\ and\ \citenamefont {Smith}}]{Smolin2012}%
  \BibitemOpen
  \bibfield  {author} {\bibinfo {author} {\bibfnamefont {J.~A.}\ \bibnamefont
  {Smolin}}, \bibinfo {author} {\bibfnamefont {J.~M.}\ \bibnamefont
  {Gambetta}}, \ and\ \bibinfo {author} {\bibfnamefont {G.}~\bibnamefont
  {Smith}},\ }\bibfield  {title} {\bibinfo {title} {\emph {Efficient Method for
  Computing the Maximum-Likelihood Quantum State from Measurements with
  Additive Gaussian Noise}},\ }\href
  {http://dx.doi.org/10.1103/physrevlett.108.070502} {\bibfield  {journal}
  {\bibinfo  {journal} {Physical Review Letters}\ }\textbf {\bibinfo {volume}
  {108}},\ \bibinfo {pages} {070502} (\bibinfo {year} {2012})}\BibitemShut
  {NoStop}%
\bibitem [{\citenamefont {Sugiyama}\ \emph {et~al.}(2013)\citenamefont
  {Sugiyama}, \citenamefont {Turner},\ and\ \citenamefont
  {Murao}}]{Sugiyama2013}%
  \BibitemOpen
  \bibfield  {author} {\bibinfo {author} {\bibfnamefont {T.}~\bibnamefont
  {Sugiyama}}, \bibinfo {author} {\bibfnamefont {P.~S.}\ \bibnamefont
  {Turner}}, \ and\ \bibinfo {author} {\bibfnamefont {M.}~\bibnamefont
  {Murao}},\ }\bibfield  {title} {\bibinfo {title} {\emph {Precision-Guaranteed
  Quantum Tomography}},\ }\href
  {http://dx.doi.org/10.1103/PhysRevLett.111.160406} {\bibfield  {journal}
  {\bibinfo  {journal} {Phys. Rev. Lett.}\ }\textbf {\bibinfo {volume} {111}},\
  \bibinfo {pages} {160406} (\bibinfo {year} {2013})}\BibitemShut {NoStop}%
\bibitem [{\citenamefont {Gu{\c{t}}{\u{a}}}\ \emph {et~al.}(2020)\citenamefont
  {Gu{\c{t}}{\u{a}}}, \citenamefont {Kahn}, \citenamefont {Kueng},\ and\
  \citenamefont {Tropp}}]{Guta2020}%
  \BibitemOpen
  \bibfield  {author} {\bibinfo {author} {\bibfnamefont {M.}~\bibnamefont
  {Gu{\c{t}}{\u{a}}}}, \bibinfo {author} {\bibfnamefont {J.}~\bibnamefont
  {Kahn}}, \bibinfo {author} {\bibfnamefont {R.}~\bibnamefont {Kueng}}, \ and\
  \bibinfo {author} {\bibfnamefont {J.~A.}\ \bibnamefont {Tropp}},\ }\bibfield
  {title} {\bibinfo {title} {\emph {Fast state tomography with optimal error
  bounds}},\ }\href {http://dx.doi.org/10.1088/1751-8121/ab8111} {\bibfield
  {journal} {\bibinfo  {journal} {Journal of Physics A: Mathematical and
  Theoretical}\ }\textbf {\bibinfo {volume} {53}},\ \bibinfo {pages} {204001}
  (\bibinfo {year} {2020})}\BibitemShut {NoStop}%
\bibitem [{\citenamefont {Blume-Kohout}(2010)}]{Blume_Kohout_2010}%
  \BibitemOpen
  \bibfield  {author} {\bibinfo {author} {\bibfnamefont {R.}~\bibnamefont
  {Blume-Kohout}},\ }\bibfield  {title} {\bibinfo {title} {\emph {Optimal,
  reliable estimation of quantum states}},\ }\href
  {http://dx.doi.org/10.1088/1367-2630/12/4/043034} {\bibfield  {journal}
  {\bibinfo  {journal} {New Journal of Physics}\ }\textbf {\bibinfo {volume}
  {12}},\ \bibinfo {pages} {043034} (\bibinfo {year} {2010})}\BibitemShut
  {NoStop}%
\bibitem [{\citenamefont {Lukens}\ \emph {et~al.}(2020)\citenamefont {Lukens},
  \citenamefont {Law}, \citenamefont {Jasra},\ and\ \citenamefont
  {Lougovski}}]{Lukens_2020}%
  \BibitemOpen
  \bibfield  {author} {\bibinfo {author} {\bibfnamefont {J.~M.}\ \bibnamefont
  {Lukens}}, \bibinfo {author} {\bibfnamefont {K.~J.~H.}\ \bibnamefont {Law}},
  \bibinfo {author} {\bibfnamefont {A.}~\bibnamefont {Jasra}}, \ and\ \bibinfo
  {author} {\bibfnamefont {P.}~\bibnamefont {Lougovski}},\ }\bibfield  {title}
  {\bibinfo {title} {\emph {A practical and efficient approach for Bayesian
  quantum state estimation}},\ }\href
  {http://dx.doi.org/10.1088/1367-2630/ab8efa} {\bibfield  {journal} {\bibinfo
  {journal} {New Journal of Physics}\ }\textbf {\bibinfo {volume} {22}},\
  \bibinfo {pages} {063038} (\bibinfo {year} {2020})}\BibitemShut {NoStop}%
\bibitem [{\citenamefont {Elben}\ \emph {et~al.}(2020)\citenamefont {Elben},
  \citenamefont {Kueng}, \citenamefont {Huang}, \citenamefont {van Bijnen},
  \citenamefont {Kokail}, \citenamefont {Dalmonte}, \citenamefont {Calabrese},
  \citenamefont {Kraus}, \citenamefont {Preskill}, \citenamefont {Zoller},\
  and\ \citenamefont {Vermersch}}]{elben2020}%
  \BibitemOpen
  \bibfield  {author} {\bibinfo {author} {\bibfnamefont {A.}~\bibnamefont
  {Elben}}, \bibinfo {author} {\bibfnamefont {R.}~\bibnamefont {Kueng}},
  \bibinfo {author} {\bibfnamefont {H.-Y.~R.}\ \bibnamefont {Huang}}, \bibinfo
  {author} {\bibfnamefont {R.}~\bibnamefont {van Bijnen}}, \bibinfo {author}
  {\bibfnamefont {C.}~\bibnamefont {Kokail}}, \bibinfo {author} {\bibfnamefont
  {M.}~\bibnamefont {Dalmonte}}, \bibinfo {author} {\bibfnamefont
  {P.}~\bibnamefont {Calabrese}}, \bibinfo {author} {\bibfnamefont
  {B.}~\bibnamefont {Kraus}}, \bibinfo {author} {\bibfnamefont
  {J.}~\bibnamefont {Preskill}}, \bibinfo {author} {\bibfnamefont
  {P.}~\bibnamefont {Zoller}}, \ and\ \bibinfo {author} {\bibfnamefont
  {B.}~\bibnamefont {Vermersch}},\ }\bibfield  {title} {\bibinfo {title} {\emph
  {Mixed-State Entanglement from Local Randomized Measurements}},\ }\href
  {http://dx.doi.org/10.1103/PhysRevLett.125.200501} {\bibfield  {journal}
  {\bibinfo  {journal} {Phys. Rev. Lett.}\ }\textbf {\bibinfo {volume} {125}},\
  \bibinfo {pages} {200501} (\bibinfo {year} {2020})}\BibitemShut {NoStop}%
\bibitem [{\citenamefont {Neven}\ \emph {et~al.}(2021)\citenamefont {Neven},
  \citenamefont {Carrasco}, \citenamefont {Vitale}, \citenamefont {Kokail},
  \citenamefont {Elben}, \citenamefont {Dalmonte}, \citenamefont {Calabrese},
  \citenamefont {Zoller}, \citenamefont {Vermersch}, \citenamefont {Kueng},\
  and\ \citenamefont {Kraus}}]{Neven2021}%
  \BibitemOpen
  \bibfield  {author} {\bibinfo {author} {\bibfnamefont {A.}~\bibnamefont
  {Neven}}, \bibinfo {author} {\bibfnamefont {J.}~\bibnamefont {Carrasco}},
  \bibinfo {author} {\bibfnamefont {V.}~\bibnamefont {Vitale}}, \bibinfo
  {author} {\bibfnamefont {C.}~\bibnamefont {Kokail}}, \bibinfo {author}
  {\bibfnamefont {A.}~\bibnamefont {Elben}}, \bibinfo {author} {\bibfnamefont
  {M.}~\bibnamefont {Dalmonte}}, \bibinfo {author} {\bibfnamefont
  {P.}~\bibnamefont {Calabrese}}, \bibinfo {author} {\bibfnamefont
  {P.}~\bibnamefont {Zoller}}, \bibinfo {author} {\bibfnamefont
  {B.}~\bibnamefont {Vermersch}}, \bibinfo {author} {\bibfnamefont
  {R.}~\bibnamefont {Kueng}}, \ and\ \bibinfo {author} {\bibfnamefont
  {B.}~\bibnamefont {Kraus}},\ }\bibfield  {title} {\bibinfo {title} {\emph
  {Symmetry-resolved entanglement detection using partial transpose moments}},\
  }\href {http://dx.doi.org/10.1038/s41534-021-00487-y} {\bibfield  {journal}
  {\bibinfo  {journal} {npj Quantum Information}\ }\textbf {\bibinfo {volume}
  {7}},\ \bibinfo {pages} {152} (\bibinfo {year} {2021})}\BibitemShut {NoStop}%
\bibitem [{\citenamefont {Schindler}\ \emph {et~al.}(2013)\citenamefont
  {Schindler}, \citenamefont {Nigg}, \citenamefont {Monz}, \citenamefont
  {Barreiro}, \citenamefont {Martinez}, \citenamefont {Wang}, \citenamefont
  {Quint}, \citenamefont {Brandl}, \citenamefont {Nebendahl}, \citenamefont
  {Roos}, \citenamefont {Chwalla}, \citenamefont {Hennrich},\ and\
  \citenamefont {Blatt}}]{Schindler2013}%
  \BibitemOpen
  \bibfield  {author} {\bibinfo {author} {\bibfnamefont {P.}~\bibnamefont
  {Schindler}}, \bibinfo {author} {\bibfnamefont {D.}~\bibnamefont {Nigg}},
  \bibinfo {author} {\bibfnamefont {T.}~\bibnamefont {Monz}}, \bibinfo {author}
  {\bibfnamefont {J.~T.}\ \bibnamefont {Barreiro}}, \bibinfo {author}
  {\bibfnamefont {E.}~\bibnamefont {Martinez}}, \bibinfo {author}
  {\bibfnamefont {S.~X.}\ \bibnamefont {Wang}}, \bibinfo {author}
  {\bibfnamefont {S.}~\bibnamefont {Quint}}, \bibinfo {author} {\bibfnamefont
  {M.~F.}\ \bibnamefont {Brandl}}, \bibinfo {author} {\bibfnamefont
  {V.}~\bibnamefont {Nebendahl}}, \bibinfo {author} {\bibfnamefont {C.~F.}\
  \bibnamefont {Roos}}, \bibinfo {author} {\bibfnamefont {M.}~\bibnamefont
  {Chwalla}}, \bibinfo {author} {\bibfnamefont {M.}~\bibnamefont {Hennrich}}, \
  and\ \bibinfo {author} {\bibfnamefont {R.}~\bibnamefont {Blatt}},\ }\bibfield
   {title} {\bibinfo {title} {\emph {A quantum information processor with
  trapped ions}},\ }\href {http://dx.doi.org/10.1088/1367-2630/15/12/123012}
  {\bibfield  {journal} {\bibinfo  {journal} {New J. Phys.}\ }\textbf {\bibinfo
  {volume} {15}},\ \bibinfo {pages} {123012} (\bibinfo {year}
  {2013})}\BibitemShut {NoStop}%
\bibitem [{\citenamefont {Enr{\'{\i}}quez}\ \emph {et~al.}(2016)\citenamefont
  {Enr{\'{\i}}quez}, \citenamefont {Wintrowicz},\ and\ \citenamefont
  {{\.{Z}}yczkowski}}]{Enriquez2016}%
  \BibitemOpen
  \bibfield  {author} {\bibinfo {author} {\bibfnamefont {M.}~\bibnamefont
  {Enr{\'{\i}}quez}}, \bibinfo {author} {\bibfnamefont {I.}~\bibnamefont
  {Wintrowicz}}, \ and\ \bibinfo {author} {\bibfnamefont {K.}~\bibnamefont
  {{\.{Z}}yczkowski}},\ }\bibfield  {title} {\bibinfo {title} {\emph {Maximally
  Entangled Multipartite States: A Brief Survey}},\ }\href
  {http://dx.doi.org/10.1088/1742-6596/698/1/012003} {\bibfield  {journal}
  {\bibinfo  {journal} {Journal of Physics: Conference Series}\ }\textbf
  {\bibinfo {volume} {698}},\ \bibinfo {pages} {012003} (\bibinfo {year}
  {2016})}\BibitemShut {NoStop}%
\bibitem [{\citenamefont {Helwig}\ and\ \citenamefont
  {Cui}(2013)}]{Helwig2013}%
  \BibitemOpen
  \bibfield  {author} {\bibinfo {author} {\bibfnamefont {W.}~\bibnamefont
  {Helwig}}\ and\ \bibinfo {author} {\bibfnamefont {W.}~\bibnamefont {Cui}},\
  }\bibfield  {title} {\bibinfo {title} {\emph {{Absolutely Maximally Entangled
  States: Existence and Applications}}},\ }\href@noop {} {\bibfield  {journal}
  {\bibinfo  {journal} {Preprint at arXiv:1306.2536}\ } (\bibinfo {year}
  {2013})}\BibitemShut {NoStop}%
\bibitem [{\citenamefont {Grassl}\ \emph {et~al.}(2003)\citenamefont {Grassl},
  \citenamefont {Beth},\ and\ \citenamefont {Roetteler}}]{Grassl2003}%
  \BibitemOpen
  \bibfield  {author} {\bibinfo {author} {\bibfnamefont {M.}~\bibnamefont
  {Grassl}}, \bibinfo {author} {\bibfnamefont {T.}~\bibnamefont {Beth}}, \ and\
  \bibinfo {author} {\bibfnamefont {M.}~\bibnamefont {Roetteler}},\ }\bibfield
  {title} {\bibinfo {title} {\emph {{On optimal quantum codes}}},\ }\href
  {http://dx.doi.org/10.1142/S0219749904000079} {\bibfield  {journal} {\bibinfo
   {journal} {Int. J. Quantum Inf.}\ }\textbf {\bibinfo {volume} {02}},\
  \bibinfo {pages} {55} (\bibinfo {year} {2003})}\BibitemShut {NoStop}%
\bibitem [{\citenamefont {Muralidharan}\ and\ \citenamefont
  {Panigrahi}(2008)}]{Muralidharan2008}%
  \BibitemOpen
  \bibfield  {author} {\bibinfo {author} {\bibfnamefont {S.}~\bibnamefont
  {Muralidharan}}\ and\ \bibinfo {author} {\bibfnamefont {P.~K.}\ \bibnamefont
  {Panigrahi}},\ }\bibfield  {title} {\bibinfo {title} {\emph {Perfect
  teleportation, quantum-state sharing, and superdense coding through a
  genuinely entangled five-qubit state}},\ }\href
  {http://dx.doi.org/10.1103/PhysRevA.77.032321} {\bibfield  {journal}
  {\bibinfo  {journal} {Phys. Rev. A}\ }\textbf {\bibinfo {volume} {77}},\
  \bibinfo {pages} {032321} (\bibinfo {year} {2008})}\BibitemShut {NoStop}%
\bibitem [{\citenamefont {Brydges}\ \emph {et~al.}(2019)\citenamefont
  {Brydges}, \citenamefont {Elben}, \citenamefont {Jurcevic}, \citenamefont
  {Vermersch}, \citenamefont {Maier}, \citenamefont {Lanyon}, \citenamefont
  {Zoller}, \citenamefont {Blatt},\ and\ \citenamefont {Roos}}]{Brydges2019}%
  \BibitemOpen
  \bibfield  {author} {\bibinfo {author} {\bibfnamefont {T.}~\bibnamefont
  {Brydges}}, \bibinfo {author} {\bibfnamefont {A.}~\bibnamefont {Elben}},
  \bibinfo {author} {\bibfnamefont {P.}~\bibnamefont {Jurcevic}}, \bibinfo
  {author} {\bibfnamefont {B.}~\bibnamefont {Vermersch}}, \bibinfo {author}
  {\bibfnamefont {C.}~\bibnamefont {Maier}}, \bibinfo {author} {\bibfnamefont
  {B.~P.}\ \bibnamefont {Lanyon}}, \bibinfo {author} {\bibfnamefont
  {P.}~\bibnamefont {Zoller}}, \bibinfo {author} {\bibfnamefont
  {R.}~\bibnamefont {Blatt}}, \ and\ \bibinfo {author} {\bibfnamefont {C.~F.}\
  \bibnamefont {Roos}},\ }\bibfield  {title} {\bibinfo {title} {\emph {Probing
  Renyi entanglement entropy via randomized measurements}},\ }\href
  {http://dx.doi.org/10.1126/science.aau4963} {\bibfield  {journal} {\bibinfo
  {journal} {Science}\ }\textbf {\bibinfo {volume} {364}},\ \bibinfo {pages}
  {260} (\bibinfo {year} {2019})}\BibitemShut {NoStop}%
\bibitem [{\citenamefont {Huang}\ \emph {et~al.}(2021)\citenamefont {Huang},
  \citenamefont {Broughton}, \citenamefont {Cotler}, \citenamefont {Chen},
  \citenamefont {Li}, \citenamefont {Mohseni}, \citenamefont {Neven},
  \citenamefont {Babbush}, \citenamefont {Kueng}, \citenamefont {Preskill}
  \emph {et~al.}}]{Huang2021}%
  \BibitemOpen
  \bibfield  {author} {\bibinfo {author} {\bibfnamefont {H.-Y.}\ \bibnamefont
  {Huang}}, \bibinfo {author} {\bibfnamefont {M.}~\bibnamefont {Broughton}},
  \bibinfo {author} {\bibfnamefont {J.}~\bibnamefont {Cotler}}, \bibinfo
  {author} {\bibfnamefont {S.}~\bibnamefont {Chen}}, \bibinfo {author}
  {\bibfnamefont {J.}~\bibnamefont {Li}}, \bibinfo {author} {\bibfnamefont
  {M.}~\bibnamefont {Mohseni}}, \bibinfo {author} {\bibfnamefont
  {H.}~\bibnamefont {Neven}}, \bibinfo {author} {\bibfnamefont
  {R.}~\bibnamefont {Babbush}}, \bibinfo {author} {\bibfnamefont
  {R.}~\bibnamefont {Kueng}}, \bibinfo {author} {\bibfnamefont
  {J.}~\bibnamefont {Preskill}},  \emph {et~al.},\ }\bibfield  {title}
  {\bibinfo {title} {\emph {Quantum advantage in learning from experiments}},\
  }\href@noop {} {\bibfield  {journal} {\bibinfo  {journal} {arXiv preprint
  arXiv:2112.00778}\ } (\bibinfo {year} {2021})}\BibitemShut {NoStop}%
\bibitem [{\citenamefont {Appleby}\ \emph {et~al.}(2017)\citenamefont
  {Appleby}, \citenamefont {Flammia}, \citenamefont {McConnell},\ and\
  \citenamefont {Yard}}]{Appleby2017}%
  \BibitemOpen
  \bibfield  {author} {\bibinfo {author} {\bibfnamefont {M.}~\bibnamefont
  {Appleby}}, \bibinfo {author} {\bibfnamefont {S.}~\bibnamefont {Flammia}},
  \bibinfo {author} {\bibfnamefont {G.}~\bibnamefont {McConnell}}, \ and\
  \bibinfo {author} {\bibfnamefont {J.}~\bibnamefont {Yard}},\ }\bibfield
  {title} {\bibinfo {title} {\emph {{SICs} and Algebraic Number Theory}},\
  }\href {http://dx.doi.org/10.1007/s10701-017-0090-7} {\bibfield  {journal}
  {\bibinfo  {journal} {Foundations of Physics}\ }\textbf {\bibinfo {volume}
  {47}},\ \bibinfo {pages} {1042} (\bibinfo {year} {2017})}\BibitemShut
  {NoStop}%
\bibitem [{\citenamefont {Haah}\ \emph {et~al.}(2017)\citenamefont {Haah},
  \citenamefont {Harrow}, \citenamefont {Ji}, \citenamefont {Wu},\ and\
  \citenamefont {Yu}}]{Haah2017}%
  \BibitemOpen
  \bibfield  {author} {\bibinfo {author} {\bibfnamefont {J.}~\bibnamefont
  {Haah}}, \bibinfo {author} {\bibfnamefont {A.~W.}\ \bibnamefont {Harrow}},
  \bibinfo {author} {\bibfnamefont {Z.}~\bibnamefont {Ji}}, \bibinfo {author}
  {\bibfnamefont {X.}~\bibnamefont {Wu}}, \ and\ \bibinfo {author}
  {\bibfnamefont {N.}~\bibnamefont {Yu}},\ }\bibfield  {title} {\bibinfo
  {title} {\emph {Sample-optimal tomography of quantum states}},\ }\href
  {http://dx.doi.org/10.1109/tit.2017.2719044} {\bibfield  {journal} {\bibinfo
  {journal} {IEEE Trans. Inform. Theory}\ }\textbf {\bibinfo {volume} {63}},\
  \bibinfo {pages} {5628} (\bibinfo {year} {2017})}\BibitemShut {NoStop}%
\bibitem [{\citenamefont {O'Donnell}\ and\ \citenamefont
  {Wright}(2016)}]{Wright2016}%
  \BibitemOpen
  \bibfield  {author} {\bibinfo {author} {\bibfnamefont {R.}~\bibnamefont
  {O'Donnell}}\ and\ \bibinfo {author} {\bibfnamefont {J.}~\bibnamefont
  {Wright}},\ }\bibfield  {title} {\bibinfo {title} {\emph {Efficient quantum
  tomography}},\ }in\ \href {http://dx.doi.org/10.1145/2897518.2897544} {\emph
  {\bibinfo {booktitle} {S{TOC}'16---{P}roceedings of the 48th {A}nnual {ACM}
  {SIGACT} {S}ymposium on {T}heory of {C}omputing}}}\ (\bibinfo  {publisher}
  {ACM, New York},\ \bibinfo {year} {2016})\ pp.\ \bibinfo {pages}
  {899--912}\BibitemShut {NoStop}%
\bibitem [{\citenamefont {Roberts}\ and\ \citenamefont
  {Yoshida}(2017)}]{Yoshida2017}%
  \BibitemOpen
  \bibfield  {author} {\bibinfo {author} {\bibfnamefont {D.~A.}\ \bibnamefont
  {Roberts}}\ and\ \bibinfo {author} {\bibfnamefont {B.}~\bibnamefont
  {Yoshida}},\ }\bibfield  {title} {\bibinfo {title} {\emph {Chaos and
  complexity by design}},\ }\href {http://dx.doi.org/10.1007/JHEP04(2017)121}
  {\bibfield  {journal} {\bibinfo  {journal} {J. High Energy Phys.}\ ,\
  \bibinfo {pages} {121, front matter+63}} (\bibinfo {year}
  {2017})}\BibitemShut {NoStop}%
\bibitem [{\citenamefont {Brand\~ao}\ \emph {et~al.}(2021)\citenamefont
  {Brand\~ao}, \citenamefont {Chemissany}, \citenamefont {Hunter-Jones},
  \citenamefont {Kueng},\ and\ \citenamefont {Preskill}}]{Brandao2021}%
  \BibitemOpen
  \bibfield  {author} {\bibinfo {author} {\bibfnamefont {F.~G.}\ \bibnamefont
  {Brand\~ao}}, \bibinfo {author} {\bibfnamefont {W.}~\bibnamefont
  {Chemissany}}, \bibinfo {author} {\bibfnamefont {N.}~\bibnamefont
  {Hunter-Jones}}, \bibinfo {author} {\bibfnamefont {R.}~\bibnamefont {Kueng}},
  \ and\ \bibinfo {author} {\bibfnamefont {J.}~\bibnamefont {Preskill}},\
  }\bibfield  {title} {\bibinfo {title} {\emph {Models of Quantum Complexity
  Growth}},\ }\href {http://dx.doi.org/10.1103/PRXQuantum.2.030316} {\bibfield
  {journal} {\bibinfo  {journal} {PRX Quantum}\ }\textbf {\bibinfo {volume}
  {2}},\ \bibinfo {pages} {030316} (\bibinfo {year} {2021})}\BibitemShut
  {NoStop}%
\bibitem [{\citenamefont {Huangjun}(2012)}]{Zhu2012}%
  \BibitemOpen
  \bibfield  {author} {\bibinfo {author} {\bibfnamefont {Z.}~\bibnamefont
  {Huangjun}},\ }\emph {\bibinfo {title} {{Q}uantum {S}tate {E}stimation and
  {S}ymmetric {I}nformationally {C}omplete {POM}s}},\ \href
  {http://scholarbank.nus.edu.sg/handle/10635/35247} {Ph.D. thesis},\ \bibinfo
  {school} {National University of Singapore} (\bibinfo {year}
  {2012})\BibitemShut {NoStop}%
\bibitem [{\citenamefont {Mahler}\ \emph {et~al.}(2013)\citenamefont {Mahler},
  \citenamefont {Rozema}, \citenamefont {Darabi}, \citenamefont {Ferrie},
  \citenamefont {Blume-Kohout},\ and\ \citenamefont
  {Steinberg}}]{mahler2013adaptive}%
  \BibitemOpen
  \bibfield  {author} {\bibinfo {author} {\bibfnamefont {D.~H.}\ \bibnamefont
  {Mahler}}, \bibinfo {author} {\bibfnamefont {L.~A.}\ \bibnamefont {Rozema}},
  \bibinfo {author} {\bibfnamefont {A.}~\bibnamefont {Darabi}}, \bibinfo
  {author} {\bibfnamefont {C.}~\bibnamefont {Ferrie}}, \bibinfo {author}
  {\bibfnamefont {R.}~\bibnamefont {Blume-Kohout}}, \ and\ \bibinfo {author}
  {\bibfnamefont {A.}~\bibnamefont {Steinberg}},\ }\bibfield  {title} {\bibinfo
  {title} {\emph {Adaptive quantum state tomography improves accuracy
  quadratically}},\ }\href
  {https://link.aps.org/doi/10.1103/PhysRevLett.111.183601} {\bibfield
  {journal} {\bibinfo  {journal} {Physical Review Letters}\ }\textbf {\bibinfo
  {volume} {111}},\ \bibinfo {pages} {183601} (\bibinfo {year}
  {2013})}\BibitemShut {NoStop}%
\bibitem [{\citenamefont {Husz{\'a}r}\ and\ \citenamefont
  {Houlsby}(2012)}]{huszar2012adaptive}%
  \BibitemOpen
  \bibfield  {author} {\bibinfo {author} {\bibfnamefont {F.}~\bibnamefont
  {Husz{\'a}r}}\ and\ \bibinfo {author} {\bibfnamefont {N.~M.}\ \bibnamefont
  {Houlsby}},\ }\bibfield  {title} {\bibinfo {title} {\emph {Adaptive Bayesian
  quantum tomography}},\ }\href
  {https://link.aps.org/doi/10.1103/PhysRevA.85.052120} {\bibfield  {journal}
  {\bibinfo  {journal} {Physical Review A}\ }\textbf {\bibinfo {volume} {85}},\
  \bibinfo {pages} {052120} (\bibinfo {year} {2012})}\BibitemShut {NoStop}%
\bibitem [{\citenamefont {Ferrie}(2014)}]{ferrie2014self}%
  \BibitemOpen
  \bibfield  {author} {\bibinfo {author} {\bibfnamefont {C.}~\bibnamefont
  {Ferrie}},\ }\bibfield  {title} {\bibinfo {title} {\emph {Self-guided quantum
  tomography}},\ }\href
  {https://link.aps.org/doi/10.1103/PhysRevLett.113.190404} {\bibfield
  {journal} {\bibinfo  {journal} {Physical Review Letters}\ }\textbf {\bibinfo
  {volume} {113}},\ \bibinfo {pages} {190404} (\bibinfo {year}
  {2014})}\BibitemShut {NoStop}%
\bibitem [{\citenamefont {Hou}\ \emph {et~al.}(2016)\citenamefont {Hou},
  \citenamefont {Zhu}, \citenamefont {Xiang}, \citenamefont {Li},\ and\
  \citenamefont {Guo}}]{hou2016achieving}%
  \BibitemOpen
  \bibfield  {author} {\bibinfo {author} {\bibfnamefont {Z.}~\bibnamefont
  {Hou}}, \bibinfo {author} {\bibfnamefont {H.}~\bibnamefont {Zhu}}, \bibinfo
  {author} {\bibfnamefont {G.-Y.}\ \bibnamefont {Xiang}}, \bibinfo {author}
  {\bibfnamefont {C.-F.}\ \bibnamefont {Li}}, \ and\ \bibinfo {author}
  {\bibfnamefont {G.-C.}\ \bibnamefont {Guo}},\ }\bibfield  {title} {\bibinfo
  {title} {\emph {Achieving quantum precision limit in adaptive qubit state
  tomography}},\ }\href {https://www.nature.com/articles/npjqi20161} {\bibfield
   {journal} {\bibinfo  {journal} {npj Quantum Information}\ }\textbf {\bibinfo
  {volume} {2}},\ \bibinfo {pages} {1} (\bibinfo {year} {2016})}\BibitemShut
  {NoStop}%
\bibitem [{\citenamefont {Kazim}\ \emph {et~al.}(2021)\citenamefont {Kazim},
  \citenamefont {Farooq}, \citenamefont {ur~Rehman}, \citenamefont {Sohail},
  \citenamefont {Pasha}, \citenamefont {Nadeem},\ and\ \citenamefont
  {Ali}}]{Kazim2021}%
  \BibitemOpen
  \bibfield  {author} {\bibinfo {author} {\bibfnamefont {S.~M.}\ \bibnamefont
  {Kazim}}, \bibinfo {author} {\bibfnamefont {A.}~\bibnamefont {Farooq}},
  \bibinfo {author} {\bibfnamefont {J.}~\bibnamefont {ur~Rehman}}, \bibinfo
  {author} {\bibfnamefont {M.}~\bibnamefont {Sohail}}, \bibinfo {author}
  {\bibfnamefont {M.}~\bibnamefont {Pasha}}, \bibinfo {author} {\bibfnamefont
  {M.}~\bibnamefont {Nadeem}}, \ and\ \bibinfo {author} {\bibfnamefont
  {A.}~\bibnamefont {Ali}},\ }\bibfield  {title} {\bibinfo {title} {\emph
  {Adaptive quantum state tomography with iterative particle filtering}},\
  }\href {http://dx.doi.org/10.1007/s11128-021-03267-x} {\bibfield  {journal}
  {\bibinfo  {journal} {Quantum Information Processing}\ }\textbf {\bibinfo
  {volume} {20}},\ \bibinfo {pages} {348} (\bibinfo {year} {2021})}\BibitemShut
  {NoStop}%
\bibitem [{\citenamefont {Chapman}\ \emph {et~al.}(2016)\citenamefont
  {Chapman}, \citenamefont {Ferrie},\ and\ \citenamefont
  {Peruzzo}}]{PhysRevLett.117.040402}%
  \BibitemOpen
  \bibfield  {author} {\bibinfo {author} {\bibfnamefont {R.~J.}\ \bibnamefont
  {Chapman}}, \bibinfo {author} {\bibfnamefont {C.}~\bibnamefont {Ferrie}}, \
  and\ \bibinfo {author} {\bibfnamefont {A.}~\bibnamefont {Peruzzo}},\
  }\bibfield  {title} {\bibinfo {title} {\emph {Experimental Demonstration of
  Self-Guided Quantum Tomography}},\ }\href
  {http://dx.doi.org/10.1103/PhysRevLett.117.040402} {\bibfield  {journal}
  {\bibinfo  {journal} {Phys. Rev. Lett.}\ }\textbf {\bibinfo {volume} {117}},\
  \bibinfo {pages} {040402} (\bibinfo {year} {2016})}\BibitemShut {NoStop}%
\bibitem [{\citenamefont {Hou}\ \emph {et~al.}(2020)\citenamefont {Hou},
  \citenamefont {Tang}, \citenamefont {Ferrie}, \citenamefont {Xiang},
  \citenamefont {Li},\ and\ \citenamefont {Guo}}]{PhysRevA.101.022317}%
  \BibitemOpen
  \bibfield  {author} {\bibinfo {author} {\bibfnamefont {Z.}~\bibnamefont
  {Hou}}, \bibinfo {author} {\bibfnamefont {J.-F.}\ \bibnamefont {Tang}},
  \bibinfo {author} {\bibfnamefont {C.}~\bibnamefont {Ferrie}}, \bibinfo
  {author} {\bibfnamefont {G.-Y.}\ \bibnamefont {Xiang}}, \bibinfo {author}
  {\bibfnamefont {C.-F.}\ \bibnamefont {Li}}, \ and\ \bibinfo {author}
  {\bibfnamefont {G.-C.}\ \bibnamefont {Guo}},\ }\bibfield  {title} {\bibinfo
  {title} {\emph {Experimental realization of self-guided quantum process
  tomography}},\ }\href {http://dx.doi.org/10.1103/PhysRevA.101.022317}
  {\bibfield  {journal} {\bibinfo  {journal} {Phys. Rev. A}\ }\textbf {\bibinfo
  {volume} {101}},\ \bibinfo {pages} {022317} (\bibinfo {year}
  {2020})}\BibitemShut {NoStop}%
\bibitem [{\citenamefont {Rambach}\ \emph {et~al.}(2021)\citenamefont
  {Rambach}, \citenamefont {Qaryan}, \citenamefont {Kewming}, \citenamefont
  {Ferrie}, \citenamefont {White},\ and\ \citenamefont
  {Romero}}]{PhysRevLett.126.100402}%
  \BibitemOpen
  \bibfield  {author} {\bibinfo {author} {\bibfnamefont {M.}~\bibnamefont
  {Rambach}}, \bibinfo {author} {\bibfnamefont {M.}~\bibnamefont {Qaryan}},
  \bibinfo {author} {\bibfnamefont {M.}~\bibnamefont {Kewming}}, \bibinfo
  {author} {\bibfnamefont {C.}~\bibnamefont {Ferrie}}, \bibinfo {author}
  {\bibfnamefont {A.~G.}\ \bibnamefont {White}}, \ and\ \bibinfo {author}
  {\bibfnamefont {J.}~\bibnamefont {Romero}},\ }\bibfield  {title} {\bibinfo
  {title} {\emph {Robust and Efficient High-Dimensional Quantum State
  Tomography}},\ }\href {http://dx.doi.org/10.1103/PhysRevLett.126.100402}
  {\bibfield  {journal} {\bibinfo  {journal} {Phys. Rev. Lett.}\ }\textbf
  {\bibinfo {volume} {126}},\ \bibinfo {pages} {100402} (\bibinfo {year}
  {2021})}\BibitemShut {NoStop}%
\bibitem [{\citenamefont {Qi}\ \emph {et~al.}(2017)\citenamefont {Qi},
  \citenamefont {Hou}, \citenamefont {Wang}, \citenamefont {Guo}, \citenamefont
  {Fang}, \citenamefont {Bao}, \citenamefont {Zhu},\ and\ \citenamefont
  {Liu}}]{Qi2017}%
  \BibitemOpen
  \bibfield  {author} {\bibinfo {author} {\bibfnamefont {B.}~\bibnamefont
  {Qi}}, \bibinfo {author} {\bibfnamefont {Z.}~\bibnamefont {Hou}}, \bibinfo
  {author} {\bibfnamefont {Y.}~\bibnamefont {Wang}}, \bibinfo {author}
  {\bibfnamefont {Q.}~\bibnamefont {Guo}}, \bibinfo {author} {\bibfnamefont
  {Z.-Y.}\ \bibnamefont {Fang}}, \bibinfo {author} {\bibfnamefont {X.-H.}\
  \bibnamefont {Bao}}, \bibinfo {author} {\bibfnamefont {X.}~\bibnamefont
  {Zhu}}, \ and\ \bibinfo {author} {\bibfnamefont {R.-B.}\ \bibnamefont
  {Liu}},\ }\bibfield  {title} {\bibinfo {title} {\emph {{Adaptive quantum
  state tomography via linear regression estimation: Theory and two-qubit
  experiment}}},\ }\href {http://dx.doi.org/10.1038/s41534-017-0016-4}
  {\bibfield  {journal} {\bibinfo  {journal} {{npj Quantum Information}}\
  }\textbf {\bibinfo {volume} {3}},\ \bibinfo {pages} {19} (\bibinfo {year}
  {2017})}\BibitemShut {NoStop}%
\bibitem [{\citenamefont {Struchalin}\ \emph {et~al.}(2016)\citenamefont
  {Struchalin}, \citenamefont {Pogorelov}, \citenamefont {Straupe},
  \citenamefont {Kravtsov}, \citenamefont {Radchenko},\ and\ \citenamefont
  {Kulik}}]{PhysRevA.93.012103}%
  \BibitemOpen
  \bibfield  {author} {\bibinfo {author} {\bibfnamefont {G.~I.}\ \bibnamefont
  {Struchalin}}, \bibinfo {author} {\bibfnamefont {I.~A.}\ \bibnamefont
  {Pogorelov}}, \bibinfo {author} {\bibfnamefont {S.~S.}\ \bibnamefont
  {Straupe}}, \bibinfo {author} {\bibfnamefont {K.~S.}\ \bibnamefont
  {Kravtsov}}, \bibinfo {author} {\bibfnamefont {I.~V.}\ \bibnamefont
  {Radchenko}}, \ and\ \bibinfo {author} {\bibfnamefont {S.~P.}\ \bibnamefont
  {Kulik}},\ }\bibfield  {title} {\bibinfo {title} {\emph {Experimental
  adaptive quantum tomography of two-qubit states}},\ }\href
  {http://dx.doi.org/10.1103/PhysRevA.93.012103} {\bibfield  {journal}
  {\bibinfo  {journal} {Phys. Rev. A}\ }\textbf {\bibinfo {volume} {93}},\
  \bibinfo {pages} {012103} (\bibinfo {year} {2016})}\BibitemShut {NoStop}%
\bibitem [{\citenamefont {Dieks}(1988)}]{Dieks1988}%
  \BibitemOpen
  \bibfield  {author} {\bibinfo {author} {\bibfnamefont {D.}~\bibnamefont
  {Dieks}},\ }\bibfield  {title} {\bibinfo {title} {\emph {Overlap and
  distinguishability of quantum states}},\ }\href
  {http://dx.doi.org/https://doi.org/10.1016/0375-9601(88)90840-7} {\bibfield
  {journal} {\bibinfo  {journal} {Physics Letters A}\ }\textbf {\bibinfo
  {volume} {126}},\ \bibinfo {pages} {303} (\bibinfo {year}
  {1988})}\BibitemShut {NoStop}%
\bibitem [{\citenamefont {Fischer}\ \emph {et~al.}(2022)\citenamefont
  {Fischer}, \citenamefont {Miller}, \citenamefont {Tacchino}, \citenamefont
  {Barkoutsos}, \citenamefont {Egger},\ and\ \citenamefont
  {Tavernelli}}]{Fischer2022}%
  \BibitemOpen
  \bibfield  {author} {\bibinfo {author} {\bibfnamefont {L.~E.}\ \bibnamefont
  {Fischer}}, \bibinfo {author} {\bibfnamefont {D.}~\bibnamefont {Miller}},
  \bibinfo {author} {\bibfnamefont {F.}~\bibnamefont {Tacchino}}, \bibinfo
  {author} {\bibfnamefont {P.~K.}\ \bibnamefont {Barkoutsos}}, \bibinfo
  {author} {\bibfnamefont {D.~J.}\ \bibnamefont {Egger}}, \ and\ \bibinfo
  {author} {\bibfnamefont {I.}~\bibnamefont {Tavernelli}},\ }\href@noop {}
  {\bibinfo {title} {\emph {Ancilla-free implementation of generalized
  measurements for qubits embedded in a qudit space}},\ } (\bibinfo {year}
  {2022})\BibitemShut {NoStop}%
\bibitem [{\citenamefont {M\o{}lmer}\ and\ \citenamefont
  {S\o{}rensen}(1999)}]{MSgate}%
  \BibitemOpen
  \bibfield  {author} {\bibinfo {author} {\bibfnamefont {K.}~\bibnamefont
  {M\o{}lmer}}\ and\ \bibinfo {author} {\bibfnamefont {A.}~\bibnamefont
  {S\o{}rensen}},\ }\bibfield  {title} {\bibinfo {title} {\emph {Multiparticle
  Entanglement of Hot Trapped Ions}},\ }\href
  {http://dx.doi.org/10.1103/PhysRevLett.82.1835} {\bibfield  {journal}
  {\bibinfo  {journal} {Phys. Rev. Lett.}\ }\textbf {\bibinfo {volume} {82}},\
  \bibinfo {pages} {1835} (\bibinfo {year} {1999})}\BibitemShut {NoStop}%
\bibitem [{\citenamefont {Vidal}\ and\ \citenamefont
  {Werner}(2002)}]{Vidal2001}%
  \BibitemOpen
  \bibfield  {author} {\bibinfo {author} {\bibfnamefont {G.}~\bibnamefont
  {Vidal}}\ and\ \bibinfo {author} {\bibfnamefont {R.~F.}\ \bibnamefont
  {Werner}},\ }\bibfield  {title} {\bibinfo {title} {\emph {Computable measure
  of entanglement}},\ }\href {http://dx.doi.org/10.1103/PhysRevA.65.032314}
  {\bibfield  {journal} {\bibinfo  {journal} {Phys. Rev. A}\ }\textbf {\bibinfo
  {volume} {65}},\ \bibinfo {pages} {032314} (\bibinfo {year}
  {2002})}\BibitemShut {NoStop}%
\bibitem [{\citenamefont {Horodecki}(1997)}]{Horodecki1997}%
  \BibitemOpen
  \bibfield  {author} {\bibinfo {author} {\bibfnamefont {P.}~\bibnamefont
  {Horodecki}},\ }\bibfield  {title} {\bibinfo {title} {\emph {Separability
  criterion and inseparable mixed states with positive partial
  transposition}},\ }\href
  {http://dx.doi.org/https://doi.org/10.1016/S0375-9601(97)00416-7} {\bibfield
  {journal} {\bibinfo  {journal} {Physics Letters A}\ }\textbf {\bibinfo
  {volume} {232}},\ \bibinfo {pages} {333} (\bibinfo {year}
  {1997})}\BibitemShut {NoStop}%
\bibitem [{\citenamefont {Klappenecker}\ and\ \citenamefont
  {Rotteler}(2005)}]{Klappenecker2005}%
  \BibitemOpen
  \bibfield  {author} {\bibinfo {author} {\bibfnamefont {A.}~\bibnamefont
  {Klappenecker}}\ and\ \bibinfo {author} {\bibfnamefont {M.}~\bibnamefont
  {Rotteler}},\ }\bibfield  {title} {\bibinfo {title} {\emph {Mutually unbiased
  bases are complex projective 2-designs}},\ }in\ \href
  {http://dx.doi.org/10.1109/ISIT.2005.1523643} {\emph {\bibinfo {booktitle}
  {Proceedings. International Symposium on Information Theory, 2005. ISIT
  2005.}}}\ (\bibinfo {year} {2005})\ pp.\ \bibinfo {pages}
  {1740--1744}\BibitemShut {NoStop}%
\bibitem [{\citenamefont {Kueng}\ and\ \citenamefont
  {Gross}(2015)}]{Kueng2015}%
  \BibitemOpen
  \bibfield  {author} {\bibinfo {author} {\bibfnamefont {R.}~\bibnamefont
  {Kueng}}\ and\ \bibinfo {author} {\bibfnamefont {D.}~\bibnamefont {Gross}},\
  }\bibfield  {title} {\bibinfo {title} {\emph {Qubit stabilizer states are
  complex projective 3-designs}},\ }\href@noop {} {\bibfield  {journal}
  {\bibinfo  {journal} {arXiv preprint arXiv:1510.02767}\ } (\bibinfo {year}
  {2015})}\BibitemShut {NoStop}%
\bibitem [{\citenamefont {Zhu}(2017)}]{Zhu2017}%
  \BibitemOpen
  \bibfield  {author} {\bibinfo {author} {\bibfnamefont {H.}~\bibnamefont
  {Zhu}},\ }\bibfield  {title} {\bibinfo {title} {\emph {Multiqubit Clifford
  groups are unitary 3-designs}},\ }\href
  {http://dx.doi.org/10.1103/PhysRevA.96.062336} {\bibfield  {journal}
  {\bibinfo  {journal} {Phys. Rev. A}\ }\textbf {\bibinfo {volume} {96}},\
  \bibinfo {pages} {062336} (\bibinfo {year} {2017})}\BibitemShut {NoStop}%
\bibitem [{\citenamefont {Webb}(2016)}]{Webb2016}%
  \BibitemOpen
  \bibfield  {author} {\bibinfo {author} {\bibfnamefont {Z.}~\bibnamefont
  {Webb}},\ }\bibfield  {title} {\bibinfo {title} {\emph {The Clifford Group
  Forms a Unitary 3-Design}},\ }\href@noop {} {\bibfield  {journal} {\bibinfo
  {journal} {Quantum Info. Comput.}\ }\textbf {\bibinfo {volume} {16}},\
  \bibinfo {pages} {1379–1400} (\bibinfo {year} {2016})}\BibitemShut
  {NoStop}%
\bibitem [{\citenamefont {Ambainis}\ and\ \citenamefont
  {Emerson}(2007)}]{Ambainis2007}%
  \BibitemOpen
  \bibfield  {author} {\bibinfo {author} {\bibfnamefont {A.}~\bibnamefont
  {Ambainis}}\ and\ \bibinfo {author} {\bibfnamefont {J.}~\bibnamefont
  {Emerson}},\ }\bibfield  {title} {\bibinfo {title} {\emph {Quantum t-designs:
  t-wise Independence in the Quantum World}},\ }in\ \href
  {http://dx.doi.org/10.1109/CCC.2007.26} {\emph {\bibinfo {booktitle}
  {Twenty-Second Annual IEEE Conference on Computational Complexity
  (CCC'07)}}}\ (\bibinfo {year} {2007})\ pp.\ \bibinfo {pages}
  {129--140}\BibitemShut {NoStop}%
\bibitem [{\citenamefont {Foucart}\ and\ \citenamefont
  {Rauhut}(2013)}]{Foucart2013}%
  \BibitemOpen
  \bibfield  {author} {\bibinfo {author} {\bibfnamefont {S.}~\bibnamefont
  {Foucart}}\ and\ \bibinfo {author} {\bibfnamefont {H.}~\bibnamefont
  {Rauhut}},\ }\href {http://dx.doi.org/10.1007/978-0-8176-4948-7} {\emph
  {\bibinfo {title} {A mathematical introduction to compressive sensing}}},\
  Applied and Numerical Harmonic Analysis\ (\bibinfo  {publisher}
  {Birkh\"{a}user/Springer, New York},\ \bibinfo {year} {2013})\ pp.\ \bibinfo
  {pages} {xviii+625}\BibitemShut {NoStop}%
\bibitem [{\citenamefont {Vershynin}(2018)}]{Vershynin2018}%
  \BibitemOpen
  \bibfield  {author} {\bibinfo {author} {\bibfnamefont {R.}~\bibnamefont
  {Vershynin}},\ }\href {http://dx.doi.org/10.1017/9781108231596} {\emph
  {\bibinfo {title} {High-dimensional probability}}},\ \bibinfo {series}
  {Cambridge Series in Statistical and Probabilistic Mathematics},
  Vol.~\bibinfo {volume} {47}\ (\bibinfo  {publisher} {Cambridge University
  Press, Cambridge},\ \bibinfo {year} {2018})\ pp.\ \bibinfo {pages}
  {xiv+284},\ \bibinfo {note} {an introduction with applications in data
  science, With a foreword by Sara van de Geer}\BibitemShut {NoStop}%
\bibitem [{\citenamefont {Chen}\ \emph {et~al.}(2022)\citenamefont {Chen},
  \citenamefont {Cotler}, \citenamefont {Huang},\ and\ \citenamefont
  {Li}}]{Chen2021}%
  \BibitemOpen
  \bibfield  {author} {\bibinfo {author} {\bibfnamefont {S.}~\bibnamefont
  {Chen}}, \bibinfo {author} {\bibfnamefont {J.}~\bibnamefont {Cotler}},
  \bibinfo {author} {\bibfnamefont {H.-Y.}\ \bibnamefont {Huang}}, \ and\
  \bibinfo {author} {\bibfnamefont {J.}~\bibnamefont {Li}},\ }\bibfield
  {title} {\bibinfo {title} {\emph {Exponential separations between learning
  with and without quantum memory}},\ }in\ \href
  {http://dx.doi.org/10.1109/FOCS52979.2021.00063} {\emph {\bibinfo {booktitle}
  {2021 {IEEE} 62nd {A}nnual {S}ymposium on {F}oundations of {C}omputer
  {S}cience---{FOCS} 2021}}}\ (\bibinfo  {publisher} {IEEE Computer Soc., Los
  Alamitos, CA},\ \bibinfo {year} {[2022] \copyright 2022})\ pp.\ \bibinfo
  {pages} {574--585}\BibitemShut {NoStop}%
\bibitem [{\citenamefont {Peres}(1996)}]{Peres1996}%
  \BibitemOpen
  \bibfield  {author} {\bibinfo {author} {\bibfnamefont {A.}~\bibnamefont
  {Peres}},\ }\bibfield  {title} {\bibinfo {title} {\emph {{Separability
  Criterion for Density Matrices}}},\ }\href
  {http://dx.doi.org/10.1103/PhysRevLett.77.1413} {\bibfield  {journal}
  {\bibinfo  {journal} {Phys. Rev. Lett.}\ }\textbf {\bibinfo {volume} {77}},\
  \bibinfo {pages} {1413} (\bibinfo {year} {1996})}\BibitemShut {NoStop}%
\bibitem [{\citenamefont {Horodecki}\ and\ \citenamefont
  {Horodecki}(1996)}]{Horodecki1996}%
  \BibitemOpen
  \bibfield  {author} {\bibinfo {author} {\bibfnamefont {R.}~\bibnamefont
  {Horodecki}}\ and\ \bibinfo {author} {\bibfnamefont {M.}~\bibnamefont
  {Horodecki}},\ }\bibfield  {title} {\bibinfo {title} {\emph
  {Information-theoretic aspects of inseparability of mixed states}},\ }\href
  {http://dx.doi.org/10.1103/PhysRevA.54.1838} {\bibfield  {journal} {\bibinfo
  {journal} {Phys. Rev. A}\ }\textbf {\bibinfo {volume} {54}},\ \bibinfo
  {pages} {1838} (\bibinfo {year} {1996})}\BibitemShut {NoStop}%
\bibitem [{\citenamefont {Horodecki}\ \emph {et~al.}(2009)\citenamefont
  {Horodecki}, \citenamefont {Horodecki}, \citenamefont {Horodecki},\ and\
  \citenamefont {Horodecki}}]{Horodecki2009}%
  \BibitemOpen
  \bibfield  {author} {\bibinfo {author} {\bibfnamefont {R.}~\bibnamefont
  {Horodecki}}, \bibinfo {author} {\bibfnamefont {P.}~\bibnamefont
  {Horodecki}}, \bibinfo {author} {\bibfnamefont {M.}~\bibnamefont
  {Horodecki}}, \ and\ \bibinfo {author} {\bibfnamefont {K.}~\bibnamefont
  {Horodecki}},\ }\bibfield  {title} {\bibinfo {title} {\emph {{Quantum
  entanglement}}},\ }\href {http://dx.doi.org/10.1103/RevModPhys.81.865}
  {\bibfield  {journal} {\bibinfo  {journal} {Rev. Mod. Phys.}\ }\textbf
  {\bibinfo {volume} {81}},\ \bibinfo {pages} {865} (\bibinfo {year}
  {2009})}\BibitemShut {NoStop}%
\end{thebibliography}%


%
\clearpage
 \onecolumngrid
\setcounter{figure}{0}
\setcounter{equation}{0}
\setcounter{table}{0}
\setcounter{section}{0}
\makeatletter 
\renewcommand{\theequation}{A\@arabic\c@equation}
\renewcommand{\thefigure}{A\@arabic\c@figure}
\renewcommand{\thetable}{A\@arabic\c@table}
\renewcommand{\thesection}{A\@Roman\c@section}

\makeatother

\begin{center}
{\bf \large Appendix \\
Experimental single-setting quantum state tomography}
\label{SI}
\end{center}
\medskip 

\section{Experimental toolbox}
\label{app:setup}
Experimental implementations here and in the main text are performed on a trapped-ion quantum computer, which is schematically shown in Fig.~\ref{figS1:Setup}(a). The device operates on a string of $^{40}$Ca$^+$ ions stored in ultra high vacuum using a linear Paul trap. Each ion acts as a qubit encoded in the electronic levels $S_{1/2}(m=-1/2) = \ket{0}$ and $D_{5/2}(m=-1/2) = \ket{1}$ denoting the computational subspace~\cite{Schindler2013}. 

\begin{figure*}[ht]
    \centering
    \includegraphics[width=0.95\textwidth]{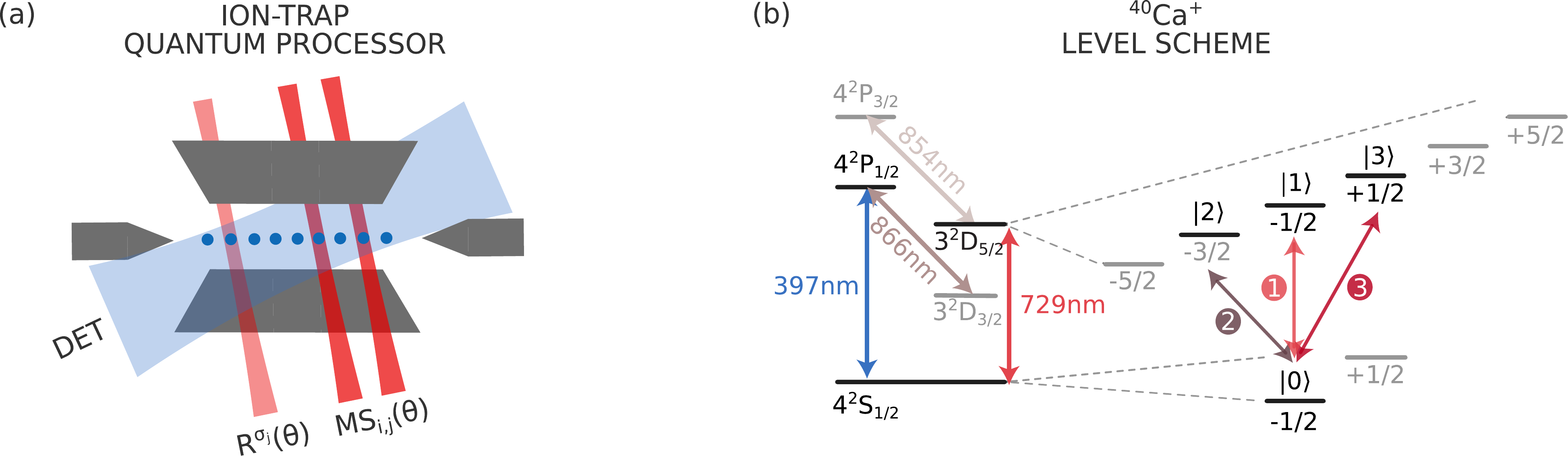}
    \caption{\textbf{Schematic of the trapped-ion quantum processor using $^{40}$Ca$^+$.} \textbf{(a)} Each ion within on the linear string encodes a qubit ($\ket{0}$, $\ket{1}$) or ququart ($\ket{0}$, $\ket{1}$, $\ket{2}$, $\ket{3}$). A universal gate set is realized upon coherent laser-ion interaction using tightly focused beams addressing single ions for local gates (bright red) and pairwise ions for entangling gates (dark red). Alternatively, we can globally address all ions simultaneously. Read-out is performed via collective fluorescence detection (DET), see text for details. \textbf{(b)} $^{40}$Ca$^+$ level-scheme with a dipole transition (\SI{397}{nm}) for cooling and detection, a metastable quadrupole transition (\SI{729}{nm}) for encoding qubits and ququarts within the Zeemann sub-manifold as well as transitions for repumping at \SI{854}{nm} and \SI{866}{nm}. The labeled transitions (1,2,3) provide coherent connection between all ququart states.} 
    \label{figS1:Setup}
\end{figure*}

Quantum state manipulation is realized upon coherent laser-ion interaction. A universal gate-set comprises of addressed single-qubit rotations with an angle $\theta$ around the x- and the y-axis of the form $\mathsf{R}^{\sigma_j}(\theta) = \exp(-i\theta\sigma_j/2)$ with the Pauli operators $\sigma_j = X_j$ or $Y_j$ acting on the $j$-th qubit, together with two-qubit M\o{}lmer-S\o{}renson entangling gate operations $\mathsf{MS}_{i,j}(\theta) = \exp(-i{\theta} X_i X_j /2)$~\cite{MSgate}. Multiple addressed laser beams, coherent among themselves, allow for arbitrary two-qubit connectivity across the entire ion string~\cite{Ringbauer2021}. Optionally, all ions can be addressed simultaneously using a global beam to enable both collective local operations as well as collective entangling operations $\mathsf{MS}(\theta) = \exp(-i{\theta}\sum_{j<\ell} X_j X_\ell /2)$. We choose whatever is more efficient for the underlying experiment. Initial state preparation in $\ket{0}$ is reached after a series of Doppler cooling, polarization-gradient cooling and sideband cooling. Read-out is realized by exciting a dipole transition coupled to the lower qubit level $\ket{0}$ and collecting its scattered photons, from which the computational basis states $\ket{0}$ and $\ket{1}$ can be identified. Thereby, a qubit's state is revealed by accumulating probabilities from multiple experimental runs. The dipole laser collectively covers the entire ion string, which enables a complete read-out in a single measurement round. Additional pump lasers support efficient state preparation as well as cooling and prevent the occupation of unwanted meta-stable states outside the computational subspace $\lbrace\ket{0},\ket{1}\rbrace$. Beyond qubit level, we hold equivalent control over the entire S- and D-state Zeeman manifold with up to 8 levels in each ion, allowing us to encode a higher dimensional quantum decimal digit (qudit), see Fig.~\ref{app:setup}(b). In this work we make use of up to four levels, denoting a ququart via additionally employing $D_{5/2}(m=-3/2)=\ket{2}$ and $D_{5/2}(m=+1/2)=\ket{3}$ alongside both qubit states. Ququart read-out can be performed via three consecutive fluorescence detections, where before the second detection the population between $\ket{0}$ and $\ket{1}$ is switched and before the third and final detection the population between $\ket{0}$ and $\ket{2}$ is switched, see main text Fig.~\ref{fig1:TomographySchematic}(c). Combining this three binary outcomes enables us to evaluate the ququart state probability within a single experimental run. Note that, fluorescence detection in case of measuring a bright state heats up the ion string due to photon scattering. This is counter acted by a sequence of Doppler and polarization-gradient cooling after each individual detection to keep the quality of post measurement bit-flip operations high and thereby suppress detection errors. 

The last paragraph of this experimental setup section is dedicated to technical errors limiting our tomography experiments. Performance on \emph{SIC tomography} is generally found to be moderately lower compared to Pauli tomography, see main text Fig.~\ref{fig3:5qAME}. Evidently, this performance decrease is not inherent to the method, but rather owed to technical errors for two main reasons: 1) The SIC tomography implementation generates an overhead of 5 local pulses per qubit used for mapping qubit to ququart, depicted in Fig.~\ref{fig1:TomographySchematic}(b), as well as two additional bit-flips realizing the four-outcome read-out, see Fig.~\ref{fig1:TomographySchematic}(c). In contrast each Pauli \emph{setting} requires just 1 local pulse per qubit, yet requiring three orthogonal measurements per qubit to extract full tomography information. Our trapped-ion setup has a single-qubit gate fidelity of 0.9994(3) estimated from Randomized benchmarking as well as 2-qubit gate infidelity of roughly 0.98(1) estimated from a decay of fully entangling MS-gates~\cite{Ringbauer2021}. The latter two-qubit gate fidelity might slightly fluctuate from pair to pair. Moreover cross-talk to adjacent ions have an influence. 2) After each insequence-detection the CCD camera demands for a \SI{3}{ms} pause to process the data, before the upcoming sequence continues. However, we utilize this time for re-cooling the ion-string via Doppler and polarization-gradient cooling. During this pause the ions are exposed to amplitude damping due to spontaneous decay from the upper D-states states, having a life-time at about \SI{1}{s}. Accounting for all this throughout measurement taking, we observe a loss of fidelity per qubit between Pauli and SIC tomography of less than \SI{1}{\%}. Thus, extracting complete tomography information in a single experimental run, i.e. shot, comes at the expense of a more complex experiment, that however is of technical nature and can be overcome on future devices with reduced single-qubit error-rates as well as with faster processing CCD-cameras, which nowadays already exist. More importantly, only SIC tomography offers the unique potential to predict non-linear properties in large-scale systems, as pointed out here further below and in the main text. 

\section{Comparing tomography methods}
\label{app:tomo-comparison}
We start off by presenting complementary experimental data covering the 5-qubit AME-state from main text Fig.~\ref{fig3:5qAME}. Whereas in the main text the focus was on scalable approaches, especially SIC-based classical shadows, here we compare these results to Pauli tomography according to Fig.~\ref{fig0:measurement-illustration}(a),(c) using linear inversion and MLE reconstruction following Eq.~\eqref{eq:MLEcvxopt}. Generally, linearly reconstructed density matrices exhibit negative eigenvalues in the case of insufficient statistics, which manifest themselves particularly in unphysical values for non-linear functions of the density matrix, as indicated by purity values above 1. Physical constraints can be imposed on lin-inv, through truncation of negative eigenvalues following Ref.~\cite{Smolin2012,Sugiyama2013,Guta2020}, which we refer to as \emph{projected least squares} (PLS). Importantly, PLS adds only a negligible computational overhead to lin-inv. Consequently, the covered set of tomography approaches is representative in the field of quantum computation and quantum information. Figure~\ref{figS2:5qAMEComparison} depicts results on fidelity, purity, and negativity, with the latter being a common measure of entanglement, albeit one that is challenging to access with experiments.

Note that for 5-qubit Pauli tomography $3^5=243$ settings are required, where for this particular case each setting was repeated a 100 times, leading to the stated maximum shot number of {24\,300}. 100 shots per setting prove to be a good trade-off in the trapped ion platform accounting for both statistics and systematic drifts in the experiment. The moderately lower performance (in terms of the numerical values) of SIC tomography is not inherent to the method, but comes from technical imperfections due to experimental overhead. Particularly, the mapping of the SIC POVM to ququart states, as well as the four-outcome read-out, both essential for single-setting tomography, add experimental complexity, see Appendix~\ref{app:setup}. 

\begin{figure*}[ht]
    \centering
    \includegraphics[width=\textwidth]{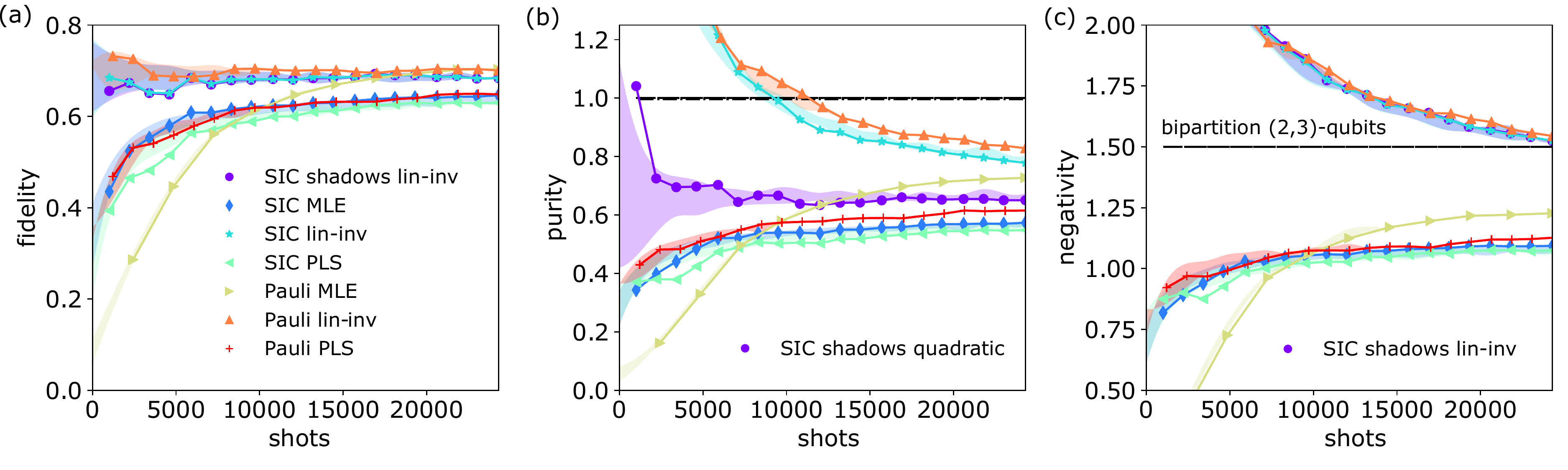}
    \caption{\textbf{Comparison between SIC and Pauli tomography for a standard set of reconstruction methods.} Using the 5-qubit AME-state from main text Fig.~\ref{fig3:5qAME}, we analyze fidelity, purity and negativity from Eq.~\eqref{eq:quantum-negativity}. Due to the increased experimental complexity, we observe that SIC methods generally converge to slightly lower values than Pauli, however this is not inherent to the method, but rather the implementation. \textbf{(a)} Fidelity converges quickest for lin-inv approaches, as expected for general linear functions, see Appendix~\ref{app:linearobservables}. Convergence of PLS, however, is very slow compared to all other methods. SIC MLE at low shot numbers performs better than Pauli MLE as SIC measurements provide the maximum information gain.
    \textbf{(b)} SIC-based classical shadows (``SIC shadows quadratic'') demonstrate the best convergence for purity. MLE methods converge similarly to (a), but lin-inv methods show very slow convergence with highly unphysical results, making their use problematic in practice. \textbf{(c)} As a commonly used entanglement measure, we evaluate the negativity of the reconstructed states, see text for details. As expected, this property converges the slowest, since it is a global property of the density matrix (making this non scalable), with lin-inv again showing highly unphysical results. Bright shaded regions represent 1 standard-deviation around the mean-value when averaging multiple sets at the given shot number.}
    \label{figS2:5qAMEComparison}
\end{figure*}

Figure~\ref{figS2:5qAMEComparison}(a) shows that the fidelity converges very quickly for lin-inv approaches, as is expected for general linear operators, see Appendix~\ref{app:linearobservables}.  Interestingly, Pauli MLE performs worse than all other methods at very few shows, while SIC-based methods perform very well, even for low shot numbers, since this measurement extracts the maximal amount of information for a generic state. 
PLS on the other hand, despite the computational efficiency, converges much slower than the other methods, even including MLE~\cite{Guta2020}. For quadratic measures in Fig.~\ref{figS2:5qAMEComparison}(b) it becomes clear that lin-inv methods produce highly unphysical results which take a long time to converge, limiting the usefulness of these methods in practice. PLS solves this problem, but again shows very slow convergence. Both problems are solved by the SIC-based classical shadow purity estimator from Eq.~\eqref{eq:purity-estimator}, which demonstrates both fast convergence, and accurate (physical) estimates. Finally we study negativity as a commonly used measure of quantum entanglement~\cite{Vidal2001}: 
\begin{equation}
\label{eq:quantum-negativity}
    \mathcal{N}(\rho) = \frac{\vert\vert\rho^{\Gamma_A}\vert\vert_1-1}{2} ,
\end{equation}
where $\rho^{\Gamma_A}$ represents the partial transpose with respect to subsystem $A$ of a bipartition (A,B) together forming $\rho$. The 1-norm in Eq.~\eqref{eq:quantum-negativity} denotes the absolute sum of all negative eigenvalues given by $\rho^{\Gamma_A}$. By construction, the partial transpose of a separable state cannot have negative eigenvalues, such that the negativity vanishes. 
Quantum negativity is an entanglement monotone: if it is positive, then the underlying state must be entangled. The converse, however, need not be true in general~\cite{Horodecki1997}.

As in the main text, we consider the bipartition [2,3] for the 5-qubit system. We find a significantly slower convergence than for Rényi-entropy (see Fig.~\ref{fig3:5qAME}). This is due to the requirement for processing the entire density matrix $\rho$ for quantum negativity as opposed to the classical shadow subsystem purity estimator, which is only evaluated on the smaller partition. The same convergence behaviour is confirmed by numerical simulations discussed in Fig.~\ref{figS6:TomographyComparison} below. 
Classical shadows, on the other hand, allow for tighter classifications of entanglement~\cite{elben2020,Neven2021}. The key idea is to probe the presence of negative eigenvalues in the partial transpose by comparing degree-$d$ polynomials in the underlying density matrices. As $d$ increases, these tests become tighter and eventually recover the negativity condition for entanglement ($\mathcal{N}(\rho) >0$). Classical shadows allow the estimation of all polynomials involved, but the classical post-processing cost becomes less and less favorable as the polynomial degree $d$ increases~\cite{Neven2021}.

We emphasize that from the given set of tomography schemes, SIC-based classical shadow estimators deliver the best results in terms of both convergence as well as practicability. MLE reconstruction typically fails due to lack in computational power and lin-inv neglects physical constraints --- a shortcoming that becomes very pronounced for non-linear observables. Incorporating physical constraints by projection (PLS) remains computationally efficient, but leads to considerably poorer convergence behaviour. Finally, classocal shadows are the only approach for efficiently predicting non-linear functions, such as mixed-state entanglement~\cite{elben2020,Neven2021} of large-scale systems.

\section{Complementary results on rotated 8-qubit GHZ-state}
In the live update discussion of main text Fig.~\ref{fig4:Live8qGHZrotated} all two-qubit bipartitions were evaluated on top of fidelity and purity. Here we analyzed the same data in post-processing to present results on all bipartitions ([1,7]-qubits, [2,6]-qubits, [3,5]-qubits and [4,4]-qubits). This evaluation for all possible pairs was not possible in real-time on a standard desktop computer. We average until 50,000 shots (roughly \SI{1200}{s} of data taking), where the SIC based classical shadow purity estimator (Eq.~\eqref{eq:purity-estimator}) of the 8-qubit state, representing the most demanding property, has converged. Data was taken for a total of \SI{12500}{s}. Note that at around 40,000 shots almost no change is visible in the classical shadow purity as well as the respective standard-deviation. The latter is due to systematic experimental drifts over the course of the long time measurement. Here we batched 100 shots for each approximate $\rho$ following Eq.~\eqref{eq:rho_classical_shadows}, which speeds up the analysis without significant loss in accuracy, see Appendix~\ref{app:batch-size} for a thorough study of batch-sizes considering both analysis time and accuracy. We also observe that the convergence of bipartitions from lin-inv become significantly slower than classical shadows as the subsets get bigger. For PLS individual bipartitions even visually separate, indicating very slow convergence behaviour, confirming similar observations throughout the main text and appendix.

\begin{figure*}[ht]
    \centering
    \includegraphics[width=\textwidth]{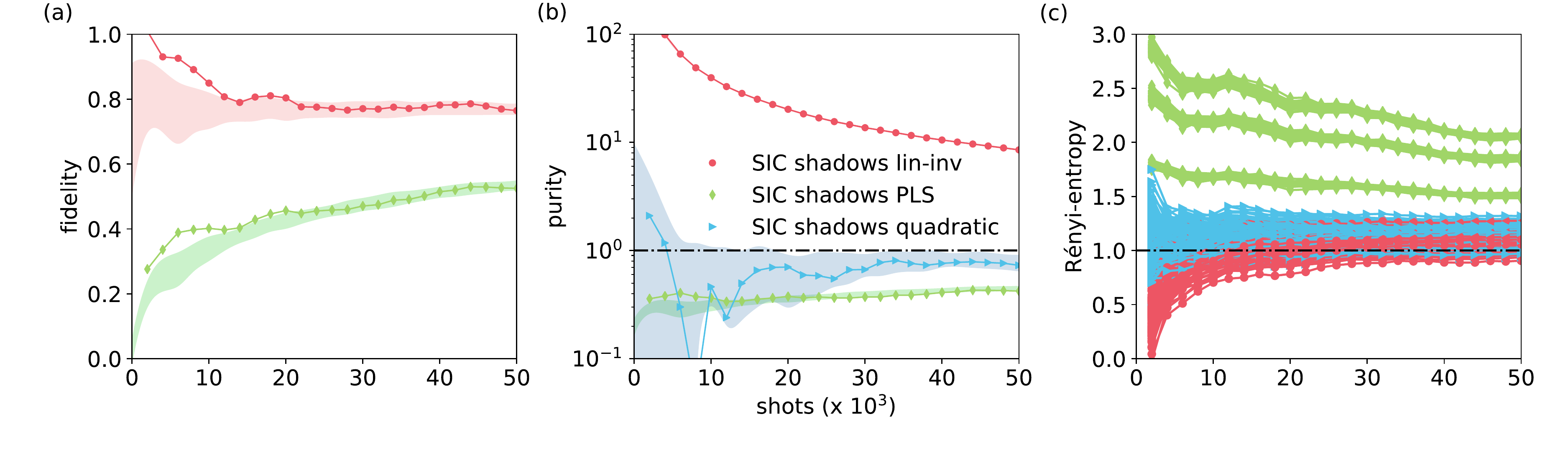}
    \caption{\textbf{Detailed post-processing analysis of the rotated 8-qubit GHZ-state from main text Fig.~\ref{fig4:Live8qGHZrotated}}. Results now cover all bipartitions, which was not feasible in real time on a desktop computer. Shaded regions represent 1 standard-deviation around the mean-value when averaging over multiple sets of a total of \SI{12500}{s} of data. \textbf{(a)} We find fidelity from lin-inv to converge before 20,000 shots, i.e.\ \SI{500}{s} of data-taking. \textbf{(b)} Classical shadow purity as converged well by 40,000 shots, or about \SI{1000}{s} of data-taking, which in the post-processing analysis here is also confirmed by the saturation of the standard-deviation towards more shots. This remaining fluctuations arise from systematic drifts in the experiment over the course of the measurements. \textbf{(c)} SIC-based classical shadows converge even quicker than the time it takes to measure just one shot per Pauli basis. On the other hand, lin-inv and PLS are significantly slower, with PLS even showing a distinct separation between the individual bipartitions ((1,7)-qubits, (2,6)-qubits), (3,5)-qubits, (4,4)-qubits).}
    \label{figS3:Acc8qGHZrotated}
\end{figure*}

Pauli tomography was neglected here as all bipartition Rényi-entropies from SIC-based classical shadows have converged faster than the time it would take to just obtain a single measurement per Pauli setting. Moreover, there is no efficient way of adding physical constraints to the reconstructed state $\rho$ as MLE for an 8-qubit state is unfeasible on standard desktop computers, and PLS delivers significantly worse convergence. The findings here agree with those of the 5-qubit AME-state, previously discussed. SIC-based classical shadow estimators demonstrated to be most suitable for predicting non-linear functions towards larger system-sizes. Apart from the exemplified quadratic measures utilized in purity and Rényi-entropy, classical shadow estimators support higher polynomial functions following the same principles~\cite{elben2020}.

\section{Classical shadows convergences and practicability}
\label{app:batch-size}
In the live update studies from main text Fig.~\ref{fig4:Live8qGHZrotated}(b), we demonstrated real-time analysis of ongoing SIC tomography experiments, until all properties of interest were accurately estimated. How long the analysis can be performed in real time, however, depends on the size of the subsystems to be analyzed, as well as the type and number of functions to be estimated in parallel. Whereas time consumption for linear observables by means of SIC classical shadows remains constant over the course of data acquisition --- individual experimental shots are simply processed and accumulated according to Eq.~\eqref{eq:shadow-convergence} --- non-linear functions generally require higher order products of all combinations from the given set of shots, see Eq.~\eqref{eq:purity-estimator}. The resulting scaling is governed by the maximum polynomial order of the function in question minus one. For quadratic functions, in particular, the computation time for every new shot grows linearly with the number of already accumulated shots $M$. Hence, over the course of the data-acquisition, this can eventually become computationally demanding, especially for large problem sizes where a large number of shots must be accumulated. Thus, to perform non-linear function analysis in real-time, one either keeps subsystems (and by that shot requirements) relatively small, or uses a so-called batching approach, where multiple shots are bundled to estimate a more accurate $\rho$ and thereby reduce the number of costly higher order product combinations. In this section, we discuss the potential, as well as limits of increasing the batch-size in contrast to comparing single shots. To develop a quantitative statement we perform simulations on a 3-qubit linear cluster state considering only statistical quantum projection noise. Along those lines, simulated tomography data is sampled from a multinomial distribution considering different sample sizes to mimic experimental shots. Each noisy set of tomography data is then reconstructed by means of SIC-based classical shadows, in particular, focusing on the quadratic purity estimator from Eq.~\eqref{eq:purity-estimator}. Since we are only interested in changes in convergence behaviour, the choice of state does not affect the qualitative statement of these numerical simulations. We compare 100 different shot numbers for 100 different batch-sizes in a practical regime for 3-qubit states. On the one extreme we compare all combinations of shadows obtained from a single shot each, which is known to be statistically optimal (see also the discussions around Eq.~\eqref{eq:purity-estimator}). On the other extreme we compare just two shadows, each linearly more accurate, since they are obtained from averaging half the shots. In between we can trade-off the quality of the individual estimators versus the number of comparisons between estimators. By additionally accounting for analysis time a sweet-spot can be determined on a case-by-case basis. 

We analyze the convergence behaviour of these different strategies in Fig.~\ref{figS4:3qClassicalShadows2D} by means of the standard deviation of the classical shadow purity estimator, when repeating every point in the 2d-grid 100 times. We plot ``shots'' against ``batch-size'' using a logarithmic color coding. Darker regions refer to a more accurate purity estimate for the underlying state. A region at a specific color always shows a certain almost vertical extension, indicating that batching has a very small effect on accuracy. At the same time, however, it offers to speed up the data analysis significantly. 
We note that the observed pattern qualitatively remains the same for bigger system sizes and accordingly more shots. Especially for bigger subsystems, where a lot of statistics is required, batching has the potential to significantly speed up analysis. We made use of this method in main text Fig.~\ref{fig4:Live8qGHZrotated}, where 100 experimental shots were used for each classical shadow. This turned out to be sufficient to estimate all relevant target properties in real-time.  
\begin{figure*}[ht]
    \centering
    \includegraphics[width=0.8\textwidth]{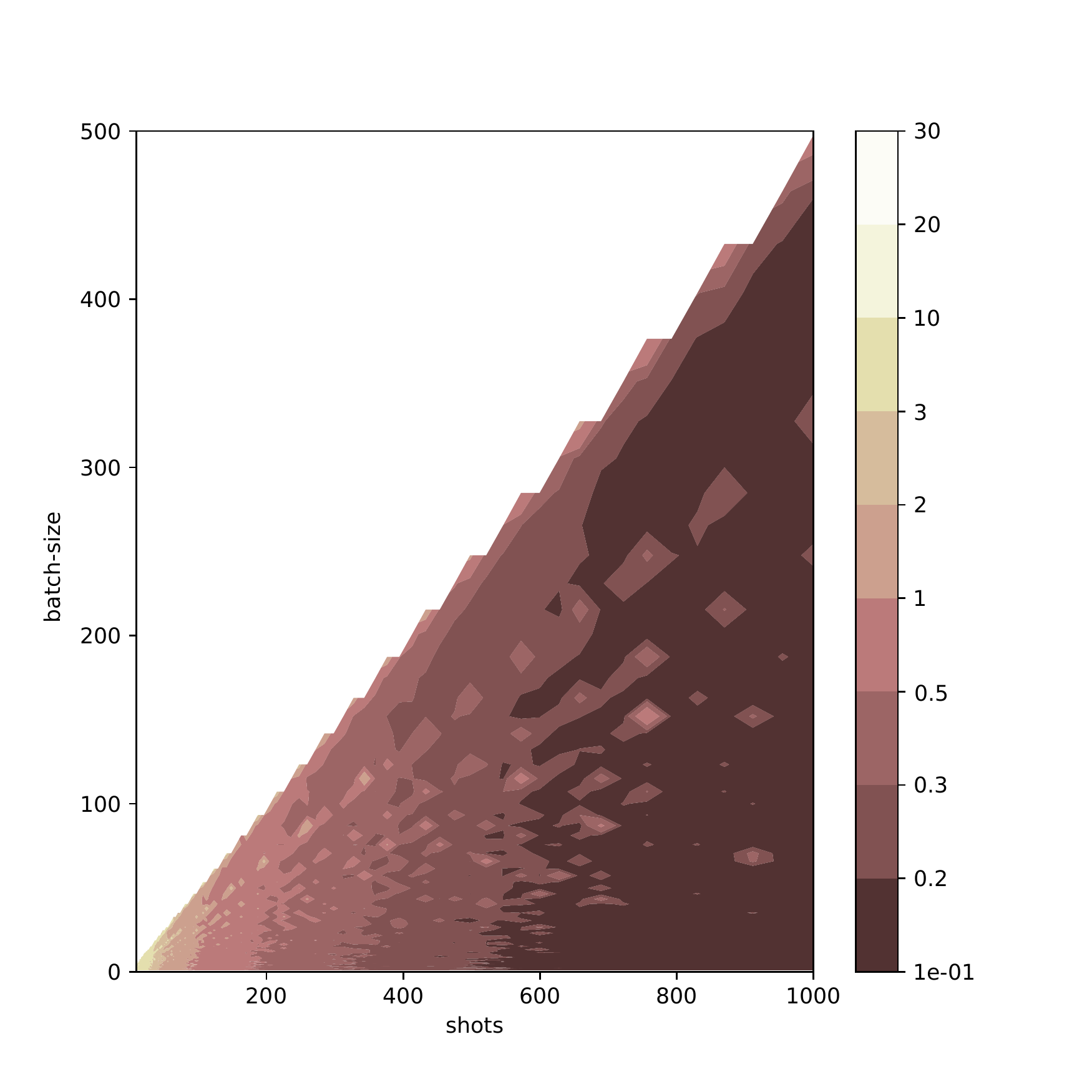}
    \caption{\textbf{Numerical simulations of batching SIC based classical shadow purities.} We plot the purity's standard deviation from 100 repetitions of each point in the 2d grid using a logarithmic color scheme. A 3-qubit linear cluster state is utilized in these numerical simulations, although the particular state and qubit number does not effect the qualitative picture. Regions of similar accuracy are found to have a certain almost vertical extension, indicating that bunching does not significantly degrade accuracy, while greatly reducing computation complexity. The ideal batch size must be evaluated on a case-by-case basis, weighing convergence against analysis time.}
    \label{figS4:3qClassicalShadows2D}
\end{figure*}

\section{Fast state identification with classical shadows}
\label{app:quantumgame}
We now consider a quantum game, where a quiz master targets us experimenters, having access to a quantum computer with a state to prepare and a question about a certain property. Importantly, the question is only revealed after performing the experiment. Such a setup is partly inspired by recent works on quantum-enhanced learning~\cite{Huang2021}. 
Here, we follow a game where the goal is to prepare a random target state from a fixed set of 16 states and then pinpoint this target state in as few experimental shots as possible. Fig.~\ref{figS5a:DecisionGame} depicts the results, where SIC tomography allows us to receive a reward in less than 20 shots. In stark contrast, the minimum number of shots for the same task using Pauli tomography is $3^4=81$. The figure of merit for this game is the estimated distance between the states as the minimum difference in fidelity between the target state and all others. A reward is obtained as soon as that minimum distance remains positive.

\begin{figure}[ht]
    \centering
    \includegraphics[width=0.4\columnwidth]{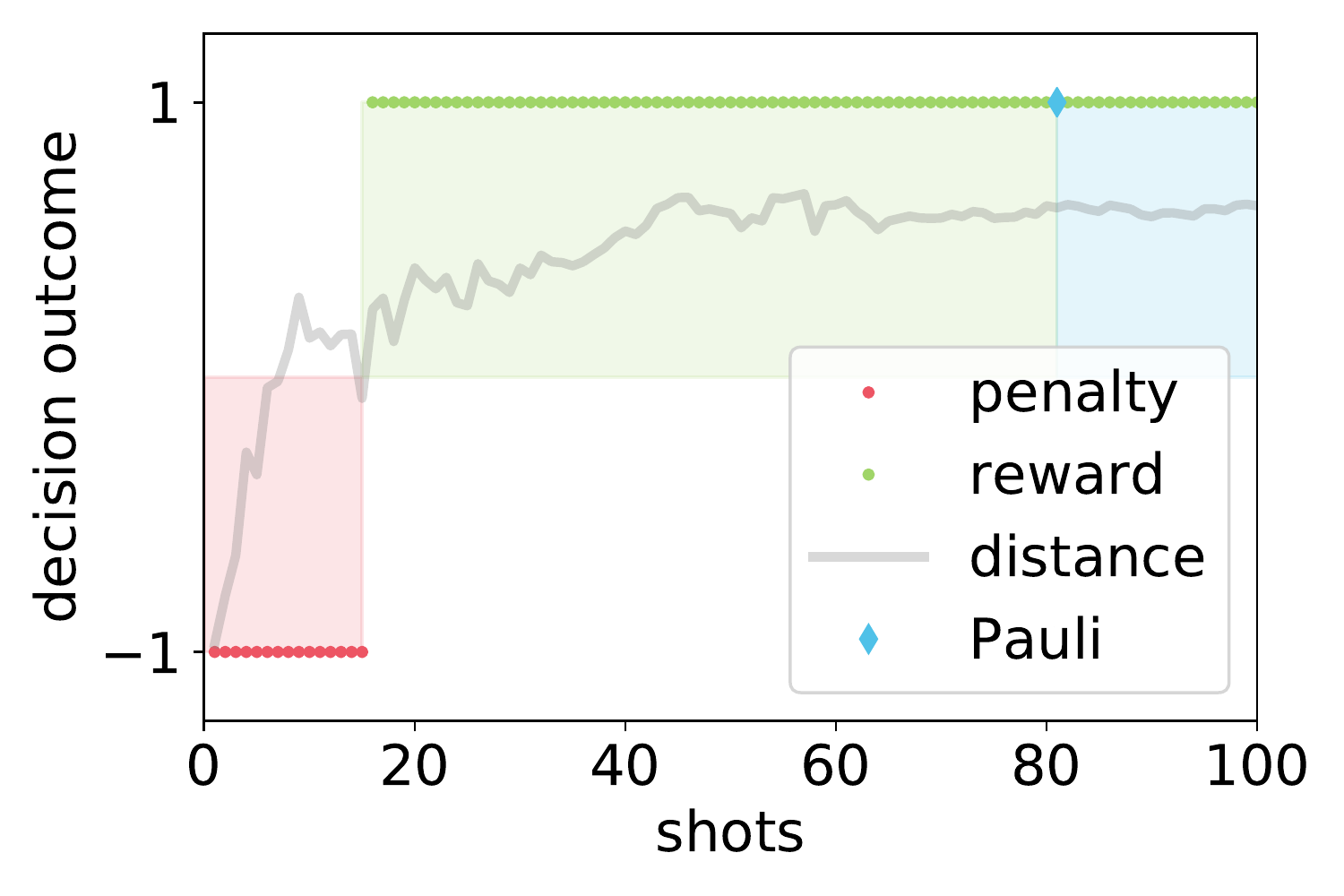}
    \caption{\textbf{Decision game.} A quiz master constructs a quantum game where she provides us experimenters, having access to a quantum computer, with a circuit to prepare a state and perform measurements. After the measurement stage, we must report a certain property. Importantly, the property is only revealed after the experiment has been completed, preventing us from simply measuring the property in question. The aim is to win this challenge with as few experimental shots as possible. \newline In the present case a single state from a fixed set was randomly chosen and implemented, which we then had to identify in as few shots as possible. SIC tomography delivered a reward in less than 20 shots, whereas the smallest Pauli tomography instance requires $3^4=81$ shots. Distance refers to the minimum fidelity difference between the target state and all others, see Fig.~\ref{figS5b:Ortogonal4qCluster}(b).}
    \label{figS5a:DecisionGame}
\end{figure}

Specifically, in the decision game of Fig.~\ref{figS5a:DecisionGame}, we randomly prepare 1 out of 16 orthogonal 4-qubit linear cluster states. These states correspond to all combinations of input states $\ket{\pm}$ on the 4 qubits. The target state is then identified by comparing the the linear-inversion fidelities between the prepared state and all 16 possible targets. Note that, fidelity is a good option for state identifications as linear observables generally converge quickest under lin-inv, as we showed experimentally (see Fig.~\ref{fig3:5qAME} and  Fig.~\ref{fig4:Live8qGHZrotated}) and also formally derived in Appendix~\ref{app:linearobservables}. Figure~\ref{figS5b:Ortogonal4qCluster}(a) depicts fidelities with respect to all 16 states for SIC and Pauli tomography against the number of experimental shots. As expected, only one of the curves is close to fidelity 1 indicating the target state, whereas all others approach 0 within experimental uncertainties. The distinguishing performance is even better visualized by plotting the difference between the target state fidelity and all others in  Fig.~\ref{figS5b:Ortogonal4qCluster}(b). The minimum of these distances (i.e.\ the worst case) is also shown in the background of Fig.~\ref{figS5a:DecisionGame} and used as our state distinguishability criteria. Surprisingly, less than 20 experimental shots are required to pin point the state, clearly undercutting the minimum for Pauli tomography. This argument can in principle be extended to larger systems, where we have seen that properties such as linear observables or Rényi-entropies (see Fig.~\ref{fig4:Live8qGHZrotated}) converge much faster than the $3^N$ experimental shots required for the smallest instance of Pauli tomography. Fast state identification represents another fruitful example where SIC-based classical shadows not only appear more practical but also impart quicker convergence for certain tasks than Pauli tomography. The difference in performance between the tomography approaches again originates from SIC's bigger experimental overhead, see Appendix~\ref{app:setup}.

\begin{figure*}[ht]
    \centering
    \includegraphics[width=0.8\textwidth]{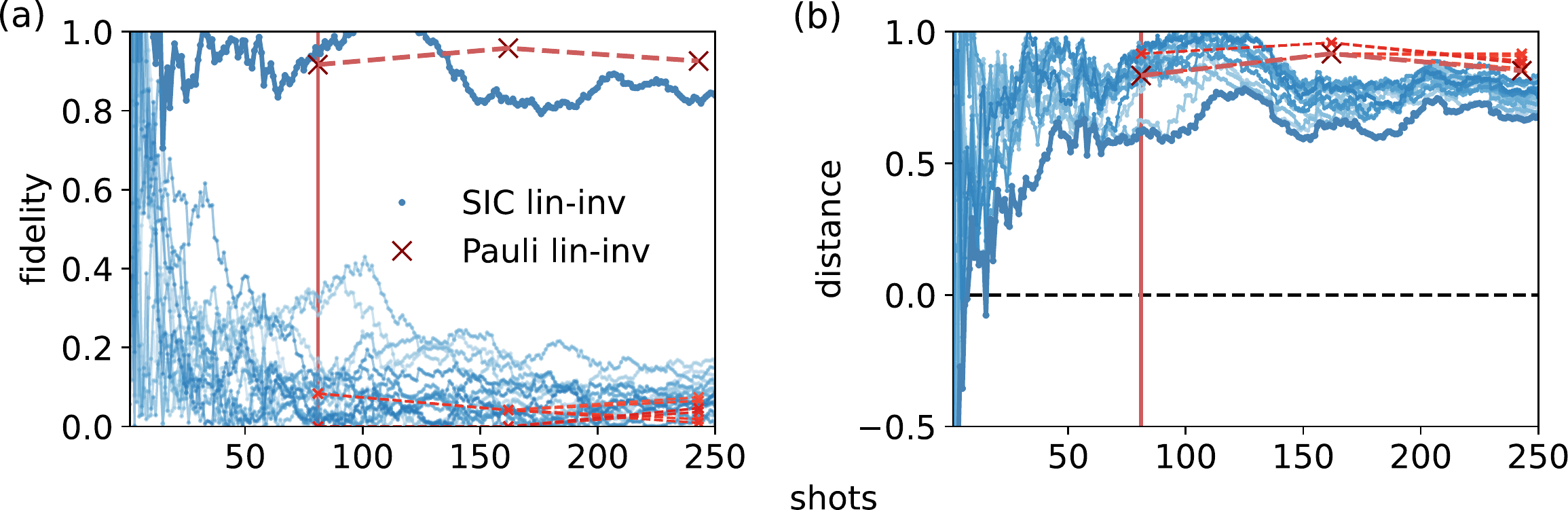}
    \caption{\textbf{Fast state identification from SIC tomography.} Complementary analysis of the quantum game from Fig.~\ref{figS5a:DecisionGame}. \textbf{(a)} Estimated fidelity with each of the 16 possible states. Via SIC tomography, we manage to identify the state with less than 20 shots, clearly undercutting the smallest Pauli implementation using $3^4=81$ shots. \textbf{(b)} For illustration purposes, we plot the fidelity difference between the generated state and all possible states as individual curves, which provides a distance measure. The minimum of these differences was used as a state distinguishability criteria to gain a reward.}
    \label{figS5b:Ortogonal4qCluster}
\end{figure*}

\section{Numerical simulations comparing tomography approaches}
\label{app:simulations}
This section aims to support previous experimental findings through numerical simulations of convergence under statistical quantum shot noise. To this extent, simulated tomography data is generated by sampling from a multinomial distribution, where the sample size is given by the number of shots. For the sake of comparability, we perform these numerical convergence simulations using use the 5-qubit AME-state from main text Fig.~\ref{fig3:5qAME} and Eq.~\eqref{eq:AMEstate}. We study both SIC and Pauli tomography and cover all reconstruction methods as experimentally studied in Appendix~\ref{app:tomo-comparison}, i.e.\ lin-inv, PLS, MLE, and, for quadratic functions also SIC-based classical shadows. Numerical results are presented in Fig.~\ref{figS6:TomographyComparison} on a double-logarithmic scale, covering infidelity (to better illustrate convergence), and trace-distance as another commonly used property for state distinguishability
\begin{equation}
\label{eq:trace-distance}
    T(\rho,\sigma) =\frac{1}{2}\Tr\big [\sqrt{[(\rho-\sigma)^\dagger(\rho-\sigma)}\big ] .
\end{equation}
Lin-inv approaches are again found to converge significantly faster than all other methods in terms of infidelity, while the trace-distance shows the opposite behaviour due to being much more complicated to estimate. This is indicative of the unphysical nature of lin-inv estimates.
In contrast, MLE approaches respects those physical boundaries, which result in valid values on all estimators at the cost of a slower convergence. Here we also confirm previous experimental observations (see Fig.~\ref{figS2:5qAMEComparison}) that SIC MLE performs better for small shot numbers than Pauli MLE. This is likely due to the fact that SIC POVMs provide the optimal information gain in each shot. Curiously, however, for larger numbers of shots, Pauli MLE eventually converges faster. This might be due to the overcompleteness of the Pauli basis but remains to be fully understood. While not specifically presented in this manuscript, we found the very same behaviour for other states having different overlap with Pauli basis and SIC POVM, hence this effect does not seem to be due to the choice of state. PLS again produces the slowest converges as we have already seen across our experimental studies. 

Overall, quantum negativity exhibits the slowest convergence, which confirms what we experimentally observed in Fig.~\ref{figS3:Acc8qGHZrotated}. We note that the negativity calculation always requires the full density matrix independent of the bipartition. This is in stark contrast to R`enyi-entropy based measures and general non-linear functions supported on a subset of the system, which can be estimated efficiently using SIC-based classical shadows. The latter again show the fastest convergence for purity according to Eq.~\eqref{eq:purity-estimator}, which will generally be true for classical shadows by construction. To keep simulations efficient the batch-size was chosen as a constant fraction of the total number of shots, which has a negligible effect on convergence, see Appendix~\ref{app:batch-size}. Finally, these numerical simulations could reproduce all features and findings from the experimental studies, thus indicating that there are no principal limitations to our experimental implementation of both SIC and Pauli tomography. Moreover, all reconstruction methods converge to the same values indicating no principle draw back across the various methods.

\begin{figure*}[ht]
    \centering
    \includegraphics[width=0.8\textwidth]{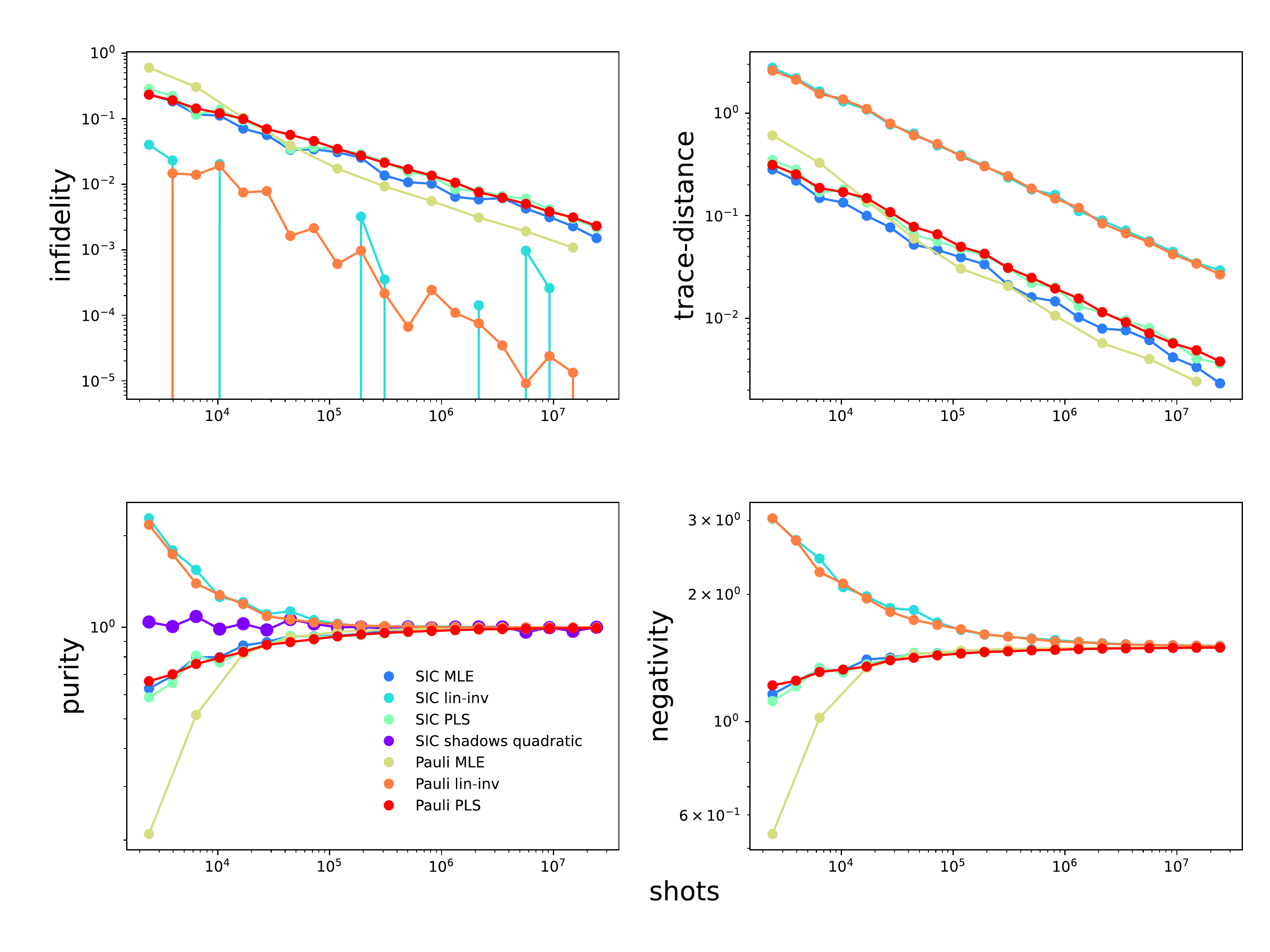}
    \caption{\textbf{Numerical simulations incorporating statistical noise covering all tomography approaches.} To this extend quantum shot noise is incorporated via sampling ideal tomography data from a multinomial distribution, with sample sizes corresponding to number of shots. We additionally present infidelity (1-fidelity) and complementary trace-distance following Eq.~\eqref{eq:trace-distance} and plot in double-logarithmic scale to more clearly visualize several orders of magnitudes. Note that infidelity estimation with lin-inv sometimes delivers unphysical results, indicated by values out of range of the plot).}
    \label{figS6:TomographyComparison}
\end{figure*}

\section{Overlap with tomography basis in experiment \& simulation}
\label{app:overlap}
Our experimental studies covered by Fig.~\ref{fig3:5qAME} and~\ref{fig4:Live8qGHZrotated} focus on maximally entangled states that were differently aligned with respect to Pauli basis and SIC POVM in order to not particularly favor either of them. We resume this discussion in more detail by investigating states of varying orientation, with respect to the measurement states, to study its effect on both experiments and numerical simulations. To this extend, convergence behaviour of first purely local states followed by entangled states is demonstrated. 

\begin{figure*}[ht]
    \centering
    \includegraphics[width=\textwidth]{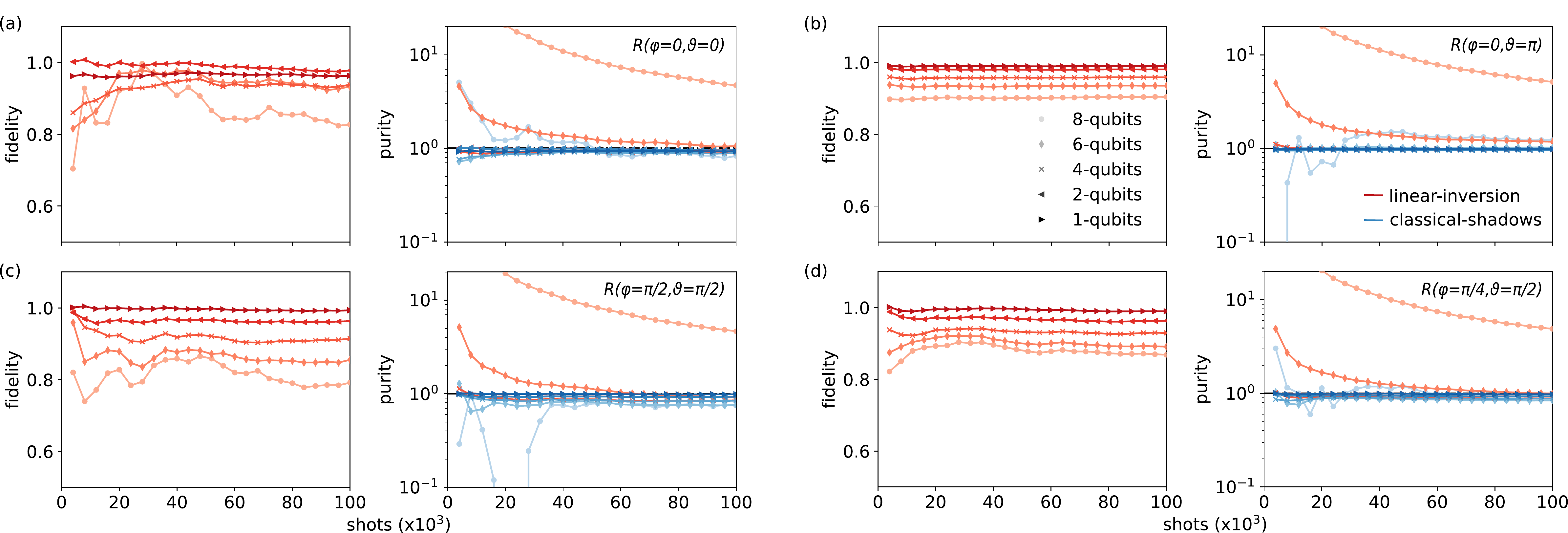}
    \caption{\textbf{Experimental convergence investigation of local states with different orientation, i.e.\ basis overlap.} The given states were collectively prepared on 8 qubits and partially traced to study multiple subsets. \textbf{(a)} $\ket{0}$ state perfectly aligned with the first SIC-vector. \textbf{(b)} $\ket{1}$ state orthogonal to the first SIC-vector, whose component completely vanishes. Convergence in (b) is significantly faster than in (a) as only 3 of the non-orthogonal SIC vectors take part. Boosted convergence is related to unambiguous state discrimination~\cite{Dieks1988} limiting non-orthogonal measurements. \textbf{(c)} Superposition state with maximized overlap with one of the SIC-vectors. \textbf{(d)} Superposition state with minimum overlap with the SIC-vectors. Here, convergence of the state in (d) is better than (a) as the state's information is more regularly distributed over the SIC-vectors, not favouring one, resulting in a higher information gain. Also note that the achievable fidelity generally drops with higher qubit number due to experimental imperfections.}
    \label{figS7:LocalStates}
\end{figure*}

We start off by presenting experimental results on SIC tomography of local states up to 8 qubits having different overlap with the Pauli basis and SIC POVM, depicted in Fig.~\ref{figS7:LocalStates}. Analyzed metrics include fidelity and purity using lin-inv and classical shadows as those represent the scalable approaches. The given states were collectively prepared and subsequently partial traced to study multiple qubit numbers, see Appendix~\ref{app:setup}. Brighter colors denote higher qubit numbers. Among the Pauli basis states $\ket{0}$ and $\ket{1}$, the latter performs significantly better under the SIC POVM. This reflects a general theme, where in using non-orthogonal bases it is preferable if one component vanishes, which is related to the concept of unambiguous state discrimination~\cite{Dieks1988}. In the example of \textbf{$\ket{1}$} this is indeed true for the first SIC-vector aligned with $\ket{0}$ (see Fig.~\ref{figS9:SICrotated}(a)). For superposition states, we find that states that maximize the overlap with one SIC vector (Fig.~\ref{figS7:LocalStates}(c)) perform worse than those that feature more even overlap with all SIC vectors (Fig.~\ref{figS7:LocalStates}(d)). Generally higher qubit number states result in lower fidelities, due to experimental imperfections. We again emphasize the particularly fast convergence of the SIC-based classical shadow purity estimator, being here only moderately slower than fidelity.

\begin{figure*}[ht]
    \centering
    \includegraphics[width=\textwidth]{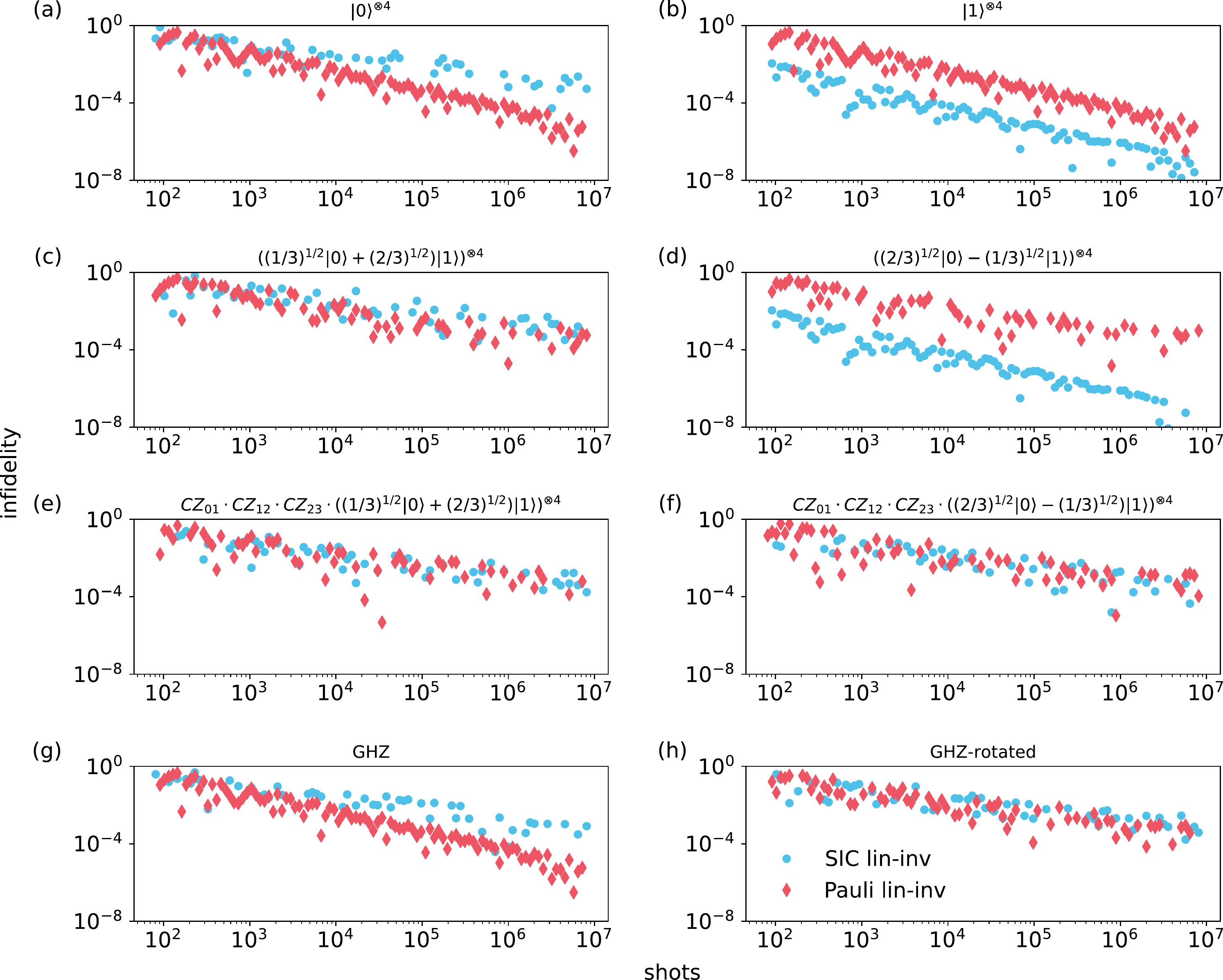}
    \caption{\textbf{Convergence of SIC and Pauli tomography in numerical simulations for states with different basis overlap.} \textbf{(a-d)} Local states that align with a SIC-vector (a,c) show slower convergence for SIC tomography than those orthogonal to a SIC vector (b,d), confirming the experimental observations of Fig.~\ref{figS7:LocalStates}. Pauli tomography, on the other hand, generally performs best for states aligned with the basis (a,b).  \textbf{(e-h)} The differences in convergence performance vanish for SIC tomography when considering entangled states. In case of the Pauli basis, however, some improved convergence can still be observed for specific states, such as the GHZ-state (g), which is aligned with the Pauli basis and thus performs better.}
    \label{figS8:BasisOverlap}
\end{figure*}

To confirm these results for purely local states and to extend the discussion to entangled states of different orientation, we performed numerical simulations on 4 qubit states under quantum shot noise as previously explained within Appendix~\ref{app:simulations}. Figure~\ref{figS8:BasisOverlap} contains results for SIC tomography as well as now for Pauli tomography for predicting infidelity via lin-inv, which is efficient and enables scalability in contrast to MLE, that however, would not change the essence of statements. Figures~\ref{figS8:BasisOverlap}(a-d) cover local states both along a SIC-vector (a,c) as well as orthogonal to it (b,d), i.e.\ in the opposite direction of the Bloch-sphere. These simulations reproduce the effects seen experimentally in Fig.~\ref{figS7:LocalStates}, where states aligned with a SIC vector (a,c) converge more slowly than those orthogonal to it (b,d). In contrast for Pauli tomography, representing an orthogonal basis, best results are obtained for states that are aligned with the basis, as seen in Fig.~\ref{figS8:BasisOverlap}(a-b). 

In case of local states, measurements are uncorrelated and the above state dependent convergence is to be expected. The situation might change when moving to entangled states. In  Figs.~\ref{figS8:BasisOverlap}(e-f) we apply controlled-phase gates states along the SIC vectors from (c,d) to generate an entangled state that is still (in a sense) aligned with the SIC basis. The resulting convergence is similar for both tomography methods. Curiously, SIC tomography performs equally well as for special states like a GHZ state (Fig.~\ref{figS8:BasisOverlap}(g)) or a more generic rotated GHZ state (Fig.~\ref{figS8:BasisOverlap}(h)). In contrast, Pauli tomography, which also shows little dependence on the state once entanglement is involved, does outperform for the GHZ state (g), which is perfectly aligned with the Pauli basis. Importantly, the rotated GHZ state, which is somewhat randomly aligned with both SIC POVM and Pauli basis (h) shows no difference between the tomography methods. The same rotated state, yet on 8 qubits, was utilized for the real-time analysis in main text Fig.~\ref{fig4:Live8qGHZrotated} to make sure the comparison does not favour any approach.

\begin{figure*}[ht]
    \centering
    \includegraphics[width=0.7\textwidth]{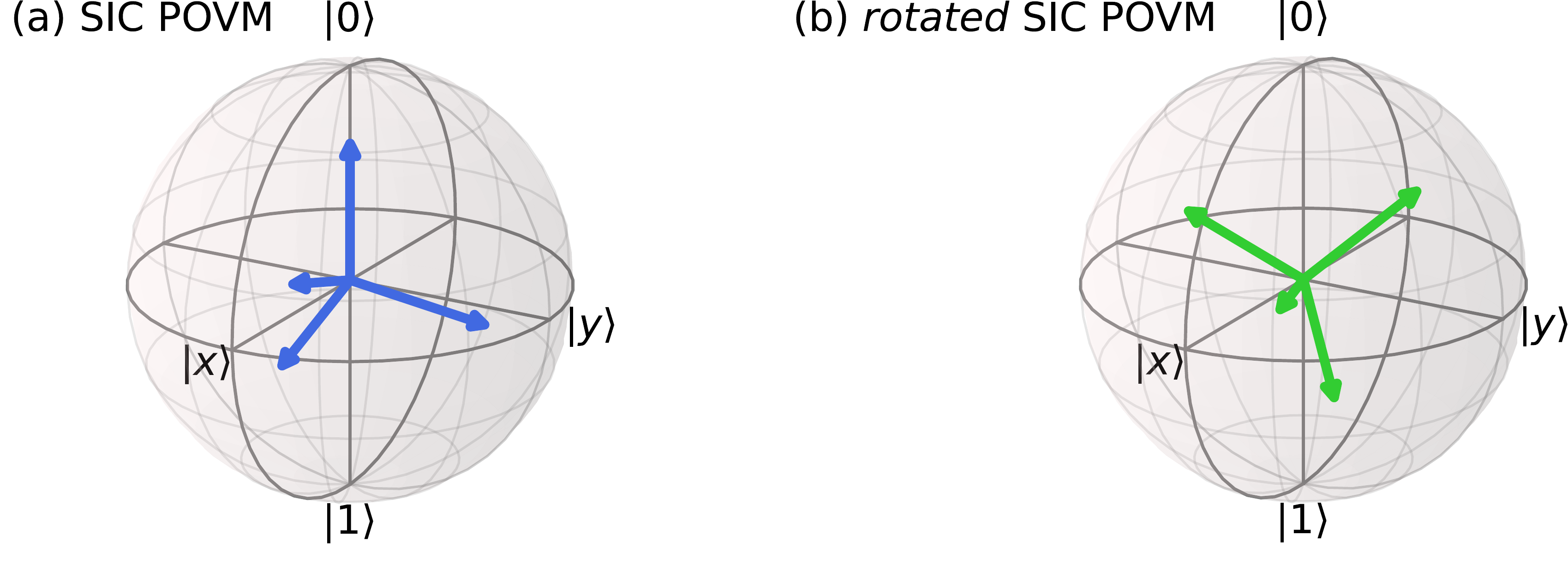}
    \caption{\textbf{Various representations of SIC POVM for maximizing information gain} \textbf{(a)} Standard representation of SIC POVM as used throughout this manuscript on both experiments and numerical simulations. \textbf{(b)} Rotated SIC optimized for predicting Pauli observables.}
    \label{figS9:SICrotated}
\end{figure*}

Inspired by these findings, it can be beneficial to rotate the SIC POVM used throughout this manuscript from Fig.~\ref{figS9:SICrotated}(a) to favour the particular state or application that it is used for. Particularly, the alignment depicted in Fig.~\ref{figS9:SICrotated}(b), which has equal overlap with every Pauli basis vector, leads to improved prediction of Pauli observables. Intuitively, this alignment ensures that each shot contains equal information about every Pauli observable. Geometrically, imagine a triangle spanned by the three Pauli basis vectors in positive direction. Then, the orthogonal state to the first SIC-vector ($\lbrace-1/\sqrt{3},-1/\sqrt{3},-1/\sqrt{3}\rbrace$) perpendicularly intersects this triangle area through its center. The front area of the SIC tetrahedron orientates parallel with the triangle. Thus, favouring Pauli eigenstates or stabilizer states, for which we find plenty of applications across the entire field of quantum computation and quantum information. Particularly, VQE applications, which rely on the efficient estimation of many Pauli observables could benefit from this choice of SIC POVM. 

Note that, experiments on rotated SICs can straightforwardly be realized by changing the local mapping sequence from Fig.~\ref{fig1:TomographySchematic}(b) and do not add further complexity to the implementation.

\section{SIC POVM-based classical shadow framework}
On the contrary, SIC POVMs were originally discovered, because of their exceptional tomographic capabilities~\cite{Renes2004SIC}. Ever since, both (tensor products of $N$) single-qubit SIC POVMs and global $2^N$-dimensional SIC POVMs have served as idealized measurements for state reconstruction tasks. Single-qubit SIC POVMs, also known as tetrahedral POVMs, have also been used to acquire training data for neural network quantum state tomography~\cite{Torlai2018,Carrasquilla2019}. 

Geometrically speaking, SIC POVMs~\cite{Renes2004SIC} and the overcomplete Pauli basis~\cite{Klappenecker2005} both form complex projective 2-designs (the single-qubit Pauli basis is actually a 3-design~\cite{Kueng2015,Zhu2017,Webb2016}).
Roughly speaking, this means that the first two (three) moments exactly reproduce the moments of uniformly (Haar) random states. As detailed below, closed form expression for Haar-random moments can then be used to compute measurement operators and estimators analytically. For Pauli basis measurements, this observation culminated in efficient PLS estimators for full state tomography~\cite{Guta2020}, as well as the classical shadow formalism for directly predicting (non-)linear properties of the underlying state~\cite{Huang2020,elben2020}. Subsequently, some of these ideas have been extended to (single-qubit) SIC POVM measurements. Ref.~\cite{Maniscalco2021}, in particular, highlights that SIC POVM measurements can outperform Pauli basis measurements in VQE-type energy measurements and, more generally, for predicting linear state properties. 

Here, we build on all these ideas and provide a self-contained derivation of classical shadows from (single-qubit) SIC POVM measurements, as well as rigorous sample complexity bounds for general linear and quadratic property estimation. The actual definition of SIC POVM shadows is virtually identical to existing (Pauli basis) classical shadows, because both form complex projective 2-designs. Sample complexity bounds, on the other hand, require novel proof techniques. They require computing variances which correspond to 3rd order polynomials in the measurement ensemble. This is comparatively easy for Pauli basis measurements, which form a 3-design. But for SIC POVMs --- which only form a 2-design --- these existing techniques don't apply. We overcome this drawback with new proof methods that directly use the symmetries within a SIC POVM rather than an abstract 3-design property. To our knowledge, these theoretical arguments are novel and may be of independent interest. 

\subsection{Classical shadows from SIC POVM measurements}

\subsubsection{The single-qubit case}
Single-qubit density matrices live in the (real-valued) vector space of Hermitian $2 \times 2$ matrices which we denote by $\mathbb{H}_2$. On this space,
single-qubit SIC POVMs are known to form so-called projective 2-designs~\cite{Renes2004SIC}. Mathematically, this means that discrete averages over (outer products of) SIC-vectors reproduce uniform averages over all possible pure states up to second moments~\cite{Ambainis2007}. This is captured by the following two averaging formulas:
\begin{align}
\frac{1}{4} \sum_{i=1}^4 |\psi_i \rangle \! \langle \psi_i| =& \int |v \rangle \! \langle v| \mathrm{d}v = \frac{1}{2} \mathbb{I} \in \mathbb{H}_2, \label{eq:1design}\\
\frac{1}{4} \sum_{i=1}^4 \left( |\psi_i \rangle \! \langle \psi_i|\right)^{\otimes 2} =& \int \left( |v \rangle \! \langle v| \right)^{\otimes 2} \mathrm{d} v = \frac{1}{6} \left( \mathbb{I} \otimes \mathbb{I} + \mathbb{F} \right) \label{eq:2design}.
\end{align}
Here, $\mathrm{d}v$ denotes the unique pure state measure (normalized to $\int \mathrm{d}v=1$) that assigns the same infinitesimal weight to each state ($\mathrm{d}v = \mathrm{d}w$). 
The two qubit swap operator operator $\mathbb{F} |v \rangle \otimes |w \rangle = |w \rangle \otimes |v \rangle$ acts by permuting tensor factors. 

The first averaging formula~\eqref{eq:1design} confirms that the collection $\left\{ \frac{1}{2} |\psi_i \rangle \! \langle \psi_i|: i=1,2,3,4\right\} \subset \mathbb{H}_2$ forms a valid quantum measurement for single-qubit systems (POVM). Let $\rho \in \mathbb{H}_2$ be a density matrix, i.e.\ a Hermitian matrix with unit trace whose eigenvalues are non-negative. Then,
\begin{equation*}
\mathrm{Pr} \left[ i| \rho \right] = \mathrm{tr} \left( \frac{1}{2} |\psi_i \rangle \! \langle \psi_i| \rho \right) = \frac{1}{2} \langle \psi_i | \rho |\psi_i \rangle \quad \text{for $i =1,2,3,4$}
\end{equation*}
obeys $\mathrm{Pr} \left[i |\rho \right] \geq 0$, because density matrices don't have negative eigenvalues ($\langle \psi_i |\rho |\psi_i \rangle \geq 0$). Moreover, Eq.~\eqref{eq:1design} ensures proper normalization:
\begin{equation*}
\sum_{i=1}^4 \mathrm{Pr} \left[ i| \rho \right] = \mathrm{tr} \left(2  \left( \frac{1}{4} \sum_{i=1}^4 |\psi_i \rangle \! \langle \psi_i| \right) \rho \right) = \mathrm{tr} \left( \mathbb{I} \rho \right) = \mathrm{tr}(\rho)=1.
\end{equation*} 
The second averaging property~\eqref{eq:2design} is more interesting. It ensures that SIC POVM measurements are informationally complete, i.e. we can reconstruct every density matrix $\rho$ based on outcome probabilities. There are many ways to establish this property. Here, we choose one that is based on the following observation. If we weigh SIC projectors $|\psi_i \rangle \! \langle \psi_i|$ with the probability $\mathrm{Pr} \left[i |\rho \right]$ of observing this outcome, Eq.~\eqref{eq:2design} allows us to compute the resulting average. Let $\mathrm{tr}_1 (\cdot)$ denote the partial trace over the first of two qubits ($\mathrm{tr}_1(A \otimes B) = \mathrm{tr}(A) B$ and linearly extended to all of $\mathbb{H}_2^{\otimes 2}$). Then,
\begin{align}
\sum_{i=1}^4 \mathrm{Pr} \left[ i| \rho \right] |\psi_i \rangle \! \langle \psi_i| =& \sum_{i=1}^4 \frac{1}{2} \langle \psi_i| \rho |\psi_i \rangle |\psi_i \rangle \! \langle \psi_i| 
= 2 \mathrm{tr}_1 \left( \left(\frac{1}{4} \sum_{i=1}^4 \left( |\psi_i \rangle \! \langle \psi_i|\right)^{\otimes 2} \right) \rho \otimes \mathbb{I} \right) \nonumber \\
=& \frac{1}{3} \mathrm{tr}_1 \left( \left( \mathbb{I} \otimes \mathbb{I} + \mathbb{F} \right) \mathbb{I} \otimes \rho \right)
= \frac{1}{3} \left( \mathrm{tr}(\rho) \mathbb{I} + \rho \right), \label{eq:depolarizing-average}
\end{align}
where the last equation follows from the interplay between partial trace and swap operator. The final expression is equivalent to applying a depolarizing channel with parameter $p=1/3$ to the quantum state in question:
\begin{equation*}
\mathcal{D}_{1/3} \left( \rho \right) = \frac{1}{3} \rho + \left( 1- \frac{1}{3} \right) \frac{\mathrm{tr}(\rho)}{2} \mathbb{I} \in \mathbb{H}_2.
\end{equation*}
Viewed as a linear map on $\mathbb{H}_2$, this channel has a uniquely defined inverse:
\begin{equation}
\mathcal{D}_{1/3}^{-1}(A) = 3 A - \mathrm{tr}(A) \mathbb{I} \quad \text{for all $A \in \mathbb{H}_2$.}
\label{eq:depolarizing-channel}
\end{equation}
Although a linear map, this is not a physical operation. We can, however, use it in the classical post-processing stage to counterbalance the effect of averaging over SIC elements.
Indeed, linearity and Eq.~\eqref{eq:depolarizing-average} ensure
\begin{align}
\sum_{i=1}^4 \mathrm{Pr} \left[ i| \rho \right] \left( 3|\psi_i \rangle \! \langle \psi_i| - \mathbb{I} \right)
=& \sum_{i=1}^4 \mathrm{Pr} \left[ i| \rho \right] \mathcal{D}_{1/3}^{-1} \left( |\psi_i \rangle \! \langle \psi_i|\right)
= \mathcal{D}_{1/3}^{-1} \left( \sum_{i=1}^4 \mathrm{Pr} \left[ i|\rho \right] |\psi_i \rangle \! \langle \psi_i| \right) = \mathcal{D}_{1/3}^{-1} \left( \mathcal{D}_{1/3} (\rho) \right) = \rho.
\label{eq:single-qubit-lin-inv}
\end{align}
The left hand side of this display features a linear combination involving SIC outcome probabilities $\mathrm{Pr}\left[i|\rho\right]$, while the right hand side exactly reproduces the underlying state $\rho$. 
This equips us with a concrete state reconstruction formula --- the so-called linear inversion estimator. But, at least at first sight, this formula is only useful if we have precise knowledge of the SIC outcome probabilities $\mathrm{Pr} \left[ i | \rho \right]$. And, with current quantum technology, these probabilities must be estimated from repeatedly performing SIC POVM measurements on independent copies of $\rho$ and approximating these probabilities by frequencies. 

The classical shadow formalism provides an alternative perspective on this estimation process. 
Suppose that we perform a single SIC POVM measurement of an unknown quantum state $\rho$ (single shot). Then, we obtain a random measurement outcome $\hat{i} \in \left\{1,2,3,4\right\}$ with probability $\mathrm{Pr} \left[ \hat{i} |\rho \right]$ each. Inspired by the left hand side of Eq.~\eqref{eq:single-qubit-lin-inv}, we can use this outcome $\hat{i}$ to construct a Monte Carlo estimator of $\rho$:
\begin{align*}
\hat{i} \mapsto \hat{\sigma} = \left( 3 |\psi_{\hat{i}} \rangle \! \langle \hat{\psi}_{\hat{i}}| - \mathbb{I} \right) = \mathcal{D}_{1/3}^{-1} \left( |\psi_{\hat{i}} \rangle \! \langle \psi_{\hat{i}} |\right) \in \mathbb{H}_2.
\end{align*}
This is a \emph{random} $2 \times 2$ matrix that can assume 4 different forms --- one for each possible outcome $\hat{i} \in \left\{1,2,3,4\right\}$. It does exactly reproduce the underlying quantum state $\rho$ in expectation over the observed single-shot outcome:
\begin{equation}
\mathbb{E} \left[ \hat{\sigma} \right] = \sum_{i=1}^4 \mathrm{Pr} \left[ i| \rho \right] \left( 3 |\psi_i \rangle \! \langle \psi_i| - \mathbb{I} \right) = \rho. \label{eq:single-qubit-expectation}
\end{equation}
It is worthwhile to emphasize that each $\hat{\sigma}$ has the same eigenvalue structure: $\lambda_+ =2$ and $\lambda_- =-1$. In turn, these random matrices all have unit trace ($\mathrm{tr}(\hat{\sigma})=2-1=1$), but are unphysical in the sense that one eigenvalue is always negative. Eq.~\eqref{eq:single-qubit-expectation} represents a physical density matrix $\rho$ as the expectation of $4$ unphysical estimators. This desired expectation value can be approximated by empirically averaging $M$ independently generated Monte-Carlo estimators. Suppose that $\hat{\sigma}_1,\ldots,\hat{\sigma}_M$ are $M$ \textit{iid} (independently and identically distributed) Monte Carlo estimators. Then, their empirical average obeys
\begin{equation*}
\hat{\rho} = \frac{1}{M} \sum_{m=1}^M \hat{\sigma}_m \quad \overset{M \to \infty}{\longrightarrow} \quad \mathbb{E} \left[ \hat{\sigma} \right] =  \rho, 
\end{equation*}
and the rate of convergence can be controlled with arguments from probability theory. This will be the content of the next two subsections. 

\subsubsection{Extension to multi-qubit systems}
\label{app:multi-qubit-systems}
The formalism and ideas presented above readily extend to quantum systems comprised multiple qubits. Let $\rho \in \mathbb{H}_2^{\otimes N} \simeq \mathbb{H}^{2^N}$ be a $N$-qubit density matrix.
We can perform a single-qubit SIC POVM measurement on each of the $N$ qubits. As in the single qubit case, each such measurement yields one out of four possible outcomes. 
In total, a single-shot measurement produces a string $(i_1,\ldots,i_N)$ of outcomes. There are $4^N$ such outcomes and the probability of obtaining any one of them is given by
\begin{equation}
\mathrm{Pr} \left[ i_1,\ldots,i_N|\rho \right] = \mathrm{tr} \left( \bigotimes_{n=1}^N \left(\frac{1}{2} |\psi_{i_n} \rangle \! \langle \psi_{i_n}| \right) \rho \right)
= \frac{1}{2^N} \langle \psi_{i_1},\ldots,\psi_{i_N} | \rho |\psi_{i_1},\ldots,\psi_{i_N} \rangle \quad \text{for each $i_1,\ldots,i_N \in \left\{1,2,3,4\right\}$}.
\label{eq:measurement}
\end{equation}
Here, we have introduced the short-hand notation $|\psi_{i_1},\ldots,\psi_{i_N} \rangle = \bigotimes_{n=1}^N |\psi_{i_n} \rangle \in \left(\mathbb{C}^2\right)^{\otimes N} \simeq \mathbb{C}^{2^N}$. 
These expressions are non-negative, because the $N$-qubit density matrix does not have negative eigenvalues. Eq.~\eqref{eq:1design}, applied to each qubit separately, moreover ensures proper normalization:
\begin{align*}
\sum_{i_1,\ldots,i_N=1}^4 \mathrm{Pr} \left[ i_1,\ldots,i_N|\rho \right] = \mathrm{tr} \left( \bigotimes_{n=1}^N \left( \sum_{i_n=1}^4 \frac{1}{2} |\psi_{i_n} \rangle \! \langle \psi_{i_n} | \right) \rho \right)
= \mathrm{tr} \left( \mathbb{I}^{\otimes N} \rho \right) = \mathrm{tr} \left( \rho \right) =1. 
\end{align*}
Again, the second averaging property is more interesting: for single-qubit density matrices $\tilde{\rho} \in \mathbb{H}_2$, we already know that Eq.~\eqref{eq:2design} implies $\sum_{i=1}^4 \frac{1}{2} \langle \psi_i | \rho |\psi_i \rangle = \mathcal{D}_{1/3} (\rho)$, where $\mathcal{D}_{1/3}:\mathbb{H}_2 \to \mathbb{H}_2$ is the depolarizing channel from Eq.~\eqref{eq:depolarizing-channel}. 
It is now easy to check that this equation extends to general Hermitian $2 \times 2$ matrices: 
\begin{equation}
\sum_{i=1}^4 \frac{1}{4}\langle \psi_i | A |\psi_i \rangle |\psi_i \rangle \! \langle \psi_i| = \mathcal{D}_{1/3} (A)\quad \text{for all $A \in \mathbb{H}_2$}. \label{eq:multi-qubit-aux1}
\end{equation}
We can use this observation to show that $N$ single-qubit SIC POVM measurements are tomographically complete. To achieve this it is helpful to first decompose $\rho \in \mathbb{H}_2^{\otimes N}$ into a sum of elementary tensor products:
\begin{equation*}
\rho = \sum_{W_1,\ldots,W_N} r(W_1,\ldots,W_N) \bigotimes_{n=1}^N W_n \quad \text{where} \quad r(W_1,\ldots,W_N) = \mathrm{tr} \left( \bigotimes_{n=1}^N W_n \; \rho \right) \in \left[-1,1\right].
\end{equation*}
and the summation goes over all four single-qubit Pauli matrices $W_n = \mathbb{I},X,Y,Z$. Combine this with Eq.~\eqref{eq:multi-qubit-aux1} to compute
\begin{align}
\sum_{i_1,\ldots,i_N=1}^4 \mathrm{Pr} \left[ i_1,\ldots,i_N| \rho \right] \bigotimes_{n=1}^N |\psi_{i_n} \rangle \! \langle \psi_{i_n} | 
=& \sum_{W_1,\ldots,W_N} r(W_1,\ldots,W_N) \bigotimes_{n=1}^N \left( \frac{1}{2} \sum_{i_n=1}^4 \langle \psi_{i_n} | W_n |\psi_{i_n} \rangle |\psi_{i_n} \rangle \! \langle \psi_{i_n} | \right) \nonumber \\
=& \sum_{W_1,\ldots,W_N} r(W_1,\ldots,W_N) \bigotimes_{n=1}^N \mathcal{D}_{1/3} \left( W_n \right) 
= \mathcal{D}_{1/3}^{\otimes N} \left( \rho \right), \label{eq:multi-qubit-aux2}
\end{align}
where the last equality follows from linearity of depolarizing channels. The final expression is equivalent to applying $N$ independent, single-qubit depolarizing channels to the $N$-qubit quantum state $\rho$. 
Viewed as a linear map on $\mathbb{H}_2^{\otimes N}$, this tensor product channel has a uniquely defined inverse $\mathcal{D}_{1/3}^{-\otimes N}: \mathbb{H}_2^{\otimes N} \to \mathbb{H}_2^{\otimes N}$. 
For elementary tensor products, this tensor product of inverse depolarizing channels factorizes nicely into tensor products. In particular, 
\begin{equation*}
\mathcal{D}_{1/3}^{-\otimes N} \left( \bigotimes_{n=1}^N |\psi_{i_n} \rangle \! \langle \psi_{i_n} | \right) = \bigotimes_{n=1}^N \mathcal{D}_{1/3}^{-1} \left( |\psi_{i_n} \rangle \! \langle \psi_{i_n} | \right) = \bigotimes_{n=1}^N \left( 3 |\psi_{i_n} \rangle \! \langle \psi_{i_n} | - \mathbb{I} \right).
\end{equation*}
Again, this is not a physical operation, because it produces matrices with negative eigenvalues. However, we can nonetheless use it in the classical post-processing stage to counterbalance the $N$-qubit averaging effect encountered in Eq.~\eqref{eq:multi-qubit-aux2}:
\begin{align}
\sum_{i_1,\ldots,i_N=1}^4 \mathrm{Pr} \left[ i_1,\ldots,i_N|\rho \right] \bigotimes_{n=1}^N \left( 3 |\psi_{i_n} \rangle \! \langle \psi_{i_n} | - \mathbb{I} \right)
= \mathcal{D}_{1/3}^{-\otimes N} \left( \sum_{i_1,\ldots,i_N=1}^4 \mathrm{Pr} \left[ i_1,\ldots,i_N|\rho \right] \bigotimes_{n=1}^N |\psi_{i_n} \rangle \! \langle \psi_{i_n} |\right)
= \mathcal{D}_{1/3}^{-\otimes N} \left( \mathcal{D}_{1/3}^{\otimes N} \left( \rho \right) \right) = \rho.
\label{eq:multi-qubit-lin-inv}
\end{align}
The left hand side of this display features a linear combination involving single-qubit SIC outcome probabilities $\mathrm{Pr} \left[ i_1,\ldots,i_N| \rho \right]$, while the right hand side exactly reproduces the underlying $N$-qubit density matrix. This provides us with a concrete reconstruction formula for arbitrary $N$-qubit states. In fact, it is a natural and relatively straightforward extension of the single-qubit linear inversion estimator~\eqref{eq:single-qubit-lin-inv} to $N$ qubits. 

As was the case for single qubits, the classical shadow formalism provides an alternative perspective on such a linear inversion estimation process. 
Suppose that we perform $N$ single-qubit SIC POVM measurements of an unknown $N$-qubit state $\rho$ (single shot). Then, we obtain a random outcome string $(\hat{i}_1,\ldots,\hat{i}_N) \in \left\{1,2,3,4\right\}^{\times N}$ with probability $\mathrm{Pr} \left[ \hat{i}_1,\ldots,\hat{i}_N|\rho \right]$ each. Inspired by the left hand side of Eq.~\eqref{eq:multi-qubit-lin-inv}, we can use this random outcome string to construct a Monte Carlo estimator of $\rho$:
\begin{align}
(\hat{i}_1,\ldots,\hat{i}_N) \mapsto \hat{\sigma} = \bigotimes_{n=1}^N \left( 3 |\psi_{\hat{i}_n} \rangle \! \langle \psi_{\hat{i}_n}| - \mathbb{I} \right) \in \mathbb{H}_2^{\otimes N}. \label{eq:classical-shadow}
\end{align}
This is a random $2^N \times 2^N$ matrix that decomposes nicely into tensor products of single-qubit contributions. Each tensor factor contributes a matrix with eigenvalue $\lambda_{n,+}=+2$ and $\lambda_{n,-}=-1$. The eigenvalues of the tensor product $\hat{\sigma}$ then correspond to $N$-fold products of these two possible numbers. The largest eigenvalue is $+2^N$, the smallest is $-2^{N-1}$, so $\hat{\sigma}$ is a very unphysical random matrix. The randomness stems from an actual quantum measurement and depends on the underlying quantum state. This ensures that $\hat{\sigma}$ reproduces $\rho$ in expectation:
\begin{align}
\mathbb{E} \left[ \hat{\sigma} \right]= \sum_{i_1,\ldots,i_N=1}^4 \mathrm{Pr} \left[ i_1,\ldots,i_N |\rho \right] \bigotimes_{n=1}^N \left( 3|\psi_{i_n} \rangle \! \langle \psi_{i_n} | - \mathbb{I} \right)
= \rho,
\label{eq:shadow-expectation}
\end{align}
courtesy of Eq.~\eqref{eq:multi-qubit-lin-inv}. This expectation value can now be approximated by empirically averaging $M$ independently generated Monte-Carlo estimators, so called \emph{classical shadows}. Let $\hat{\sigma}_1,\ldots,\hat{\sigma}_M \in \mathbb{H}_2^{\otimes N}$ be estimators generated from repeatedly preparing $\rho$ and performing single-qubit SIC POVM measurements (\textit{iid}). Then, their empirical average obeys
\begin{equation}
\hat{\rho} := \frac{1}{M} \sum_{m=1}^M \hat{\sigma}_m \quad \overset{M \to \infty}{\longrightarrow} \rho,
\label{eq:shadow-convergence}
\end{equation}
in full analogy to the single-qubit case. As detailed in the next section, the rate of convergence will depend on the number of qubits $N$. The larger the space, the longer it takes for convergence to kick in, and this scaling can become unfavorable. It is therefore worthwhile to emphasize another distinct advantage of the tensor product structure of the estimators~\eqref{eq:classical-shadow}: marginalization to subsystem density operators is straightforward. Let $\rho$ be an $N$-qubit state, but suppose we are only interested in a subsystem $\mathsf{K} \subset \left[N\right]=\left\{1,\ldots,N\right\}$ comprised of only $|\mathsf{K}| \leq N$ qubits. Such a subsystem is fully described by the reduced $|\mathsf{K}|$-qubit density matrix $\rho_\mathsf{K} = \mathrm{tr}_{\neg \mathsf{K}} (\rho) \in \mathbb{H}_2^{\otimes \mathsf{K}}$ that results from tracing out all qubits not in $\mathsf{K}$. This partial trace is a linear operation that plays nicely with the tensor product structure in Eq.~\eqref{eq:classical-shadow}. Each tensor product factor has unit trace, which ensures that
\begin{align*}
\hat{\sigma}_\mathsf{K} = \mathrm{tr}_{\neg \mathsf{K}} \left( \hat{\sigma}\right) = \bigotimes_{k \in \mathsf{K}} \left( 3|\psi_{\hat{i}_k} \rangle \! \langle \psi_{\hat{i}_k}|-\mathbb{I} \right) \in \mathbb{H}_2^{\otimes |\mathsf{K}|} \quad \text{obeys} \quad \mathbb{E} \left[ \hat{\sigma}_\mathsf{K}\right] = \mathrm{tr}_{\neg \mathsf{K}} \left( \mathbb{E} \left[ \hat{\sigma} \right] \right) = \mathrm{tr}_{\neg \mathsf{K}} (\rho) = \rho_\mathsf{K}.
\end{align*}
The object at the very left is a random $2^{|\mathsf{K}|} \times 2^{|\mathsf{K}|}$ matrix that can be generated from performing a complete $N$-qubit SIC POVM measurement to obtain outcomes $(\hat{i}_1,\ldots,\hat{i}_N )\in \left\{1,\ldots,4\right\}^{\times N}$. Subsequently, we only use outcomes that correspond to qubits in $K$ to directly construct an estimator for the subsystem density matrix $\rho_K$ in question. 
This trick reduces the question of convergence to a problem that only involves $|\mathsf{K}|$ qubits, not $N$. This can be highly advantageous if $|\mathsf{K}| \ll N$.
What is more, we can use the same $N$-qubit measurement outcome $(\hat{i}_1,\ldots,\hat{i}_N)^{\times N}$ to construct estimators for multiple subsystems $\mathsf{K}_1,\ldots,\mathsf{K}_L \subset \left[N\right]$ at once.
This allows us to use the same $N$-qubit measurement statistics to estimate many subsystem properties in parallel.

\subsection{Convergence for predicting linear observables}
\label{app:linearobservables}
Suppose that we have access to $M$ independent Monte Carlo approximations $\hat{\sigma}_1,\ldots,\hat{\sigma}_M$ of an unknown $N$-qubit state $\rho$. Each of them arises from measuring $N$ single-qubit SIC POVMs on an independent copy of $\rho$. We can then use these approximations to estimate observable expectation values $\mathrm{tr}(O \rho)$ with $O \in \mathbb{H}_2^{\otimes N}$:
\begin{align*}
\hat{o} = \mathrm{tr} \left(O \hat{\rho}  \right) = \frac{1}{M} \sum_{m=1}^M \mathrm{tr} \left( O \hat{\sigma}_m \right).
\end{align*}
This is an empirical average of $M$ \textit{iid} random numbers $X_1,\ldots,X_M \overset{\textit{iid}}{\sim} X = \mathrm{tr} \left( O \hat{\sigma} \right)$ that converges to the true expectation $\mathbb{E} \left[ X \right] = \mathrm{tr} \left( O \mathbb{E} \left[ \hat{\sigma}_m \right] \right) = \mathrm{tr}(O \rho)$ as $M$ increases. The rate of convergence is controlled by the variance $\mathrm{Var}[X]$. 

\begin{lemma} \label{lem:linear-variance}
Fix an $N$-qubit observable $O$ and let $\hat{\sigma} \in \mathbb{H}_2^{\otimes N}$ be a (SIC POVM) classical shadow as defined in Eq.~\eqref{eq:classical-shadow}.
Then,
\begin{equation*}
\mathrm{Var} \left[ \mathrm{tr} \left( O \hat{\sigma} \right) \right] \leq 3^N \mathrm{tr}\left(O^2 \right) \quad \text{for any underlying $N$-qubit state $\rho$.}
\end{equation*}
\end{lemma}

This bound is reminiscent of existing variance bounds for randomized single-qubit Pauli measurements~\cite{Huang2020}, but slightly weaker (random Pauli basis measurements achieve $\mathrm{Var} \left[ \mathrm{tr}(O \hat{\sigma})\right] \leq 2^N \mathrm{tr}(O^2)$). The proof exploits the 2-design property of SIC POVM measurements and will be supplied in Sec.~\ref{sub:linear-technical} below. 
For now, we use this variance bound to conclude strong convergence guarantees for observable estimation with SIC POVM measurements. 
The key ingredient is a strong tail bound for sums of \textit{iid} random variables with known variance and bounded magnitude. 
The \emph{Bernstein inequality} is a stronger version of the better known Hoeffding inequality, see e.g.\ \cite[Corollary~7.31]{Foucart2013} or \cite[Theorem~2.8.4]{Vershynin2018}: \\ 
Let $X_1,\ldots,X_M \overset{\textit{iid}}{\sim} X \in \mathbb{R}$ be \textit{iid} random variables with expectation $\mu = \mathbb{E} \left[ X \right]$ and variance $\sigma^2 = \mathbb{E} \left[(X-\mu)^2 \right]$ that also obey $\left|X_m \right| \leq R$ (almost surely).
Then, for $\epsilon >0$
\begin{equation}
\mathrm{Pr} \left[ \left| \frac{1}{M} \sum_{m=1}^M X_m - \mu \right| \geq \epsilon \right] \leq 2 \exp \left( - \frac{M\epsilon^2/2}{\sigma^2+R\epsilon/3} \right) \leq
\begin{cases}
2 \exp \left( - \tfrac{3}{8} M \epsilon^2/\sigma^2 \right) & \text{if $\epsilon \leq \sigma^2/R$}, \\
2 \exp \left( - \tfrac{3}{8} M \epsilon/R \right) & \text{if $\epsilon \geq \sigma^2/R$}.
\end{cases}
\label{eq:bernstein}
\end{equation}
For the task at hand, we write $\hat{o} = \frac{1}{M} \sum_{m=1}^M \mathrm{tr} \left( O \hat{\sigma}_m\right) = \frac{1}{M} \sum_{m=1}^M X_m$, where $X_m = \mathrm{tr} \left( O \hat{\sigma}_m \right) \overset{\textit{iid}}{\sim}X= \mathrm{tr} \left( O \hat{\sigma} \right)$. For technical reasons, we also assume $\mathrm{tr}(O^2) \geq (5/9)^N > 2^{-N}$. Note that this is achieved by (i) physical observables ($\mathrm{tr}(O^2) \geq \|O\|_\infty^2 =1$), as well as quantum states ($O=\rho$ obeys $\mathrm{tr}(\rho^2)\geq \mathrm{tr}((\mathbb{I}/2^N)^2)=2^{-N}$).
The random variable $X$ obeys
\begin{align*}
\mu = \mathrm{tr} \left( O \mathbb{E} \left[ \hat{\sigma} \right] \right) = \mathrm{tr} \left( O \rho \right), \quad
\mathrm{Var} \left[ \mathrm{tr} \left( O \hat{\sigma} \right) \right] \leq 3^N \mathrm{tr}(O^2)=:\sigma^2 \quad \text{and} \quad  \left| \mathrm{tr} \left( O \hat{\sigma} \right) \right| \leq 5^{N/2} \sqrt{\mathrm{tr}(O^2)} \leq 3^N \mathrm{tr}(O^2)=:R.
\end{align*}
The first equality is Eq.~\eqref{eq:shadow-expectation}, the second bound is Lemma~\ref{lem:linear-variance} and the last bound is a consequence of the Cauchy-Schwarz inequality: $\left| \mathrm{tr}(O \hat{\sigma}) \right| \leq \sqrt{\mathrm{tr}(O^2)} \sqrt{\mathrm{tr}(\hat{\sigma}^2)}$.
Since classical shadows are tensor products of single-qubit blocks with eigenvalues $\lambda_+=2,\lambda_-=-1$, we can readily conclude $\mathrm{tr}\left( \hat{\sigma}^2 \right) = \left( \lambda_+^2 + \lambda_-^2 \right)^N = 5^N$. The final bound is rather loose and follows from the assumption $\mathrm{tr}(O^2) \geq (5/9)^N$.
Inserting these bounds into Eq.~\eqref{eq:bernstein} now ensures
\begin{equation}
\mathrm{Pr} \left[ \left| \hat{o}-\mathrm{tr}(O \rho ) \right| \geq \epsilon \right] =\mathrm{Pr} \left[ \left| \hat{o} - \mathrm{tr}(O \rho) \right| \geq \epsilon \right] \leq 2 \exp \left( - \frac{3M \epsilon^2}{8 \times 3^N \mathrm{tr}(O^2)}\right) \quad \text{for all $0 <\epsilon \leq 1$.}
\end{equation}
This is a bound on the probability of an $\epsilon$-deviation (or more) that diminishes exponentially in the number of Monte Carlo samples (measurements) $M$. For a fixed confidence $\delta \in (0,1)$, setting
\begin{equation*}
M \geq \tfrac{8}{3} 3^N \mathrm{tr}(O^2) \log(1/\delta)/\epsilon^2 \quad \text{ensures} \quad \left| \hat{o}-\mathrm{tr}(O \rho) \right| \leq \epsilon \quad \text{with probability (at least) $1-\delta$.}
\end{equation*}

Note that the required measurement budget $M$ scales exponentially in the number $N$ of involved qubits. At this point it is helpful to remember the marginalization property of classical shadows. 
Suppose that $O$ is localized in the sense that it only affects a subsystem $\mathsf{K} \subset \left[N\right]$ comprised of $|\mathsf{K}| \leq N$ qubits. 
Then, $\mathrm{tr} \left( O \rho \right) = \mathrm{tr} \left( O_{\mathsf{K}} \rho_{\mathsf{K}} \right)$, where $\rho_{\mathsf{K}}=\mathrm{tr}_{\neg \mathsf{K}}(\rho)$ is the reduced $|\mathsf{K}|$-qubit density matrix and $O_{\mathsf{K}}$ is the nontrivial part of $O$. In turn, we can use appropriately marginalized classical shadows $\hat{\sigma}_{m,\mathsf{K}}=\mathrm{tr}_{\neg \mathsf{K}} (\rho)$ to directly approximate this subsystem property: 
\begin{equation*}
\hat{o} = \frac{1}{M} \sum_{m=1}^M \mathrm{tr} \left( O \hat{\sigma}_m\right)= \frac{1}{M} \sum_{m=1}^M \mathrm{tr} \left( O_{\mathsf{K}} \hat{\sigma}_{m,\mathsf{K}}\right).
\end{equation*}
The reformulation on the right hand side now only involves the $|\mathsf{K}| < N$ relevant qubits. We can now can re-do the argument from above to obtain a measurement budget that 
only scales exponentially in $|\mathsf{K}|$:
\begin{equation}
\mathrm{Pr} \left[ \left| \hat{o}-\mathrm{tr}(O \rho ) \right| \geq \epsilon \right]
\leq 2 \exp \left( - \frac{3 M \epsilon^2}{8 \times 3^{|\mathsf{K}|} \mathrm{tr} \left( O_{\mathsf{K}}^2 \right)}\right) \quad \text{for $\epsilon \in (0,1)$}.
\label{eq:tail-bound}
\end{equation}
This refinement asserts that the probability of an $\epsilon$-deviation for a single observable estimation diminishes exponentially in the number of measurements. We can use this exponential concentration to bound the probability of a single deviation among many. This allows us to use the same measurement data to predict many observables $\mathrm{tr} \left( O_1 \rho \right),\ldots,\mathrm{tr} \left( O_L \rho \right)$ in parallel.

\begin{theorem} \label{thm:linear-estimation}
Let $O_1,\ldots,O_L$ be $N$-qubit observables that are all localized to (at most) $K$ qubits and fix $\epsilon,\delta \in (0,1)$. Then,
\begin{equation}
M \geq \tfrac{8}{3} 3^K \max_{1 \leq l \leq L} \mathrm{tr} \left( O_{l,\mathsf{K}_l}^2 \right) \log (2L/\delta)/\epsilon^2 \label{eq:observable-measurement-budget}
\end{equation}
$N$-qubit SIC POVM measurements of an unknown state $\rho$ are very likely to $\epsilon$-approximate all observables simultaneously. More precisely, the resulting classical shadows $\hat{\sigma}_1,\ldots,\hat{\sigma}_M$ obey
\begin{equation*}
\max_{1 \leq l \leq L} \left| \frac{1}{M} \sum_{m=1}^M \mathrm{tr} \left( O_l \hat{\sigma}_m \right) - \mathrm{tr} \left( O_l \rho \right) \right| \leq \epsilon
\quad \text{with probability (at least) $1-\delta$.}
\end{equation*}
\end{theorem}

The convergence bound advertised in Eq.~\eqref{eq:main-text-observable-measurement-budget} of the main text is a simplified consequence of this result. Note that, by and large, physical observables are normalized in operator norm: $\|O_l \|_\infty = \|O_{l,\mathsf{K}} \|_\infty =1$. Eq.~\eqref{eq:observable-measurement-budget} features squared Hilbert-Schmidt norms $\|O_{l,\mathsf{K}} \|_2^2 = \mathrm{tr} \left( O_{l\mathsf{K}}^2\right)$ on the $|\mathsf{K}|$-qubit subsystems in question. These Hilbert-Schmidt norms can be related to the operator norm, which is bounded:
\begin{equation*}
\mathrm{tr} \left( O_{l,\mathsf{K}}^2\right) = \| O_{l,\mathsf{K}}\|_2^2 \leq 2^{|\mathsf{K}|} \|O_{l,\mathsf{K}}\|_\infty^2 = 2^{|\mathsf{K}|} \quad \text{for all $1 \leq l \leq L$}.
\end{equation*}
Here, we have used the fact that each $O_{l,\mathsf{K}}$ is a matrix of size (at most) $2^{|\mathsf{K}|} \cdot 2^{|\mathsf{K}|}$. This implies the bound $\max_{1\leq l \leq L}\mathrm{tr} \left( O_{l,\mathsf{K}}^2\right) \leq 2^{|\mathsf{K}|}$, which can be very pessimistic. Inserting it into Eq.~\eqref{eq:observable-measurement-budget} yields Eq.~\eqref{eq:main-text-observable-measurement-budget} in the main text.

\begin{proof}[Proof of Theorem~\ref{thm:linear-estimation}]
A maximum deviation larger than $\epsilon$ occurs if at least one individual prediction $\hat{o}_l$ is further than $\epsilon$ off from the actual target $\mathrm{tr}(O_l \rho)$. 
The union bound, also known as Boole's inequality, tells us that the probability of such a maximum deviation is upper bounded by the sum of individual deviation probabilities. These, in turn, can be controlled via the tail bound from Eq.~\eqref{eq:tail-bound}:
\begin{align*}
\mathrm{Pr} \left[\max_{1 \leq l \leq L} \left| \frac{1}{M} \sum_{m=1}^M \mathrm{tr} \left( O_l \hat{\sigma}_m \right) - \mathrm{tr} \left( O_l \rho \right) \right| \geq \epsilon \right]
\leq & \sum_{l=1}^L \mathrm{Pr} \left[  \left| \frac{1}{M} \sum_{m=1}^M \mathrm{tr} \left( O_l \hat{\sigma}_m \right) - \mathrm{tr} \left( O_l \rho \right) \right| \geq \epsilon \right] \\
\leq & \sum_{l=1}^L 2 \exp \left( - \frac{3M \epsilon^2}{8 \times 3^{|\mathsf{K}_l|} \mathrm{tr} \left( O_{l,\mathsf{K}_l}^2\right)} \right).
\end{align*}
We see that each of these summands diminishes exponentially in the measurement budget $M$. The right-hand side of  Eq.~\eqref{eq:observable-measurement-budget} ensures that each term contributes at most $\delta/L$ to this sum. Since there are $L$ summands in total, we conclude
$\mathrm{Pr} \left[\max_{1 \leq l \leq L} \left| \frac{1}{M} \sum_{m=1}^M \mathrm{tr} \left( O_l \hat{\sigma}_m \right) - \mathrm{tr} \left( O \rho \right) \right| \geq \epsilon \right]
\leq  \delta
$.
This is equivalent to the advertised display.
\end{proof}

\subsection{Convergence for predicting (subsystem) purities}
\label{app:nonlinearobservables}
Classical shadows can also be used to predict non-linear quantum state properties, see e.g.\ \cite{Huang2020,elben2020}. 
A prototypical example is the purity of an $N$-qubit density matrix $\rho$:
\begin{equation*}
p(\rho) = \mathrm{tr} \left( \rho^2 \right) = \mathrm{tr} \left( \rho \; \rho \right) \in (0,1].
\end{equation*}
The purity equals one if and only if $\rho$ describes a pure quantum state $|\phi \rangle \! \langle \phi|$. Conversely, it achieves its minimum value for the maximally mixed state: $\rho = \left(\frac{1}{2}\mathbb{I} \right)^{\otimes N}$ achieves $p(\rho) = 1/2^N \ll 1$. We now describe how to obtain a purity estimator based on classical shadows $ \hat{\sigma}_1,\ldots, \hat{\sigma}_M$ that arise from measuring $N$ single-qubit SIC POVMs on (independent copies of) $\rho$. By construction, each $\hat{\sigma}_m$ is a Monte Carlo estimator of $\rho$. Indeed, Eq.~\eqref{eq:shadow-expectation} asserts $\mathbb{E} \left[ \hat{\sigma}_m\right]=\rho$ for all $1 \leq m \leq M$. What is more, distinct Monte Carlo estimators $\hat{\sigma}_m$ and $\hat{\sigma}_{m'}$ with $m \neq m'$ are statistically independent. The expectation over statistically independent random matrices factorizes. This ensures that the trace of the product of any two distinct classical shadows reproduces the purity in expectation:
\begin{equation*}
\mathrm{tr} \left( \hat{\sigma}_m \hat{\sigma}_{m'} \right) \quad \text{obeys} \quad \mathbb{E} \left[ \mathrm{tr} \left( \hat{\sigma}_m \hat{\sigma}_{m'} \right) \right]
= \mathrm{tr} \left( \mathbb{E} \left[ \hat{\sigma}_m \right] \mathbb{E} \left[ \hat{\sigma}_{m'} \right] \right) = \mathrm{tr} \left( \rho \; \rho \right) = p(\rho),
\end{equation*}
whenever $m \neq m'$. To boost convergence to this desired expectation, we can form the empirical average over all distinct pairs:
\begin{equation}
\hat{p} = \frac{1}{M(M-1)} \sum_{m \neq m'} \mathrm{tr} \left( \hat{\sigma}_m \hat{\sigma}_{m'}\right) = \binom{M}{2}^{-1} \sum_{m < m'} \mathrm{tr} \left( \hat{\sigma}_m \hat{\sigma}_{m'}\right).
\label{eq:purity-estimator}
\end{equation}
This formula describes an empirical average of $\binom{M}{2}$ random variables with the correct expectation value $p(\rho)$. This, in turn, ensures $\mathbb{E} \left[ \hat{p}\right] = p_2 (\rho)$. 
However, in contrast to before, the individual random variables are not necessarily statistically independent. The first two terms $\mathrm{tr} \left( \hat{\sigma}_1 \hat{\sigma}_2 \right)$ and $\mathrm{tr} \left( \hat{\sigma}_1 \hat{\sigma}_3\right)$, for instance, both depend on $\hat{\sigma}_1$. This prevents us from re-using exponential concentration inequalities, like the Bernstein inequality, to establish rapid convergence to this desired expectation value.
More general, albeit weaker, concentration arguments still apply. Chebyshev's inequality, for instance, implies
\begin{equation}
\mathrm{Pr} \left[ \left| \hat{p} - p (\rho) \right| \geq \epsilon \right] = \mathrm{Pr} \left[ \left| \hat{p} - \mathbb{E} \left[ \hat{p} \right] \right| \geq \epsilon \right] \leq \frac{\mathrm{Var} \left[ \hat{p} \right]}{\epsilon^2} \quad \text{for any $\epsilon >0$.} \label{eq:chebyshev}
\end{equation}
In words: the probability of an $\epsilon$-deviation (or larger) is bounded by the variance $\mathrm{Var} \left[ \hat{p}\right]$ of our estimator divided by $\epsilon^2$. This variance can be decomposed into individual contributions:
\begin{align*}
\mathrm{Var} \left[ \hat{p} \right]
=& \mathbb{E} \left[ \hat{p}^2 \right] - \mathbb{E} \left[ \hat{p}\right]^2 = \mathbb{E} \left[ \hat{p}^2 \right] - \mathrm{tr} \left( \rho^2 \right)^2 \\
=& \binom{M}{2}^{-2} \sum_{m_1 <m_1'} \sum_{m_2 <m_2'} \left( \mathrm{tr} \left( \hat{\sigma}_{m_1} \hat{\sigma}_{m_1'} \right) \mathrm{tr} \left( \hat{\sigma}_{m_2} \hat{\sigma}_{m'_2}\right) - \mathrm{tr} \left( \rho^2 \right) \mathrm{tr} \left( \rho^2 \right) \right)
\end{align*}
Now, note that $\mathbb{E} \left[ \hat{\sigma}_{m_1} \right] = \mathbb{E} \left[ \hat{\sigma}_{m_1'} \right] = \mathbb{E} \left[ \hat{\sigma}_{m_2}\right] = \mathbb{E} \left[ \hat{\sigma}_{m_2'} \right] = \rho$ implies that these summands vanish unless either two or all four summation indices coincide. A careful case-by-case analysis yields
\begin{align}
\mathrm{Var} \left[ \hat{p}\right]
=& \binom{M}{2}^{-1} 2 (M-2) \mathrm{Var} \left[ \mathrm{tr} \left( \rho \hat{\sigma} \right) \right] + \binom{M}{2}^{-1} \mathrm{Var} \left[ \mathrm{tr} \left( \hat{\sigma} \hat{\sigma}'\right) \right] \nonumber \\
=& \frac{4(M-2)}{M(M-1)}  \mathrm{Var} \left[ \mathrm{tr} \left( \rho \hat{\sigma} \right) \right]  + \frac{2}{M(M-1)} \mathrm{Var} \left[ \mathrm{tr} \left( \hat{\sigma} \hat{\sigma}'\right) \right]
\label{eq:variance-aux1}
\end{align}
and we refer to \cite[Supplemental material]{elben2020} for details. Here, $\rho$ is the underlying state and $\hat{\sigma},\hat{\sigma}'$ denote independent instances of a classical shadow approximation. This reformulation contains two variance terms that depend on one (linear contribution) and two independent classical shadows (quadratic contribution), respectively. We can use Lemma~\ref{lem:linear-variance} to control the first term. Set $O=\rho$ to conclude
\begin{equation}
\mathrm{Var} \left[ \mathrm{tr} \left( \rho \hat{\rho} \right) \right] \leq 3^N \mathrm{tr} \left( \rho^2 \right) \leq 3^N , \label{eq:linear-aux}
\end{equation}
because $\mathrm{tr}(\rho^2) \leq 1$ for any underlying quantum state. Bounding the quadratic variance term requires more work. The following statement is a consequence of the geometric structure of SIC POVM measurements and substitutes existing arguments which rely on 3-design properties which do not apply here.

\begin{lemma}
Let $\hat{\sigma},\hat{\sigma}'$ be independent classical shadows of an underlying $N$-qubit state $\rho$. Then,
\begin{equation*}
\mathrm{Var} \left[ \mathrm{tr} \left( \hat{\sigma} \hat{\sigma}'\right) \right] = \mathbb{E} \left[ \mathrm{tr} \left( \hat{\sigma} \hat{\sigma}'\right)^2 \right] - \mathbb{E} \left[ \mathrm{tr} \left( \hat{\sigma} \hat{\sigma}'\right) \right]^2 \leq 9^N.
\end{equation*}
\end{lemma}

We provide a detailed argument in Sec.~\ref{sub:quadratic-variance} below. For now, we insert both bounds into Eq.~\eqref{eq:variance-aux1} to obtain
\begin{align*}
\mathrm{Var} \left[ \hat{p} \right] \leq & \frac{4(M-2)}{M(M-1)} 3^N + \frac{2}{M(M-1)} 9^N
\leq  \frac{4 \times 3^N}{M-1} + \left(\frac{\sqrt{2} \times 3^N}{M-1}\right)^2.
\end{align*}
This variance bound diminishes as $M$ increases. For fixed $\epsilon \in (0,1)$ and $\delta  \in (0,1)$, a measurement budget of $M \geq \left( 5 \times 3^N+1 \right)/(\delta \epsilon^2)$ ensures $\mathrm{Var} \left[ \hat{p} \right] \leq \frac{4}{5} \epsilon^2 \delta + \frac{2}{25} \epsilon^4 \delta^2 < \epsilon^2 \delta$.
We can insert this implication into the Chebyshev bound~\eqref{eq:chebyshev} to obtain a rigorous convergence guarantee for purity estimation:
\begin{equation*}
M \geq \left( 5 \times 3^N+1 \right)/(\delta \epsilon^2) \quad \text{ensures} \quad \mathrm{Pr} \left[ \left| \hat{p} - \mathrm{tr}(\rho^2) \right| \geq \epsilon \right] \leq \delta.
\end{equation*}
In words: with probability (at least) $1-\delta$, the purity estimator $\hat{p}$ is $\epsilon$-close to the true purity. 
Again, the required measurement budget $M$ scales exponentially in the number of qubits involved. For global purities, this exponentially increasing measurement demand quickly becomes prohibitively expensive --- a situation that cannot be avoided due to recent fundamental lower bounds~\cite{Chen2021}. 
However, once more, the situation changes if we consider \emph{subsystem purities} instead. Let $\mathsf{K} \subseteq \left[N \right]$ be a subsystem comprised of $|\mathsf{K}|$ qubits. 
The associated density matrix is $\rho_{\mathsf{K}} = \mathrm{tr}_{\neg \mathsf{K}} (\rho)$ and we can estimate it by averaging appropriately marginalized classical shadows:
\begin{align}
\hat{p}_{\mathsf{K}} = \binom{M}{2} \sum_{m \neq m'} \mathrm{tr} \left( \mathrm{tr}_{\neg \mathsf{K}}\left( \hat{\sigma}_m \right) \mathrm{tr}_{\neg \mathsf{K}} \left( \hat{\sigma}_m'\right) \right)
\quad \text{obeys} \quad \mathbb{E} \left[ \hat{p}_{\mathsf{K}}\right] = \mathrm{tr} \left( \rho_{\mathsf{K}}^2 \right) = p \left( \rho_{\mathsf{K}}\right).
\label{eq:reduced-purity-estimator}
\end{align}
Importantly, this estimation process now only depends on the $|\mathsf{K}| < N$ qubits involved, such that
\begin{equation*}
M \geq \left( 5 \times 3^{|\mathsf{K}|}+1 \right)/(\delta \epsilon^2) \quad \text{ensures} \quad \mathrm{Pr} \left[ \left| \hat{p}_{\mathsf{K}} - p \left( \rho_{\mathsf{K}}\right) \right| \geq \epsilon \right] \leq \delta.
\end{equation*}
This scaling is much more favorable, especially for small subsystems ($|\mathsf{K}| \ll N$). Similar to linear observable estimation, we can use this assertion to predict many subsystem purities based on the same classical shadows. A union bound argument, similar to the proof of Theorem~\ref{thm:linear-estimation} above, readily implies the following statement.

\begin{theorem} \label{thm:purity-estimation}
Suppose we are interested in predicting $L$ subsystem purities $p \left( \rho_{\mathsf{K}_l}\right)$ of an unknown $N$-qubit state $\rho$. Let $K = \max_{1 \leq l \leq L} |\mathsf{K}_l|$ be the largest subsystem size involved and set $\epsilon, \delta \in (0,1)$. Then,
\begin{equation}
M \geq 6 L3^{K} /(\epsilon^2 \delta)
\label{eq:reduced-purity-measurement-budget}
\end{equation}
$N$-qubit SIC POVM measurements on (independent copies of) $\rho$ are likely to $\epsilon$-approximate all subsystem purities simultaneously. More precisely, the resulting subsystem purity estimators $\hat{p}_{\mathsf{K}}$ defined in Eq.~\eqref{eq:reduced-purity-estimator} obey
\begin{equation*}
\max_{1 \leq l \leq L} \left| \hat{p}_{\mathsf{K}_l} - \mathrm{tr} \left( \rho_{\mathsf{K}_l}^2 \right) \right| \leq \epsilon \quad \text{with probability (at least) $1-\delta$}
\end{equation*}
\end{theorem}

The dependence on subsystem size $K$ and  accuracy $\epsilon$ is virtually identical to convergence guarantees for linear observable prediction, see Theorem~\ref{thm:linear-estimation}. However, the dependence on the number of subsystem purities $L$and the inverse confidence $1\delta$, now enter linearly, not logarithmically. This is a consequence of the fact that the individual contributions to $\hat{p}_{\mathsf{K}_l}$ are not statistically independent. In turn, we had to resort to Chebyshev's inequality instead of stronger exponential tail bounds like the Bernstein inequality. 

It is possible to obtain a scaling proportional to $\log (2L/\delta)$ by using a more sophisticated estimation procedure known as \emph{median of means estimation}, see e.g.\ \cite{Huang2020} for details. Practical tests with real data do, however, suggest that median of means estimation actually reduces the approximation quality overall \cite{elben2020}. This is not a contradiction, because statements like Theorem~\ref{thm:purity-estimation} are conservative mathematical statements about the worst-case rate of convergence. In practical applications, convergence can --- and usually does --- set in much earlier.

\subsection{Convergence for higher order polynomials and entanglement detection (outlook)} \label{app:entanglement-probing}

Quadratic estimation with classical shadows readily extends to higher-order polynomials. Such higher-order polynomials can be used, for instance, to probe entanglement in mixed states~\cite{elben2020,Neven2021}. This is important, because quadratic entanglement conditions --- like subsystem R\'enyi entropies (purities) --- only apply to global states which are reasonably pure ($\mathrm{tr}(\rho^2) \approx 1$). To see this, consider the maximally mixed state $\tau = (1/2 \;\mathbb{I})^{\otimes N}$ on $N$ qubits. This state is certainly not entangled, but nonetheless
\begin{align*}
R_2 (\tau) = -\log_2 \left(\mathrm{tr} \left( \rho_{\mathsf{K}}^2 \right)\right) =-\log_2 \left( \mathrm{tr} \left( (1/2\;\mathbb{I})^{\otimes |\mathsf{K}|}\right) \right)=-\log_2 \left( 2^{-|\mathsf{K}|} \right) = |\mathsf{K}| \quad \text{for all subsystems $\mathsf{K} \subseteq N$}.
\end{align*}
In words, second R\'enyi entropy is maximal for all subsystems simultaneously. This, however, is not a consequence of entanglement, but a trivial consequence of the fact that the state is very (maximally) mixed. 

Fortunately, there exist entanglement criteria that extend to (very) mixed states. Chief among them is the PPT-criterion~\cite{Peres1996,Horodecki1996,Horodecki2009}. Let $\rho$ be an $N$-qubit quantum state and let $(\mathsf{A},\bar{\mathsf{A}})$ be a bipartition of the qubits into two disjoint sets. Then, $\rho$ is entangled (across the bipartition) if the partial transpose density matrix is not positive semidefinite (i.e.\ it has negative eigenvalues):
\begin{equation}
\rho^{T_{\mathsf{A}}} \not \geq 0 \quad \text{implies $\rho$ is entangled across the bipartition}. \label{eq:PPT}
\end{equation}
The partial transpose is defined by transposing tensor factors belonging to subsystem $\mathsf{A}$, i.e. $(\rho_1 \otimes \cdots \otimes \rho_N)^{T_{\mathsf{A}}}=\bigotimes_{a \in \mathsf{A}} \rho_a^T \bigotimes_{\bar{a} \in \bar{\mathsf{A}}} \rho_{\bar{a}}$ and linearly extended to all $N$-qubit density matrices. A quick sanity check confirms that the maximally mixed state doesn't pass the PPT condition ($\mathbb{I}^T=\mathbb{I}$):
\begin{align*}
\tau^{T_{\mathsf{A}}}= \bigotimes_{a \in \mathsf{A}} (1/2 \; \mathbb{I})^T \bigotimes_{\bar{a} \in \bar{\mathsf{A}}} (1/2 \; \mathbb{I})^T = \bigotimes_{n \in [N]}(1/2\; \mathbb{I}) = \tau \geq 0.
\end{align*}
Very entangled states, like the 2-qubit Bell state $|\Omega \rangle = 1/\sqrt{2}(|00 \rangle + |11 \rangle )$ do, in contrast, have partial transposes with negative eigenvalues:
\begin{equation*}
\left( |\Omega \rangle \! \langle \Omega|\right)^{T_1}= \left( |\Omega \rangle \! \langle \Omega | \right)^{T_2}=1/\sqrt{2}\;\mathbb{F} \in \mathbb{H}_2^{\otimes 2},
\end{equation*}
where $\mathbb{F}|x \rangle \otimes |y \rangle=|y \rangle \otimes |x \rangle$ denotes the swap operator which has one negative eigenvalue ($\lambda_{\min}(\mathbb{F})=-1$). We call a state $\rho$ with $\rho^{T_{\mathsf{A}}}\not\geq0$ a \emph{PPT-entangled state} (with respect to the bipartition $(\mathsf{A},\bar{\mathsf{A}}$).
The PPT condition is a sufficient, but not necessary, condition for entanglement. It is known that there exist states which are entangled, but nonetheless obey $\rho^{T_{\mathsf{A}}} \geq 0$~\cite{Horodecki1997}. So, it is fruitful to view the PPT criterion as a one-sided test for entanglement: if $\rho^{T_{\mathsf{A}}}\not \geq 0$, we can be sure that the state is entangled. But, $\rho^{T_A} \geq 0$ doesn't necessarily imply that the state is not entangled (i.e.\ separable).

The PPT criterion~\eqref{eq:PPT} is conceptually appealing, but it does require full and accurate knowledge of the density matrix $\rho$. This, in turn, typically requires full state tomography which quickly becomes prohibitively expensive. It is, however, possible to test consequences of $\rho^{T_{\mathsf{A}}} \geq 0$ by comparing moments of the partially transposed density matrix. The simplest consistency check is the so-called $p_3$-criterion~\cite{elben2020}:
\begin{equation}
\rho^{T_{\mathsf{A}}}\geq 0 \quad \Rightarrow \quad \mathrm{tr} \left( \left(\rho^{T_{\mathsf{A}}}\right)^3 \right) \geq \mathrm{tr}\left( \left( \rho^{T_{\mathsf{A}}}\right)^2 \right)^2 = \mathrm{tr} \left( \rho^2 \right)^2. \label{eq:p3}
\end{equation}
The final simplification follows from the fact that partial transposition preserves the purity. The contrapositive of this implication serves as a (one-sided) test for entanglement: if $\mathrm{tr} \left( \left( \rho^{T_{\mathsf{A}}}\right)^3 \right) < \mathrm{tr}(\rho^2)^2$ (for some bipartition ($\mathsf{A},\bar{\mathsf{A}}$), then the underlying state must be PPT-entangled (across this bipartition). 

Classical shadows can be used to directly estimate the trace moments involved in this test. Indeed, $\mathrm{tr}(\rho^2)$ is just the purity, while $\mathrm{tr} \left( \left( \rho^{T_{\mathsf{A}}}\right)^3 \right)$ can be rewritten as a linear function on three copies of the underlying state: 
\begin{equation*}
\mathrm{tr} \left( \left(\rho^{T_{\mathsf{A}}}\right)^3 \right) = \mathrm{tr} \left( O \rho \otimes \rho \otimes \rho \right).
\end{equation*}
We refer to \cite[Eq.~(4)]{elben2020} for a precise reformulation. Subsequently, we can approximate this function by averaging over triples of distinct (and therefore independent) classical shadows:
\begin{equation*}
\mathrm{tr} \left(O \rho \otimes \rho \otimes \rho \right) \approx \frac{1}{6}\binom{M}{3}^{-1}\sum_{m \neq m' \neq m''} \mathrm{tr} \left( O \hat{\sigma}_m \otimes \hat{\sigma}_{m'} \otimes \hat{\sigma}_{m''}\right).
\end{equation*}
The convergence analysis from above can, in principle, be extended to this form of cubic approximation. For randomized Pauli basis measurements, this has been done in the supplemental material of Ref.~\cite{elben2020}. 
We leave a parallel treatment of cubic estimation with SIC POVM shadows for future work. 
The experimental and numerical results from the present work indicate that a SIC POVM-based approach is expected to be both cheaper and easier than existing approaches based on Pauli basis measurements. 

Finally, we point out that Ref.~\cite{Neven2021} extended the intuition behind the $p_3$-criterion~\eqref{eq:p3} to a complete family of polynomial consistency checks that compare polynomials of degree $d$ with polynomials of degree $(d-1)$ and lower. This produces a hierarchy of in total $d_{\mathrm{max}}=2^{N}$ consistency checks that is complete in the sense that a state $\rho$ passes all of them if and only if $\rho^{T_{\mathsf{A}}} \geq 0$. Although polynomial estimation with classical shadows becomes more and more challenging as the degree $d$ increases, the lower levels of this hierarchy may still be attainable with (comparatively) modest experimental and postprocessing effort.

\section{Technical auxiliary results}

\subsection{Variance bounds for observable estimation} \label{sub:linear-technical}

Here, we supply the proof of Lemma~\ref{lem:linear-variance}.

\begin{lemma}[Restatement of Lemma~\ref{lem:linear-variance}]
Fix a $N$-qubit observable $O$ and let $\hat{\sigma} \in \mathbb{H}_2^{\otimes N}$ be a (SIC POVM) classical shadow as defined in Eq.~\eqref{eq:classical-shadow}.
Then,
\begin{equation*}
\mathrm{Var} \left[ \mathrm{tr} \left( O \hat{\sigma} \right) \right] \leq 3^N \mathrm{tr}\left(O^2 \right) \quad \text{for any underlying $N$-qubit state $\rho$.}
\end{equation*}

\end{lemma}

\begin{proof}
The classical shadow $\hat{\sigma}$ is constructed from performing single-qubit SIC POVM measurements on an underlying $N$-qubit quantum state $\rho$. Recall from Eq.~\eqref{eq:measurement} that each of the $4^N$ possible outcome strings $i_1,\ldots,i_N \in \left\{1,2,3,4\right\}$ occurs with probability
\begin{equation*}
\mathrm{Pr} \left[i_1 \cdots i_N | \rho \right] = 2^{-N} \langle \psi_{i_1},\ldots, \psi_{i_N}| \rho |\psi_{i_1}, \ldots,|\psi_{i_N} \rangle \leq 2^N.
\end{equation*}
In words: the probability of any particular outcome occurring is bounded by $2^{-N}$.
This allows us to bound the dominating part of the variance by
\begin{align*}
 \mathbb{E} \left[ \mathrm{tr} \left( O \hat{\sigma} \right)^2 \right]
=& \sum_{i_1,\ldots,i_N=1}^4 \mathrm{Pr} \left[ i_1 \cdots i_N| \rho \right] \mathrm{tr} \left( O \left( 3|\psi_{i_1} \rangle \! \langle \psi_{i_1}|-\mathbb{I} \right) \otimes \cdots \otimes \left( 3 |\psi_{i_N} \rangle \! \langle \psi_{i_N}| - \mathbb{I} \right) \right)^2 \\
\leq & 2^{-N} \sum_{i_1,\ldots,i_N=1}^4 \mathrm{tr} \left( O \left( 3|\psi_{i_1} \rangle \! \langle \psi_{i_1}|-\mathbb{I} \right) \otimes \cdots \otimes \left( 3 |\psi_{i_N} \rangle \! \langle \psi_{i_N}| - \mathbb{I} \right) \right)^2
\end{align*}
Next, we expand the $N$-qubit observable $O$ in terms of tensor products of single-qubit Pauli matrices $W_1,\ldots,W_N \in \left\{\mathbb{I},X,Y,Z \right\}$:
\begin{align*}
O = \sum_{W_1,\ldots,W_N} o \left(W_1,\ldots,W_N\right) W_1 \otimes \cdots \otimes W_N \quad \text{with} \quad o (W_1,\ldots,W_N) = 2^{-N} \mathrm{tr} \left( W_1 \otimes \cdots \otimes W_N\; O \right) \in \mathbb{R}.
\end{align*}
Such a decomposition into tensor products allows us to factorize the above bound into a product of single qubit contributions:
\begin{align*}
 \mathbb{E} \left[ \mathrm{tr} \left( O \hat{\sigma} \right)^2 \right]
\leq & 2^{-N} \sum_{i_1,\ldots,i_N=1}^4 \left( \sum_{W_1,\ldots,W_N} o(W_1,\ldots,W_N) \mathrm{tr} \left( (3 |\psi_{i_1} \rangle \! \langle \psi_{i_1}| - \mathbb{I}) W_1 \right) \cdots \mathrm{tr} \left( (3|\psi_{i_N} \rangle \! \langle \psi_{i_N}|-\mathbb{I}) W_N \right) \right)^2 \nonumber \\
=& \sum_{V_1,\ldots,V_N} \sum_{W_1,\ldots,W_N} o(V_1,\ldots,V_N) o(W_1,\ldots,W_N) \prod_{n=1}^N \underset{f(V_n,W_n)}{\underbrace{\left( \frac{1}{2} \sum_{i_n=1}^4 \mathrm{tr} \left( (3|\psi_{i_n} \rangle \! \langle \psi_{i_n} | - \mathbb{I} ) V_n \right) \mathrm{tr} \left( (3|\psi_{i_n} \rangle \! \langle \psi_{i_n} | - \mathbb{I} ) W_n \right) \right)}}. 
\end{align*}
These single qubit averages can be computed individually. Use the 1-design property~\ref{eq:1design} (i.e.\ $\frac{1}{2}\sum_{i_n=1}^4 \frac{1}{2} \langle \psi_{i_n} |A| \psi_{i_n} \rangle = \mathrm{tr}(A)$) to obtain
\begin{align*}
 f(V_n, W_n) =& \frac{1}{2} \sum_{i_n=1}^4 \mathrm{tr} \left( (3|\psi_{i_n} \rangle \! \langle \psi_{i_n} | - \mathbb{I} ) V_n \right) \mathrm{tr} \left( (3|\psi_{i_n} \rangle \! \langle \psi_{i_n} | - \mathbb{I} ) W_n \right) \\
 =& \frac{9}{2} \sum_{i_n=1}^4  \langle \psi_{i_n} | V_n |\psi_{i_n}\rangle \langle \psi_{i_n} | W_n |\psi_{i_n} \rangle - 3 \mathrm{tr}(V_n) \frac{1}{2} \sum_{i_n} \langle \psi_{i_n} |W_n| \psi_{i_n} \rangle
 - 3 \mathrm{tr}(W_n)\frac{1}{2} \sum_{i_n=1}^4 \langle \psi_{i_n} | V_n| \psi_{i_n} \rangle + \frac{1}{2} \sum_{i_n=1}^4 \mathrm{tr}(V_n) \mathrm{tr}(W_n) \\
 =& \frac{9}{2} \sum_{i_n=1}^4 \langle \psi_{i_n} |V_k| \psi_{i_n} \rangle \langle \psi_{i_n} | W_n |\psi_{i_n} \rangle - 4 \mathrm{tr}(V_n) \mathrm{tr}(W_n).
\end{align*}
Next, we use the 2-design property~\eqref{eq:2design} of single-qubit SIC POVMs to obtain
\begin{align*}
 f(V_n, W_n)
 =&  \frac{9}{2} \sum_{i_n=1}^4 \langle \psi_{i_n} |V_n| \psi_{i_n} \rangle \langle \psi_{i_n} | W_n |\psi_{i_k} \rangle - 4 \mathrm{tr}(V_n) \mathrm{tr}(W_n) \\
 =& 3 \left( \mathrm{tr}(V_n W_n) + \mathrm{tr}(V_n) \mathrm{tr}(W_n) \right) - 4 \mathrm{tr}(V_n) \mathrm{tr}(W_n) \\
 =& 3 \mathrm{tr} \left(V_n W_n \right) - \mathrm{tr}(V_n) \mathrm{tr}(W_n).
\end{align*}
For Pauli matrices $V_n,W_n$, this expression vanishes whenever $V_n \neq W_n$. It equals $2$ if $V_n=W_n = \mathbb{I}$ and $6$ if $V_n = W_n \neq \mathbb{I}$. In formulas,
\begin{equation*}
f(V_n, W_n) = 2 \delta (V_n, W_n) 3^{1-\delta (W_n,\mathbb{I})}.
\end{equation*}
Inserting this closed-form expression into the original expression yields
\begin{align*}
 \mathbb{E} \left[ \mathrm{tr} \left( O \hat{\sigma} \right)^2 \right] \leq & \sum_{V_1,\ldots,V_N} \sum_{W_1,\ldots,W_N} o (V_1,\ldots,V_N) o(W_1,\ldots,W_N) \prod_{n=1}^N 2 \delta (V_n,W_n) 3^{1-\delta(V_n,\mathbb{I})} \\
\leq & 3^N 2^N \sum_{W_1,\ldots,W_N} o(W_1,\ldots,W_N)^2
= 3^N \sum_{W_1,\ldots,W_N} 2^{-N} \mathrm{tr} \left( W_1 \otimes \cdots \otimes W_N\; O \right)^2 
= 3^N \mathrm{tr}(O^2),
\end{align*}
where the last equation follows from the fact that normalized $N$-qubit Pauli matrices form an orthonormal basis of $\mathbb{H}_2^{\otimes N}$ with respect to the Hilbert-Schmidt inner product (Parseval's identity). This is the advertised result.
\end{proof}

\subsection{Variance bounds for purity estimation} \label{sub:quadratic-variance}

\begin{proposition}[purity variance bound] \label{prop:quadratic-variance}
Let $\hat{\sigma},\hat{\sigma}' \in \mathbb{H}_2^{\otimes N}$ be two independent classical shadows that arise from performing single-qubit SIC POVM measurements on a $K$-qubit state $\rho$. Then,
\begin{align*}
\mathrm{Var} \left[ \mathrm{tr} \left( \hat{\sigma} \hat{\sigma}'\right) \right] = \mathbb{E} \left[ \mathrm{tr} \left( \hat{\sigma} \hat{\sigma}'\right)^2 \right] - \mathbb{E} \left[ \mathrm{tr} \left( \hat{\sigma} \hat{\sigma}'\right) \right]^2 \leq 9^N.
\end{align*}
\end{proposition}

The proof strategy behind this statement differs from existing arguments in the literature, most notably Refs.~\cite{Huang2020,elben2020,Neven2021}. These use the fact that Pauli basis measurements form a 3-design, a structural property that doesn't apply to SIC POVMs. The key idea behind this new proof technique is to notice that trace inner products of SIC POVM shadows can only assume very discrete values. Recall that
\begin{equation*}
\hat{\sigma} = \bigotimes_{n=1}^N(3 |\psi_{i_n} \rangle \! \langle \psi_{i_n}| - \mathbb{I}) \quad \text{and} \quad \hat{\sigma}'=\bigotimes_{n=1}^N(3 |\psi_{j_n} \rangle \! \langle \psi_{j_n}| - \mathbb{I}),
\end{equation*}
where $i_1,\ldots,i_N \in \left\{1,2,3,4\right\}$ and $j_1,\ldots,j_N \in \left\{1,2,3,4\right\}$ record the outcomes of each single qubit SIC POVM measurement. 
This tensor product structure then implies
\begin{align*}
\mathrm{tr} \left( \hat{\sigma} \hat{\sigma}' \right)
=& \prod_{n=1}^N\mathrm{tr} \left( (3 |\psi_{i_n} \rangle \! \langle \psi_{i_n}| - \mathbb{I}) (3 |\psi_{j_n} \rangle \! \langle \psi_{j_n}| - \mathbb{I}) \right)
= \prod_{n=1}^N \left( 9 \left| \langle \psi_{i_n} |\psi_{j_n} \rangle \right|^2 -4 \right)
\end{align*}
and, because $|\psi_{i_n}\rangle,|\psi_{j_n} \rangle \in \mathbb{C}^2$ are SIC vectors, each contribution can only assume one of two discrete values:
\begin{align}
9 \left| \langle \psi_{i_n} | \psi_{j_n} \rangle\right|^2-4 =
\begin{cases}
+5 & \text{if $i_n = j_n$}, \\
-1 & \text{else if $i_n \neq j_n$}.
\end{cases}
\label{eq:discrete-values}
\end{align}
So, the magnitude of $\mathrm{tr}(\hat{\sigma} \hat{\sigma}')$ scales exponentially in the number of coincidental measurement outcomes ($i_n = j_n$). 
This observation can be used to control the variance of this trace inner product. We first illustrate this for $N=2$ qubits, which is enough to convey the main gist. The proof of Proposition~\ref{prop:quadratic-variance} is then a straightforward, yet somewhat technical, generalization to an arbitrary number of qubits. 

In the 2-qubit case, $\mathrm{tr} \left( \hat{\sigma} \hat{\sigma}'\right)^2$ can only assume 3 values: $25^2$ if all single-qubit outcomes coincide, $25$ if exactly one single-qubit outcome coincides and $1$ if no outcomes coincide. In formulas,
\begin{align}
\mathrm{tr}\left( \hat{\sigma} \hat{\sigma}'\right)^2 = 25^2 \mathbf{1} \left\{ i_1 = j_1 \wedge i_2 = j_2 \right\} +25 \mathbf{1} \left\{i_1 = j_1 \wedge i_2 \neq j_2 \right\} +25 \mathbf{1} \left\{ i_1 \neq j_1 \wedge i_2 = j_2 \right\}  \mathbf{1} \left\{ i_1 \neq j_1 \wedge i_2 \neq j_2 \right\}, \label{eq:quadratic-variance-aux1}
\end{align}
where $\mathbf{1} \left\{ E \right\}$ denotes the indicator function of the event $E$. Next, we re-express these indicator functions in terms of simpler ones:
\begin{align*}
\mathbf{1} \left\{ i_1 = j_1 \wedge i_2 \neq j_2 \right\} =& \mathbf{1} \left\{i_1 = j_1 \right\} - \mathbf{1} \left\{ i_1 = j_1 \wedge i_2 = j_2 \right\}, \\
\mathbf{1} \left\{ i_1 \neq j_2 \wedge i_2 = j_2 \right\} =& \mathbf{1} \left\{i_2 = j_2 \right\} - \mathbf{1} \left\{ i_1 = j_1 \wedge i_2 = j_2 \right\}, \\
\mathbf{1} \left\{ i_1 \neq j_2 \wedge i_2 \neq j_2 \right\} =& 1 - \mathbf{1} \left\{ i_1 = j_1 \right\} - \mathbf{1} \left\{i_2 = j_2 \right\} + \mathbf{1} \left\{i_1 = j_1 \wedge i_2 = j_2 \right\}.
\end{align*}
Inserting these reformulations into Eq.~\eqref{eq:quadratic-variance-aux1} and rearranging terms yields
\begin{align*}
\mathrm{tr}\left( \hat{\sigma} \hat{\sigma}'\right)^2 
=& \left( 25^2 - 2 \times 25 +1 \right) \mathbf{1} \left\{ i_1 = j_1 \wedge i_2 = j_2 \right\} + \left( 25 - 1 \right) \mathbf{1} \left\{ i_1 = j_1 \right\} + \left( 25-1 \right) \mathbf{1} \left\{i_2=j_2 \right\} + 1 \\
=& \left(25-1 \right)^2 \mathbf{1} \left\{ i_1 = j_1 \wedge i_2 = j_2 \right\} + (25-1) \mathbf{1} \left\{ i_1 = j_1 \right\} + (25-1)^1 \mathbf{1} \left\{ i_2 = j_2 \right\} + 1\\
=& 8^2\times  3^2 \mathbf{1} \left\{ i_1 = j_1 \wedge i_2 = j_2 \right\} + 8 \times 3  \mathbf{1} \left\{ i_1 = j_1 \right\} + 8 \times 3 \mathbf{1} \left\{ i_2 = j_2 \right\} +1,
\end{align*}
where we have used $(25-1)=24=8 \times 3$. Now, we are ready to take expectation values. Recall that taking the expectation of an indicator function produces the probability of the associated event:
\begin{align}
\mathbb{E} \left[ \mathrm{tr}\left( \hat{\sigma} \hat{\sigma}'\right)^2 \right]
=& 8^2 \times 3^2 \mathrm{Pr} \left[ i_1 = j_1 \wedge i_2 = j_2 \right] + 8 \times 3 \mathrm{Pr} \left[ i_1 = j_1 \right] + 8 \times 3 \mathrm{Pr} \left[i_2=j_2 \right] + 1. \label{eq:quadratic-variance-aux2}
\end{align}
These probabilities for coincidental measurement outcomes can be computed explicitly. 
This is the content of the following auxiliary result.

\begin{lemma} \label{lem:coincidence}
Suppose that we perform two $N$-qubit SIC POVM measurements on (distinct copies of) a quantum state $\rho$ and let $\mathsf{K} \subseteq \left[N\right] = \left\{1,\ldots,N\right\}$ be a subset of $K=|\mathsf{K}|$ qubits. 
Then, the probability that the obtained measurement outcomes are equal ($i_k=j_k$) for all $k \in \mathsf{K}$ obeys
\begin{equation*}
\mathrm{Pr} \left[ \bigwedge_{k \in \mathsf{K}} \left\{i_k = j_k \right\}  \right] = \mathrm{tr} \left( \rho_{\mathsf{K}} \mathcal{D}_{1/3}^{\otimes K} (\rho_{\mathsf{K}}) \right) \leq 3^{-K},
\end{equation*}
where $\rho_{\mathsf{K}} = \mathrm{tr}_{\neg \mathsf{K}}(\rho)$ is the reduced density matrix supported on the relevant qubit subset and each $\mathcal{D}_{1/3}$ is a single-qubit depolarizing channel.
\end{lemma}

The proof follows from exploiting the fact that the two SIC POVM measurements are statistically independent, as well as the 2-design property of SIC POVMs. We defer it to the end of this section.
For the task at hand, Lemma~\ref{lem:coincidence} bounds all remaining probabilities in Eq.~\eqref{eq:quadratic-variance-aux2}. Doing so, conveniently cancels the existing powers of $3$ and produces the bound advertised in Proposition~\ref{prop:quadratic-variance} for $K=2$ qubits:
\begin{align*}
\mathbb{E} \left[ \mathrm{tr}\left( \hat{\sigma} \hat{\sigma}'\right)^2 \right]
=&  8^2 \times 3^2 \mathrm{Pr} \left[ i_1 = j_1 \wedge i_2 = j_2 \right] + 8 \times 3 \mathrm{Pr} \left[ i_1 = j_1 \right] + 8 \times 3 \mathrm{Pr} \left[i_2=j_2 \right] + 1 \\
\leq & 8^2 + 2 \times 8 + 1 = \left( 8+1\right)^2 = 9^2
\end{align*}
This argument can be readily extended to an arbitrary number of qubits. 

\begin{proof}[Proof of Proposition~\ref{prop:quadratic-variance}]
The trace inner product between two $N$-qubit SIC POVM shadows can only assume discrete values. Indeed, Eq.~\eqref{eq:discrete-values} states that $\mathrm{tr} \left( \hat{\sigma} \hat{\sigma}'\right) = \pm 5^{c}$, where $c$ is the number of coincidental measurement outcomes. We can use indicator functions to single out all possibilities for coincidences and multiplying them with the correct scaling factor provides a closed-form expression of the trace inner product in terms of measurement outcomes alone. In the following, we will use $\mathsf{K} \subseteq \left[N\right]$ to denote a subset of coincidental indices. The complementary set (where indices mustn't coincide) will be denoted by $\bar{\mathsf{K}}=\left[N\right] \setminus \mathsf{K}$. For the squared trace inner product we then obtain
\begin{align*}
\mathrm{tr} \left( \hat{\sigma} \hat{\sigma}'\right)^2
=& \sum_{\mathsf{K} \subseteq \left[N\right]} 25^{|\mathsf{K}|}  \mathbf{1} \Big( \bigwedge_{k \in \mathsf{K}} \left\{ i_k = j_k \right\} \bigwedge_{\bar{k} \in \bar{\mathsf{K}}} \left\{i_{\bar{k}}\neq j_{\bar{k}}\right\}\Big) \\
=& \sum_{\mathsf{K} \subseteq \left[N \right]} 25^{|\mathsf{K}|} \left( \sum_{\mathsf{T} \subseteq \bar{\mathsf{K}}} (-1)^{|\mathsf{T}|} \mathbf{1} \Big( \bigwedge_{u \in \mathsf{K} \cup \mathsf{T}} \left\{i_u = j_u \right\} \Big) \right) \\
=& \sum_{\mathsf{U} \subseteq \left[N\right]} \left( \sum_{\mathsf{T} \subseteq \mathsf{U}} (-1)^{|\mathsf{T}|} 25^{|\mathsf{U}|-|\mathsf{T}|} \right) \mathbf{1} \Big( \bigwedge_{u \in \mathsf{U}} \left\{ i_u = j_u \right\} \Big).
\end{align*}
In the second line, we have re-expressed conditions for non-coincidence ($\{i_{\bar{k}} \neq j_{\bar{k}}$) as linear combinations of coincidences on larger subsystems ($\mathsf{K} \cup \mathsf{T}$ with $\mathsf{T} \subseteq \bar{\mathsf{K}}$). The last line follows from introducing the union $\mathsf{U}=\mathsf{K} \cup \mathsf{T}$ and rewriting $\mathsf{K}$ as $\mathsf{K} \setminus \mathsf{T}$. The inner sum over subsets $\mathsf{T} \subseteq \mathsf{U}$ has no effect on the indicator function. The size of such sets ranges from $T=|\mathsf{T}|=0$ up to $T=|\mathsf{T}|=|\mathsf{U}|$ and for each $T$, there are $\binom{|\mathsf{U}|}{T}$ subsets of that size. For a fixed set $\mathsf{U}$, we therefore obtain
\begin{equation*}
\sum_{\mathsf{T} \subseteq \mathsf{U}} (-1)^{|\mathsf{T}|} 25^{|\mathsf{U}|-|\mathsf{T}|} = \sum_{T=0}^{|\mathsf{U}|} \binom{|\mathsf{U}|}{T} (-1)^t 25^{|\mathsf{U}|-T} = \left(25-1\right)^{|\mathsf{U}|} = 8^{|\mathsf{U}|} \times 3^{|\mathsf{U}|},
\end{equation*}
which considerably simplifies the entire function. Taking the expectation now produces
\begin{align*}
\mathbb{E} \left[ \mathrm{tr} \left( \hat{\sigma} \hat{\sigma}'\right) \right]
=& \sum_{\mathsf{U} \subseteq \left[N\right]} 8^{|\mathsf{U}|} \times 3^{|\mathsf{U}|} \mathbb{E} \left[ \mathbf{1} \Big( \bigwedge_{u \in \mathsf{U}} \left\{ i_u = j_u \right\} \Big) \right] 
=  \sum_{\mathsf{U} \subseteq \left[N\right]} 8^{|\mathsf{U}|} \times 3^{|\mathsf{U}|} \mathrm{Pr} \left[ \bigwedge_{u \in \mathsf{U}} \left\{ i_u = j_u \right\} \right]
\end{align*}
and we can use Lemma~\ref{lem:coincidence} to complete the argument. Indeed, $\mathrm{Pr} \left[ \bigwedge_{u \in \mathsf{U}} \left\{i_u = j_u \right\} \right] \leq 3^{-|\mathsf{U}|}$ ensures
\begin{align*}
\mathbb{E} \left[ \mathrm{tr} \left( \hat{\sigma} \hat{\sigma}'\right) \right] \leq \sum_{\mathsf{U} \subseteq [N]} 8^{|\mathsf{U}|}
= \sum_{U=0}^N \binom{N}{U} 8^U = (8+1)^N = 9^N.
\end{align*}
This is the advertised bound for the variance of $N$-qubit purity estimators.
\end{proof}

Finally, we provide the proof for the bound on coincidental SIC POVM measurement outcomes.

\begin{proof}[Proof of Lemma~\ref{lem:coincidence}]
For simplicity, we assume that the subset $\mathsf{K}=\left\{1,\ldots,K\right\} \subseteq \left[N\right]$ encompasses the first $K=|\mathsf{K}|$ qubits. The general case works analogously, but notation becomes somewhat cumbersome.
We perform two independent single-qubit SIC POVM measurements on (two copies of) a $N$-qubit quantum state $\rho$. The probability of getting $K=|\mathsf{K}|$ particular outcomes only depends on the reduced density matrix $\rho_{\mathsf{K}}=\mathrm{tr}_{\neg \mathsf{K}}(\rho)$ of the relevant qubit subset:
\begin{equation*}
\mathrm{Pr} \left[ i_{1} \cdots \ldots,i_{K} | \rho_{\mathsf{K}} \right] 
= 2^{-K} \langle \psi_{i_1},\ldots,\psi_{i_K} | \rho_{\mathsf{K}} |\psi_{i_1},\ldots,\psi_{i_K} \rangle \quad \text{for $i_1,\ldots,i_K \in \left\{1,2,3,4\right\}$.}
\end{equation*}
This observation allows us to rewrite the probability for $K$ coincidental measurement outcomes as
\begin{align*}
\mathrm{Pr} \left[ \bigwedge_{k=1}^K \left\{i_k = j_k \right\}  \right]
=& \sum_{i_1=1}^4 \cdots \sum_{i_K=1}^4 \mathrm{Pr} \left[ i_1 \cdots i_K | \rho_{\mathsf{K}} \right]^2
= \frac{1}{4^K} \sum_{i_1=1}^4 \cdots \sum_{i_K=1}^4 \langle \psi_{i_1},\ldots,\psi_{i_K} | \rho_{\mathsf{K}} |\psi_{i_1},\ldots,\psi_{i_K} \rangle^2
\end{align*}
At this point it is helpful to decompose (one) $\rho_{\mathsf{K}}$ into a linear combination of tensor products, e.g. $\rho_{\mathsf{K}} = \sum_{W_1,\ldots,W_K} r (W_1,\ldots,W_K) W_1 \otimes \cdots \otimes W_K$. 
Doing so allows us to rewrite 
\begin{align*}
\mathrm{Pr} \left[ \bigwedge_{k=1}^K \left\{i_k = j_k \right\}  \right]
=& 2^{-K} \mathrm{tr} \left( \rho_{\mathsf{K}} \sum_{W_1,\ldots,W_K} r(W_1,\ldots,W_K) \left( \frac{1}{2} \sum_{i_1=1}^4 |\psi_{i_1} \rangle \! \langle \psi_{i_1}| \langle \psi_{i_1} | W_1 |\psi_{i_1} \rangle \right) \otimes \cdots \otimes \left( \frac{1}{2} \sum_{i_K=1}^4 | \psi_{i_K} \rangle \! \langle \psi_{i_K}| \langle \psi_{i_K} | W_K |\psi_{i_K} \rangle \right) \right) \\
=& 2^{-K} \mathrm{tr} \left( \rho_{\mathsf{K}} \sum_{W_1,\ldots,W_K} r(W_1,\ldots,W_K) \mathcal{D}_{1/3} (W_1) \otimes \cdots \otimes \mathcal{D}_{1/3}(W_K) \right), \\
=& 2^{-K} \mathrm{tr} \left( \rho_{\mathsf{K}} \mathcal{D}_{1/3}^{\otimes K} (\rho_{\mathsf{K}}) \right),
\end{align*}
as advertised.
Here, we have used the 2-design property~\eqref{eq:2design} of SIC POVMs, more precisely:$\sum_{i=1}^4\frac{1}{2} |\psi_i \rangle \! \langle \psi_i| \langle \psi_i| A |\psi_i \rangle = \mathcal{D}_{1/3}(A)$. 
To get the state-independent upper bound, we note that each depolarizing channel is a linear combination between the identity channel ($\mathcal{I}(A)=A$) and the projection onto the identity matrix ($\mathcal{T}(A)=\mathrm{tr}(A) \mathbb{I})$: $\mathcal{D}_{1/3}= \frac{1}{3} \left( \mathcal{I} + \mathcal{T} \right)$. We also drop the subscript $\mathsf{K}$ in $\rho$ to declutter notation somewhat: $\rho_{\mathsf{K}} \mapsto \rho$. Then,
\begin{align*}
2^{-K} \mathrm{tr} \left( \rho \mathcal{D}_{1/3}^{\otimes K}(\rho) \right)
=& 2^{-K} \mathrm{tr} \left( \rho 3^{-K} \left( \mathcal{I} + \mathcal{T}\right)^{\otimes K} (\rho) \right) 
= \frac{1}{3^K} 2^{-K}\sum_{T \subseteq \left\{1,\ldots,s\right\}} \mathrm{tr} \left( \rho \mathcal{T}_T \otimes \mathcal{I}_{\bar{T}} (\rho) \right) \\
=& 3^{-K} \left(2^{-K} \sum_{\mathsf{T} \subseteq \left\{1,\ldots,K\right\}} \mathrm{tr}\left( \rho \; \rho_{\mathsf{T}} \otimes \mathbb{I}_{\bar{\mathsf{T}}} \right) \right) = 3^{-K} \left( 2^{-K} \sum_{\mathsf{T} \subseteq \left\{1,\ldots,s\right\}} \mathrm{tr} \left( \rho_{\mathsf{T}}^2 \right) \right).
\end{align*}
The remaining bracket averages over all subsystem purities $\mathrm{tr}(\rho_T^2)$. Each of them obeys $\mathrm{tr}(\rho_{\mathsf{T}}^2) \leq 1$ and there are a total of $2^K$ of them (a finite set of size $K$ has $2^K$ subsets). Upper bounding each of them by $1$ produces the advertised display:
\begin{align*}
2^{-K} \mathrm{tr} \left( \rho \mathcal{D}_{1/3}^{\otimes K}(\rho) \right)=  3^{-K}\left( 2^{-K} \sum_{\mathsf{T}\subseteq \left\{1,\ldots,s\right\}} \mathrm{tr} \left( \rho_{\mathsf{T}}^2 \right) \right) \leq 3^{-K}.
\end{align*}
\end{proof}

\end{document}